\def\llncs{0}
\def\fullpage{1}
\def\anonymous{0}
\def\authnote{1}
\def\notxfont{0}
\def\submission{0}

\ifnum\submission=1
\def\llncs{1}
\fi

\ifnum\llncs=1 	\documentclass[envcountsect,a4paper,runningheads,10pt]{llncs}
\else
	\documentclass[letterpaper,hmargin=1.05in,vmargin=1.05in,11pt]{article}
			\ifnum\fullpage=1
		\usepackage{fullpage}
		\fi
\fi

\usepackage[%
  colorlinks=true,
  citecolor=blue,
  pagebackref=true
]{hyperref}

\usepackage{amsmath, amsfonts, amssymb, mathtools,amscd}

\usepackage{amsthm}

\usepackage{lmodern}
\usepackage[T1]{fontenc}
\usepackage[utf8]{inputenc}

\usepackage{arydshln} 
\usepackage{url}
\usepackage{ifthen}
\usepackage{bm}
\usepackage{multirow}
\usepackage[dvips]{graphicx}
\usepackage[usenames]{color}
\usepackage{xcolor,colortbl} 
\usepackage{threeparttable}
\usepackage{comment}
\usepackage{paralist,verbatim}
\usepackage{cases}
\usepackage{booktabs}
\usepackage{braket}
\usepackage{cancel} 
\usepackage{ascmac} 
\usepackage{framed}
\usepackage{authblk}
\usepackage{pifont}
\usepackage{qcircuit}
\definecolor{darkblue}{rgb}{0,0,0.6}
\definecolor{darkgreen}{rgb}{0,0.5,0}
\definecolor{maroon}{rgb}{0.5,0.1,0.1}
\definecolor{dpurple}{rgb}{0.2,0,0.65}

\usepackage[capitalise,noabbrev]{cleveref}
\usepackage[absolute]{textpos}
\usepackage[final]{microtype}
\usepackage[absolute]{textpos}
\usepackage{everypage}
\DeclareMathAlphabet{\mathpzc}{OT1}{pzc}{m}{it}

\usepackage{algorithmic}
\usepackage{algorithm}
\usepackage{here}

\newtheoremstyle{thicktheorem}%
{\topsep}
{\topsep}
{\itshape}{}%
{\bfseries}%
{.}
{ }%
{\thmname{#1}\thmnumber{ #2}%
		\thmnote{ (#3)}%
}

\newtheoremstyle{remark}
{\topsep}
{\topsep}
	{}
	{}
	{}
	{.}
	{ }
	{\textit{\thmname{#1}}\thmnumber{ #2}
			\thmnote{ (#3)}%
	}

\ifnum\llncs=0
	\theoremstyle{thicktheorem}
	\newtheorem{theorem}{Theorem}[section]
	\newtheorem{lemma}[theorem]{Lemma}

	\newtheorem{definition}[theorem]{Definition}

	\theoremstyle{remark}
	
	\newtheorem{remark}[theorem]{Remark}

\else
\fi

\Crefname{MyClaim}{Claim}{Claims}

	\crefname{theorem}{Theorem}{Theorems}
	\crefname{assumption}{Assumption}{Assumptions}
	\crefname{construction}{Construction}{Constructions}
	\crefname{corollary}{Corollary}{Corollaries}
	\crefname{conjecture}{Conjecture}{Conjectures}
	\crefname{definition}{Definition}{Definitions}
	\crefname{exmaple}{Example}{Examples}
	\crefname{experiment}{Experiment}{Experiments}
	\crefname{counterexample}{Counterexample}{Counterexamples}
	\crefname{lemma}{Lemma}{Lemmata}
	\crefname{observation}{Observation}{Observations}
	\crefname{proposition}{Proposition}{Propositions}
	\crefname{remark}{Remark}{Remarks}
	\crefname{claim}{Claim}{Claims}
	\crefname{fact}{Fact}{Facts}
	\crefname{note}{Note}{Notes}

\ifnum\llncs=1
 \crefname{appendix}{App.}{Appendices}
 \crefname{section}{Sec.}{Sections}
\else
\fi

\ifnum\llncs=1
\renewcommand*{\backref}[1]{}
\else
	\renewcommand*{\backref}[1]{(Cited on page~#1.)}
	\ifnum\notxfont=1
	\else
		\usepackage{newtxtext}
	\fi
\fi
\ifnum\authnote=0  
\newcommand{\mor}[1]{}
\newcommand{\minki}[1]{}
\newcommand{\takashi}[1]{}

\else
\newcommand{\mor}[1]{$\ll$\textsf{\color{red} Tomoyuki: { #1}}$\gg$}
\newcommand{\takashi}[1]{$\ll$\textsf{\color{orange} Takashi: { #1}}$\gg$}
\newcommand{\minki}[1]{$\ll$\textsf{\color{darkgreen} Minki: { #1}}$\gg$}

\fi

\newcommand{\Tr}{\mathrm{Tr}}

\newcommand{\BQP}{\mathbf{BQP}}
\newcommand{\QMA}{\mathbf{QMA}}
\newcommand{\NP}{\mathbf{NP}}

\newcommand{\StateGen}{\mathsf{StateGen}}
\newcommand{\Mint}{\mathsf{Mint}}

\newcommand{\seteq}{\coloneqq}

\newcommand{\cA}{\mathcal{A}}
\newcommand{\cB}{\mathcal{B}}
\newcommand{\cC}{\mathcal{C}}
\newcommand{\cD}{\mathcal{D}}

\def\makeuppercase#1{
\expandafter\newcommand\csname tl#1\endcsname{\widetilde{#1}}
}

\def\makelowercase#1{
\expandafter\newcommand\csname tl#1\endcsname{\widetilde{#1}}
}

\newcommand{\regC}{\mathbf{C}}
\newcommand{\regR}{\mathbf{R}}
\newcommand{\regZ}{\mathbf{Z}}

\newcommand{\regB}{\mathbf{B}}
\newcommand{\regA}{\mathbf{A}}

\newcommand{\secp}{\lambda}

\newcommand{\A}{\entity{A}}

\newcommand*{\sk}{\keys{sk}}
\newcommand*{\pk}{\keys{pk}}

\newcommand{\ct}{\keys{ct}}

\newcommand*{\keys}[1]{\mathsf{#1}}

\newcommand*{\algo}[1]{\ensuremath{\mathsf{#1}}}

\newcommand*{\entity}[1]{\mathcal{#1}}

\newenvironment{boxfig}[2]{\begin{figure}[#1]\fbox{\begin{minipage}{0.97\linewidth}
                        \vspace{0.2em}
                        \makebox[0.025\linewidth]{}
                        \begin{minipage}{0.95\linewidth}
            {{
                        #2 }}
                        \end{minipage}
                        \vspace{0.2em}
                        \end{minipage}}}{\end{figure}}

\newcommand{\pprotocol}[4]{
\begin{boxfig}{h}{\footnotesize 
\centering{\textbf{#1}}
    #4
\vspace{0.2em} } \caption{\label{#3} #2}
\end{boxfig}
}

\newcommand{\protocol}[4]{
\pprotocol{#1}{#2}{#3}{#4} }

\newcommand{\bit}{\{0,1\}}

\newcommand{\Gen}{\algo{Gen}}
\newcommand{\KeyGen}{\algo{KeyGen}}
\newcommand{\SKGen}{\algo{SKGen}}
\newcommand{\PKGen}{\algo{PKGen}}

\newcommand{\Enc}{\algo{Enc}}
\newcommand{\Dec}{\algo{Dec}}

\newcommand{\Sign}{\algo{Sign}}

\newcommand{\Ver}{\algo{Ver}}

\newcommand\SKE{\algo{SKE}}

\newcommand{\QPKE}{\algo{QPKE}}
\newcommand{\QSKE}{\algo{QSKE}}

\newcommand{\negl}{{\mathsf{negl}}}

\newcommand{\poly}{{\mathrm{poly}}}

\makeatletter
\DeclareRobustCommand
  \myvdots{\vbox{\baselineskip4\p@ \lineskiplimit\z@
    \hbox{.}\hbox{.}\hbox{.}}}
\makeatother

\title{One-Wayness in Quantum Cryptography}

\ifnum\anonymous=1
\ifnum\llncs=1
\author{\empty}\institute{\empty}
\else
\author{}
\fi
\else
\ifnum\llncs=1
\author{
	Tomoyuki Morimae\inst{1} \and Takashi Yamakawa\inst{1,2}
}
\institute{
	Yukawa Institute for Theoretical Physics, Kyoto University, Kyoto, Japan \and NTT Social Informatics Laboratories, Tokyo, Japan
}
\else
\author[1]{Tomoyuki Morimae}
\author[2,3,1]{\hskip 1em Takashi Yamakawa}
\affil[1]{{\small Yukawa Institute for Theoretical Physics, Kyoto University, Kyoto, Japan}\authorcr{\small tomoyuki.morimae@yukawa.kyoto-u.ac.jp}}
\affil[2]{{\small NTT Social Informatics Laboratories, Tokyo, Japan}\authorcr{\small takashi.yamakawa.ga@hco.ntt.co.jp}}
\affil[3]{{\small NTT Research Center for Theoretical Quantum Information, Atsugi, Japan}}

\fi 
\fi

\date{}

\begin{document}

\maketitle

\begin{abstract}
The existence of one-way functions is one of the most fundamental assumptions in classical cryptography. 
In the quantum world, on the other hand, there are evidences that some cryptographic primitives can exist even if one-way functions do not exist 
[Kretschmer, TQC 2021; Morimae and Yamakawa, CRYPTO 2022; Ananth, Qian, and Yuen, CRYPTO 2022]. 
We therefore have the following important open problem in quantum cryptography: What is the most fundamental assumption in quantum cryptography? 
In this direction, [Brakerski, Canetti, and Qian, ITCS 2023] recently defined a notion called EFI pairs, which are pairs of efficiently generatable states that are statistically distinguishable
but computationally indistinguishable, and showed its equivalence with some cryptographic primitives including commitments, oblivious transfer, and general multi-party computations. However, their work focuses on decision-type primitives and does not cover search-type primitives like quantum money and digital signatures.  
In this paper, 
we study properties of one-way state generators (OWSGs), which are a quantum analogue of one-way functions proposed by Morimae and Yamakawa. 
We first revisit the definition of OWSGs and generalize it by allowing mixed output states.
Then we show the following results.
\begin{enumerate}
\item
We define a weaker version of OWSGs, which we call weak OWSGs, and show that they are equivalent to OWSGs.
It is a quantum analogue of the amplification theorem for classical weak one-way functions.
    \item 
    (Bounded-time-secure) quantum digital signatures with quantum public keys are equivalent to
    OWSGs. 
    \item
    Private-key quantum money schemes (with pure money states)
    imply OWSGs. 
    \item
    Quantum pseudo one-time pad schemes imply both OWSGs and EFI pairs. For EFI pairs, single-copy security suffices.
   \item
    We introduce an incomparable variant of OWSGs, which we call secretly-verifiable and statistically-invertible OWSGs,
    and show that
    they are equivalent to EFI pairs. 
\end{enumerate}

\end{abstract}

\ifnum\submission=1
\else
\newpage
\setcounter{tocdepth}{2}
\tableofcontents
\newpage
\fi

\section{Introduction}
One-way functions (OWFs) are functions that are easy to compute but hard to invert.
The existence of OWFs is one of the most fundamental assumptions in classical cryptography.
OWFs are equivalent to many cryptographic primitives,
such as commitments, digital signatures, pseudorandom generators (PRGs),
symmetric-key encryption (SKE), and zero-knowledge, etc. Moreover, almost all other cryptographic primitives,  
such as collision-resistant hashes, 
public-key encryption (PKE),
oblivious transfer (OT), multi-party computations (MPCs), etc., imply OWFs. 
In the quantum world, on the other hand, it seems that OWFs are not necessarily the most fundamental
element. In fact, recently, several quantum cryptographic primitives, such as commitments, (one-time secure) digital signatures,
quantum pseudo one-time pad (QPOTP)\footnote{QPOTP schemes are a one-time-secure SKE with quantum ciphertexts where the key length is shorter than the massage length. 
(For the definition, see \cref{def:OTP}.)}, and MPCs are constructed from pseudorandom states generators (PRSGs)~\cite{C:MorYam22,C:AnaQiaYue22}.
A PRSG~\cite{C:JiLiuSon18}, which is a quantum analogue of a PRG, is a QPT algorithm
that outputs a quantum state whose polynomially-many copies are computationally indistinguishable from the same number of copies of Haar random states.
Kretschmer~\cite{Kre21} showed that PRSGs exist even if $\BQP=\QMA$ (relative to a quantum oracle), which means that PRSGs (and all the above primitives that can be constructed from PRSGs)
could exist even if all quantum-secure (classical) cryptographic primitives including OWFs are broken.\footnote{If $\QMA=\BQP$, then $\NP\subseteq\BQP$. Because all quantum-secure classical cryptographic primitives are in $\NP$,
it means that they are broken by QPT algorithms.}
Kretschmer, Qian, Sinha, and Tal~\cite{STOC23:KreQiaSinTal} also showed that 1-PRSGs (which are variants of PRSGs secure against adversaries that get only a single copy of the state)
exist even if $\NP=\bf{P}$.
We therefore have the following important open problem in quantum cryptography: 
\begin{center}
    {\bf Question 1:} {\it What is the most fundamental assumption in quantum cryptography?}
\end{center}

In classical cryptography,
a pair of PPT algorithms whose output probability distributions are
statistically distinguishable but computationally indistinguishable
is known to be fundamental.
Goldreich~\cite{Gol90} showed the equivalence of such a pair to PRGs, which also means
the equivalence of such a pair to all cryptographic primitives in Minicrypt~\cite{Impagliazzo95}.
It is natural to consider its quantum analogue: a pair of QPT algorithms whose output quantum states
are statistically distinguishable but computationally indistinguishable.
In fact, such a pair was implicitly studied in quantum 
commitments~\cite{AC:Yan22}.  
In the canonical form of quantum commitments~\cite{YWLQ15}, 
computationally hiding and statistically binding quantum commitments are equivalent to such pairs.
The importance of such a pair as an independent quantum cryptograpic primitive
was pointed out in \cite{AC:Yan22,cryptoeprint:2022/1181}. 
In particular, the authors of \cite{cryptoeprint:2022/1181} explicitly defined it as {\it EFI pairs},\footnote{It stands for efficiently samplable, statistically far but 
computationally indistinguishable pairs of distributions.}
and showed that EFI pairs are implied by several quantum cryptographic primitives
such as (semi-honest) quantum OT, (semi-honest) quantum 
MPCs, and (honest-verifier) quantum computational zero-knowledge proofs.
It is therefore natural to ask the following question.
\begin{center}
    {\bf Question 2:} {\it Which other quantum cryptographic primitives imply EFI pairs?}
\end{center}

PRSGs and EFI pairs are ``decision type'' primitives, which correspond to PRGs in classical cryptography.
An example of the other type of primitives, namely, ``search type'' one in  classical cryptography, is OWFs.
Recently, a quantum analogue of OWFs, so called one-way states generators (OWSGs), are introduced~\cite{C:MorYam22}.
A OWSG is a QPT algorithm that, on input a classical bit string (key) $k$, outputs a quantum state $\ket{\phi_k}$. 
As the security, we require that it is hard to find $k'$ such that $|\langle\phi_k|\phi_{k'}\rangle|^2$ is non-negligible given polynomially many copies of $\ket{\phi_k}$. 
The authors showed that OWSGs are implied by PRSGs, and that OWSGs imply (one-time secure) quantum digital signatures with quantum public keys.
In classical cryptography, OWFs are connected to many cryptographic primitives. We are therefore interested in the following question.
\begin{center}
    {\bf Question 3:} {\it Which quantum cryptographic primitives are related to OWSGs?}
\end{center}

In classical cryptography, PRGs (i.e., a decision-type primitive) and OWFs (i.e., a search-type primitive) are
equivalent. In quantum cryptography, on the other hand, we do not know whether
OWSGs and EFI pairs (or PRSGs) are equivalent or not.
We therefore have the following open problem.
\begin{center}
    {\bf Question 4:} {\it Are OWSGs and EFI pairs (or PRSGs) equivalent?}
\end{center}

\subsection{Our Results}
The study of quantum cryptography with complexity assumptions has became active only very recently, and therefore we 
do not yet have enough knowledge to answer {\bf Question 1}.
However, as an important initial step towards the ultimate goal,
we give some answers to other questions above. 
Our results are summarized as follows.
(See also Fig.~\ref{OWSGfig}.)

\begin{enumerate}
    \item
We first revisit the definition of OWSGs. In the original definition in~\cite{C:MorYam22}, output states of OWSGs are assumed to be pure states. 
Moreover, the verification is done as follows: a bit string $k'$ from the adversary is accepted if and only if
the state $|\phi_k\rangle\langle\phi_k|$ is measured in the basis
$\{|\phi_{k'}\rangle\langle\phi_{k'}|,I-|\phi_{k'}\rangle\langle\phi_{k'}|\}$,
and the first result is obtained. (Note that in classical OWFs, the verification is implicit because it is trivial: just computing $f(x')$ for $x'$ given by the adversary, and check
whether it is equal to $f(x)$ or not. However, in the quantum case, we have to explicitly define the verification.)
In this paper, to capture more general settings, we generalize the definition of OWSGs by allowing outputs to be mixed states. 
A non-trivial issue that arises from this modification is that there is no canonical way to verify input-output pairs of OWSGs. 
To deal with this issue, we include such a verification algorithm as a part of syntax of OWSGs. See \cref{def:OWSG} for the formal definition.   

\item
We show an ``amplification theorem'' for OWSGs. That is, we define weak OWSGs (wOWSGs), which only requires the adversary's advantage to be $1-1/\poly(\secp)$ instead of $\negl(\secp)$, and show that 
a parallel repetition of wOWSGs gives OWSGs~(\cref{sec:amplification}). This is an analogue of the equivalence of weak one-way functions and (strong) one-way functions in classical 
cryptography~\cite{FOCS:Yao82a}.  

\item 
    We show that one-time-secure quantum digital signatures (QDSs) with quantum public keys 
    are equivalent to OWSGs (\cref{sec:QDS_OWSG}).\footnote{A construction of QDSs from OWSGs was already shown in \cite{C:MorYam22}, but
    in this paper, we generalize the definition of OWSGs, and we give the proof in the new definition.} 
   Moreover, we can generically upgrade one-time-secure QDSs into bounded-time-secure one (\cref{sec:QDS_q}).\footnote{We thank Or Sattath for asking if we can get (stateless) bounded-time QDSs.}
   
    \item
    We show that private-key quantum money schemes (with pure money states or with verification algorithms that satisfy some symmetry) imply OWSGs (\cref{sec:OWSGfromQmoney_pure} and \cref{sec:OWSGfromQmoney_symmetric}).
    
   \item
    We show that QPOTP schemes imply OWSGs (\cref{sec:OWSGfromQPOTP}). 
    \ifnum\submission=1
    \else
    This in particular means that IND-CPA secure quantum SKE or quantum PKE implies OWSGs (\cref{sec:QSKE}). 
    \fi
   
\item 
    We show that
    single-copy-secure QPOTP schemes imply EFI pairs (\cref{sec:EFIfromQPOTP}).
    Single-copy-security means that the adversary receives only a single copy of the quantum ciphertext. 
    \ifnum\submission=1
    \else
    This in particular means that IND-CPA secure quantum SKE or quantum PKE implies EFI pairs (\cref{sec:QSKE}). 
    \fi
    
 \item
   We introduce an incomparable variant of OWSGs, which we call secretly-verifiable and statistically-invertible OWSGs (SV-SI-OWSGs) (\cref{sec:SVOWSG_def}),  
    and show that SV-SI-OWSGs are equivalent to EFI pairs (\cref{sec:SVOWSG_EFI}).
   
\end{enumerate}

We remark that we consider the generalized definition of OWSGs with mixed state outputs by default. However, all the relationships between OWSGs and other primitives naturally extend to the pure state version if we consider the corresponding pure state variants of the primitives. 

\begin{figure}[htbp]
\begin{center}
\includegraphics[width=0.4\textwidth]{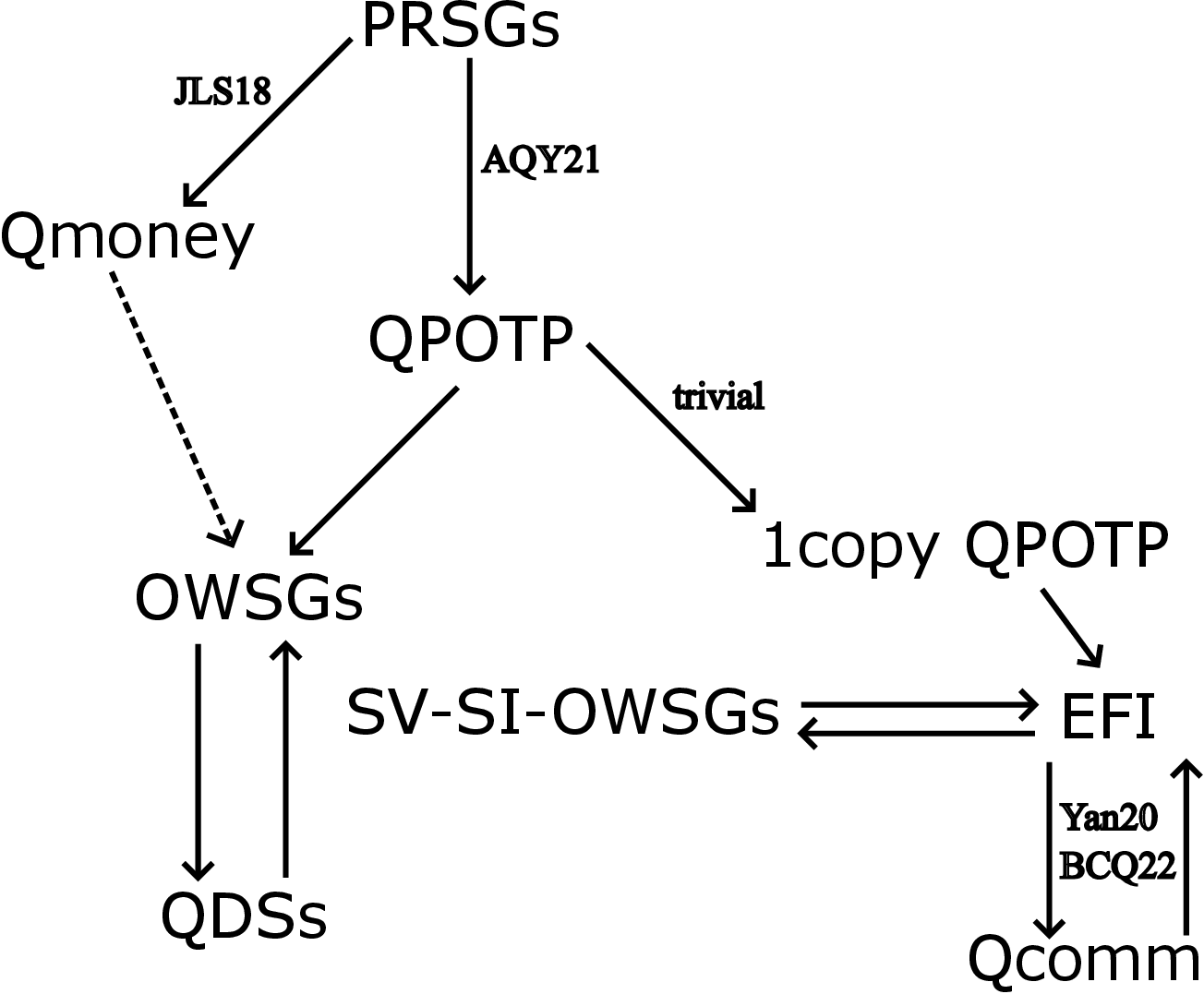}
\end{center}
\caption{
Summary of results.
The dotted line means some restrictions:
OWSGs are implied by quantum money schemes with {\it pure} money states or with {\it symmetric} verification algorithms.
}
\label{OWSGfig}
\end{figure}

\subsection{Concurrent Work}
A concurrent work by Cao and Xue~\cite{cryptoeprint:2022/1323} also studies OWSGs. 
In particular, they also show the equivalence between weak OWSGs and OWSGs. Though they consider OWSGs with pure state outputs as defined in~\cite{C:MorYam22}, it is likely that their proof extends to the mixed state version as well. 
Interestingly, the proof strategies are different in our and their works. 
Their proof is based on the classical proof of the equivalence between the weak one-way functions and one-way functions~\cite{FOCS:Yao82a,DBLP:books/cu/Goldreich2001}.  On the other hand, our proof is based on the amplification theorem for weakly verifiable puzzles~\cite{TCC:CanHalSte05}. Though their proof is simpler, an advantage of our approach is that it captures more general settings. For example, our proof also works for secretly-verifiable OWSGs~(\cref{def:SV-OWSG}). Their approach seems to rely on public verifiability in an essential manner and not applicable to secretly-verifiable OWSGs. 

Besides the equivalence between weak OWSGs and OWSGs, there is no other overlap between our and their results.

\subsection{Open Problems}
The study of ``quantum cryptography without one-way functions'' has just started, and our understanding is
very limited.
There are many open problems in this emerging field, but
we believe the following open problems are important, and some of them seem to be highly challenging.
\begin{enumerate}
    \item 
    What is the most fundamental assumption in quantum cryptography? It should be implied by many primitives, and should
    imply many primitives. 
    It should also be simple.
    Or, there is no such thing in quantum cryptography?
    \item 
    OWSGs=EFI?
    Or at least can we show OWSGs$\to$EFI\footnote{OWSGs$\to$EFI has recently been shown in \cite{cryptoeprint:2023/1620} for pure-output OWSGs.} or EFI$\to$OWSGs?
    If they are incomparable, is there any more fundamental primitive that is implied by both OWSGs and EFI pairs?
    \item
    Do OWSGs imply private-key quantum money schemes?
    An adversary who can clone a quantum money state $\psi_k$ would not necessarily be able to find the secret key $k$.
    \item 
    Do EFI pairs imply single-copy-secure PRSGs?
    \item 
    Can we construct unbounded-poly many-time secure digital signatures without one-way functions?
 \item 
    Which other primitives can be constructed without one-way functions?
    For example, how about PKE, NIZK, or proofs of quantumness?
    \item 
    ${\bf PP}\neq{\bf BQP}$ is necessary for the existence of PRSGs \cite{Kre21}. 
    What are the classical complexity assumptions necessary for the existences of other primitives, such as
    OWSGs, private-key quantum money schemes, and EFI pairs?\footnote{Recently, there have been some progresses regarding this open problem. First, it was shown that ${\bf PP}\neq{\bf BQP}$ is necessary for the existence
    of pure OWSGs~\cite{CGGHLP23}. Second, an evidence is given that single-copy PRSGs (and hence EFI pairs) could exist even if ${\bf P}={\bf ALL}$~\cite{cryptoeprint:2023/1602}.
    }
\end{enumerate}

\section{Preliminaries}

\subsection{Basic Notations}
\label{sec:basic_notations}

We use the standard notations of quantum computing and cryptography.
We use $\secp$ as the security parameter.
$[n]$ means the set $\{1,2,...,n\}$.
For any set $S$, $x\gets S$ means that an element $x$ is sampled uniformly at random from the set $S$.
$\negl$ is a negligible function, and $\poly$ is a polynomial.
PPT stands for (classical) probabilistic polynomial-time and QPT stands for quantum polynomial-time.
If we say that an adversary is QPT, it implicitly means non-uniform QPT.
A QPT unitary is a unitary operator that can be implemented in a QPT quantum circuit.

For an algorithm $A$, $y\gets A(x)$ means that the algorithm $A$ outputs $y$ on input $x$.
In particular, if $x$ and $y$ are quantum states and $A$ is a quantum algorithm, $y\gets A(x)$ means the following:
a unitary $U$ is applied on $x\otimes|0...0\rangle\langle0...0|$, and some qubits are traced out.
Then, the state of remaining qubits is $y$.
This, importantly, means that the state $y$ is {\it uniquely decided} by the state $x$.
If $A$ is a QPT algorithm, the unitary $U$ is QPT and the number of ancilla qubits 
$|0...0\rangle$ is $\poly(\secp)$.
If $x$ is a classical bit string, $y$ is a quantum state, and $A$ is a quantum algorithm,
$y\gets A(x)$ sometimes means the following:
a unitary $U_x$ that depends on $x$ is applied on $|0...0\rangle$, and some qubits are traced out.
The state of the remaining qubits is $y$.
This picture is the same as the most general one where $x$ is given as input,
but we sometime choose this picture if it is more convenient.

$\|X\|_1\coloneqq\mbox{Tr}\sqrt{X^\dagger X}$ is the trace norm.
$\mbox{Tr}_\regA(\rho_{\regA,\regB})$ means that the subsystem (register) $\regA$ of the state $\rho_{\regA,\regB}$ on
two subsystems (registers) $\regA$ and $\regB$ is traced out.
For simplicity, we sometimes write $\mbox{Tr}_{\regA,\regB}(|\psi\rangle_{\regA,\regB})$ to mean
$\mbox{Tr}_{\regA,\regB}(|\psi\rangle\langle\psi|_{\regA,\regB})$.
$I$ is the two-dimensional identity operator. For simplicity, we sometimes write $I^{\otimes n}$ as $I$ 
if the dimension is clear from the context.
For the notational simplicity, we sometimes write $|0...0\rangle$ just as $|0\rangle$,
when the number of zeros is clear from the context.
For two pure states $|\psi\rangle$ and $|\phi\rangle$,
we sometimes write $\||\psi\rangle\langle\psi|-|\phi\rangle\langle\phi|\|_1$
as
$\||\psi\rangle-|\phi\rangle\|_1$
to simplify the notation.
$F(\rho,\sigma)\coloneqq\|\sqrt{\rho}\sqrt{\sigma}\|_1^2$
is the fidelity between $\rho$ and $\sigma$.
We often use the well-known relation between the trace distance and the fidelity:
$1-\sqrt{F(\rho,\sigma)}\le\frac{1}{2}\|\rho-\sigma\|_1\le\sqrt{1-F(\rho,\sigma)}$.

\subsection{EFI Pairs}
The concept of EFI pairs was implicitly studied in \cite{AC:Yan22}, and explicitly defined in \cite{cryptoeprint:2022/1181}.
\begin{definition}[EFI pairs~\cite{cryptoeprint:2022/1181}]
\label{def:EFI}
An EFI pair is an algorithm $\StateGen(b,1^\secp)\to\rho_b$ that, on input $b\in\bit$ and the security parameter $\lambda$,
outputs a quantum state $\rho_b$ such that all of the following three conditions are satisfied.
\begin{itemize}
\item
It is a uniform QPT algorithm.
\item
$\rho_0$ and $\rho_1$ are computationally indistinguishable. In other words,
for any QPT adversary $\cA$,
$
|\Pr[1\gets\cA(1^\secp,\rho_0)]
-\Pr[1\gets\cA(1^\secp,\rho_1)]|\le\negl(\secp).
$
\item
$\rho_0$ and $\rho_1$ are statistically distinguishable,
i.e.,
$
\frac{1}{2}\|\rho_0-\rho_1\|_1\ge\frac{1}{\poly(\secp)}.
$
\end{itemize}
\end{definition}

\begin{remark}
\label{remark:EFI}
Note that in the above definition, the statistical distinguishability is defined
with only $\ge1/\poly(\secp)$ advantage. 
However, if EFI pairs with the above definition exist, EFI pairs with $\ge1-\negl(\secp)$
statistical distinguishability exist as well.
In fact, we have only to define a new $\StateGen'$ that runs $\StateGen$ $n$ times with sufficiently large $n=\poly(\secp)$, and outputs
$\rho_b^{\otimes n}$.
The $\ge1-\negl(\secp)$ statistical distinguishability for $\StateGen'$
is shown from the inequality~\cite{cryptoeprint:2022/1181},
\begin{eqnarray*}
\frac{1}{2}\|\rho^{\otimes n}-\sigma^{\otimes n}\|_1
\ge1-\exp(-n\|\rho-\sigma\|_1/4).
\end{eqnarray*}
The computational indistinguishability for $\StateGen'$ is shown by the standard hybrid argument.
\end{remark}

\subsection{Quantum Commitments}
We define canonical quantum bit commitments~\cite{AC:Yan22} as follows.  

\begin{definition}[Canonical quantum bit commitments \cite{AC:Yan22}]\label{def:canonical_com}
A canonical quantum bit commitment scheme is a family $\{Q_0(\secp),Q_1(\secp)\}_{\secp\in \mathbb{N}}$ of QPT unitaries 
on two registers $\regC$ (called the \emph{commitment} register) and $\regR$ (called the \emph{reveal} register).  
For simplicity, we often omit $\secp$ and simply write $\{Q_0,Q_1\}$ to mean $\{Q_0(\secp),Q_1(\secp)\}_{\secp\in\mathbb{N}}$. 
\end{definition}
\begin{remark}
Canonical quantum bit commitments are used as follows. In the commit phase, to commit to a bit $b\in \bit$, 
the sender generates a state $Q_b\ket{0}_{\regC,\regR}$ and sends $\regC$ to the receiver while keeping $\regR$. In the reveal phase, the sender sends $b$ and $\regR$ to the receiver. 
The receiver projects the state on $(\regC,\regR)$ onto $Q_b\ket{0}_{\regC,\regR}$, and accepts if it succeeds and otherwise rejects.  
(In other words, the receiver applies the unitary $Q_b^\dagger$ on the registers $\regC$ and $\regR$, and measure all qubits in the computational basis.
If all result are zero, accept. Otherwise, reject.)
\end{remark}


\begin{definition}[Hiding]
We say that a canonical quantum bit commitment scheme $\{Q_0,Q_1\}$ is computationally (rep. statistically) \emph{hiding} if $\Tr_{\regR}(Q_0\ket{0}_{\regC,\regR})$ is computationally (resp. statistically) indistinguishable from $\Tr_{\regR}(Q_1\ket{0}_{\regC,\regR})$. 
We say that it is perfectly hiding if they are identical states.  
\end{definition}

\begin{definition}[Binding]
We say that a canonical quantum bit commitment scheme $\{Q_0,Q_1\}$ is computationally (rep. statistically) \emph{binding} if for any QPT   
(resp. unbounded-time) unitary $U$ over $\regR$  
and an additional register $\regZ$ and any polynomial-size state $\ket{\tau}_{\regZ}$, 
it holds that 
\begin{align}
    \left\|(\bra{0}Q_1^\dagger)_{\regC,\regR}(I_{\regC}\otimes U_{\regR,\regZ})((Q_0\ket{0})_{\regC,\regR}\ket{\tau}_{\regZ})\right\|=\negl(\secp).
    \label{binding_asymmetric}
\end{align}
We say that it is perfectly hiding if the LHS is $0$ for all unbounded-time unitary $U$. 
\footnote{
The above definition is asymmetric for 0 and 1, but
it is easy to show that \cref{binding_asymmetric} implies
\begin{align*}
    \left\|(\bra{0}Q_0^\dagger)_{\regC,\regR}(I_{\regC}\otimes U_{\regR,\regZ})((Q_1\ket{0})_{\regC,\regR}\ket{\tau}_{\regZ})\right\|=\negl(\secp)
\end{align*}
for any $U$ and $\ket{\tau}$.
}
\end{definition}

\begin{remark}
One may think that honest-binding defined above is too weak because it only considers honestly generated commitments. 
However, somewhat surprisingly, \cite{AC:Yan22} proved that it is equivalent to another binding notion called the \emph{sum-binding}~\cite{EC:DumMaySal00}.\footnote{The term ``sum-binding'' is taken from \cite{EC:Unruh16}.} The sum-binding property requires that the sum of probabilities that any (quantum polynomial-time, in the case of computational binding) \emph{malicious} sender can open a commitment to $0$ and $1$ is at most $1+\negl(\secp)$. In addition, it has been shown that the honest-binding property is sufficient for cryptographic applications including zero-knowledge proofs/arguments (of knowledge), 
oblivious transfers, and multi-party computation~\cite{YWLQ15,FUYZ20,C:MorYam22,AC:Yan21}. In this paper, we refer to honest-binding if we simply write binding.  
\end{remark}

In this paper, we use the following result.
\begin{theorem}[Converting flavors \cite{AC:Yan22,EC:HhaMorYam23}]
\label{thm:convertingflavors}
Let $\{Q_0,Q_1\}$ be a canonical quantum bit commitment scheme.
Then there exists a canonical quantum bit commitment scheme $\{Q_0',Q_1'\}$, 
and the following hold for 
$\mathtt{X},\mathtt{Y}\in\{\text{computationally,statistically,perfectly}\}$:
\begin{itemize}
\item
If $\{Q_0,Q_1\}$ is $\mathtt{X}$ hiding, then $\{Q_0',Q_1'\}$
is $\mathtt{X}$ binding.
\item
If $\{Q_0,Q_1\}$ is $\mathtt{Y}$ binding, then $\{Q_0',Q_1'\}$
is $\mathtt{Y}$ hiding.
\end{itemize}
\end{theorem}

\ifnum\submission=1
\else
\subsection{PRSGs}
Although in this paper we do not use PRSGs, we provide its definition
in \cref{sec:PRSs} for the convenience of readers.
\fi

\newcommand{\che}{k}
\newcommand{\puz}{\mathsf{puz}}
\newcommand{\ans}{\mathsf{ans}}
\newcommand{\prefix}{\mathsf{prefix}}
\newcommand{\estimate}{\mathsf{Estimate}}
\newcommand{\extend}{\mathsf{Extend}}
\newcommand{\CheckGen}{\mathsf{CheckGen}}
\newcommand{\PuzzleGen}{\mathsf{PuzzleGen}}

\section{OWSGs}

In this section, we first define OWSGs (\cref{sec:OWSG_def}).
We then define weak OWSGs and show that weak OWSGs are equivalent to OWSGs (\cref{sec:amplification}).

\subsection{Definition of OWSGs}
\label{sec:OWSG_def}

In this subsection, we define OWSGs.
Note that the definition below is 
a generalization of the one given in \cite{C:MorYam22} in the following three points.
First, in \cite{C:MorYam22}, the generated states are pure, but here they can be mixed. 
Second, in \cite{C:MorYam22}, the secret key $k$ is uniformly sampled at random, 
but now it is sampled by a QPT algorithm.
Third, in \cite{C:MorYam22}, the verification algorithm
is the specific algorithm that accepts the alleged key $k'$ with probability $|\langle\phi_k|\phi_{k'}\rangle|^2$,
while here we consider a general verification algorithm.
We think the definition below is more general (and therefore more fundamental) than that in \cite{C:MorYam22}.
Hence hereafter we choose the definition below as the definition of OWSGs.

\begin{definition}[One-way states generators (OWSGs)]
\label{def:OWSG}
A one-way states generator (OWSG) is a set of algorithms $(\KeyGen,\StateGen,\Ver)$ such that
\begin{itemize}
\item
$\KeyGen(1^\secp)\to k:$
It is a QPT algorithm that, on input the security parameter $\secp$, outputs a classical key $k\in\bit^\kappa$.
    \item 
    $\StateGen(k)\to \phi_k:$ It is a QPT algorithm that, on input $k$, outputs 
    an $m$-qubit quantum state $\phi_k$. 
    \item
    $\Ver(k',\phi_k)\to\top/\bot:$ It is a QPT algorithm that, on input $\phi_k$ and a bit string $k'$, outputs $\top$ or $\bot$. 
\end{itemize}

We require the following correctness and security.

\paragraph{\bf Correctness:}
\begin{eqnarray*}
\Pr[\top\gets\Ver(k,\phi_k):k\gets\KeyGen(1^\secp),\phi_k\gets\StateGen(k)]\ge1-\negl(\secp).
\end{eqnarray*}

\paragraph{\bf Security:}
For any QPT adversary $\cA$ and any polynomial $t$\footnote{$\StateGen$ is actually run $t$ times to generate $t$ copies of $\phi_k$, but for simplicity, we just write
$\phi_k\gets\StateGen(k)$ only once. This simplification will often be used in this paper.},
\begin{eqnarray*}
\Pr[\top\gets\Ver(k',\phi_k):k\gets\KeyGen(1^\secp),\phi_k\gets\StateGen(k),k'\gets\cA(1^\secp,\phi_k^{\otimes t})]\le\negl(\secp).
\end{eqnarray*}
\end{definition}

\ifnum\submission=1
\else
\begin{remark}
In \cref{sec:Verpurephi}, we show that if
all $\phi_k$ are pure and
$\Pr[\top\gets\Ver(k,\phi_k)]\ge1-\negl(\secp)$ for all $k$,
restricting $\Ver$ to the following specific algorithm (used in \cite{C:MorYam22}) does not lose the generality:
On input $k'$ and $\phi_k$, measure $\phi_k$ with the basis
$\{|\phi_{k'}\rangle\langle\phi_{k'}|,I-|\phi_{k'}\rangle\langle\phi_{k'}|\}$.
If the first result is obtained output $\top$. Otherwise, output $\bot$.
\end{remark}
\fi

\begin{remark}
If $\phi_k$ is pure, $\StateGen$ runs as follows.
Apply a QPT unitary $U$ on $|k\rangle|0...0\rangle$ to generate
$|\phi_k\rangle\otimes|\eta_k\rangle$, and output $|\phi_k\rangle$.
In this case, the existence of the ``junk state'' $|\eta_k\rangle$ is essential, because otherwise
it is not secure against a QPT adversary who does the application of $U^\dagger$ and the computational-basis measurement.
\end{remark}

\ifnum\submission=1
\else
\begin{remark}
OWSGs are constructed from PRSGs~\cite{C:MorYam22}.
In \cite{C:MorYam22}, OWSGs are constructed from PRSGs with $m\ge c\kappa$ for $c>1$. (Here, $\kappa$ is the key-length (i.e., the input length of $\StateGen$), and
$m$ is the output length of $\StateGen$ (i.e., the number of qubits of $\phi_k$).)
It can be improved to the construction of OWSGs from PRSGs with $m\ge\log \kappa$.\footnote{We thank Luowen Qian for pointing out the fact to us.}
(For a proof, see \cref{sec:OWSGfromPRSG_improved}.) 
\end{remark}
\fi

\begin{remark}
Note that statistically-secure OWSGs do not exist. In other words, there exists an unbounded algorithm $\cA$ that can break the security of OWSGs as follows:
\begin{enumerate}
    \item 
    Given $\phi_k^{\otimes t}$ with a certain polynomial $t$ as input, run the shadow tomography algorithm~\cite{Shadow2} to find $k'$ such that
    $\Pr[\Ver(k',\phi_k)\to\top]\ge1-\frac{1}{\poly(\secp)}$. If there exists such $k'$, such $k'$ can be found with only a certain polynomial $t$.
    If there is no such $k'$, choose $k'$ uniformly at ramdom.
    \item 
    Output $k'$.
\end{enumerate}
\end{remark}

\subsection{Hardness Amplification for OWSGs}
\label{sec:amplification}
In this subsection,
we define a weaker variant called weak one-way states generators (wOWSGs), and show that they are equivalent to OWSGs.

wOWSGs
are defined as follows. 
\begin{definition}[Weak one-way states generators (wOWSGs)]
\label{def:wOWSG}
A weak one-way states generator (wOWSG) is a tuple of algorithms $(\KeyGen,\StateGen,\Ver)$ defined similarly to OWSGs except that  the security is replaced with the following weak security. 

\paragraph{\bf Weak Security:}
There exists a polynomial $p$ such that for any QPT adversary $\cA$ and any polynomial $t$, 
\begin{eqnarray*}
\Pr[\top\gets\Ver(k',\phi_k):k\gets\KeyGen(1^\secp),\phi_k\gets\StateGen(k),k'\gets\cA(1^\secp,\phi_k^{\otimes t})]\le 1-\frac{1}{p}.
\end{eqnarray*}
\end{definition}

We prove that the existence of wOWSGs imply the existence of OWSGs. This is an analogue of Yao's amplification theorem for OWFs in the classical setting~\cite{FOCS:Yao82a,DBLP:books/cu/Goldreich2001}.
\begin{theorem}\label{thm:amplification_OWSG}
OWSGs exist if and only if wOWSGs exist.
\end{theorem}

We observe that the classical proof for OWFs does not extend to OWSGs at least in a straightforward manner. The reason is that the classical proof uses the (obvious) fact that one can deterministically check if a given pair $(x,y)$ satisfies $f(x)=y$ for a OWF $f$ but we do not have such a deterministic verification algorithm for OWSGs.\footnote{We remark that the concurrent work by Cao and Xue \cite{cryptoeprint:2022/1323} avoids this issue by appropriately modifying the proof strategy.} 
On the other hand, we observe that the proof of a more general hardness amplification theorem for \emph{weakly verifiable puzzles} shown by Canetti, Halevi, and Steiner \cite{TCC:CanHalSte05} extends  to the quantum setting with minor tweak. Thus we choose to show the quantum analogue of \cite{TCC:CanHalSte05}, which is more general than \cite{FOCS:Yao82a}. 
We note that Radian and Sattath~\cite{CoRR:RadSat19} observed that the proof of \cite{TCC:CanHalSte05} extends to the post-quantum setting where the adversary is quantum with essentially the same proof, but what we show is stronger than that since we consider quantum puzzles and answers.  

\begin{remark}
We can also consider an even weaker variant of OWSGs where the correctness bound can be much smaller than $1$ and we only require there is an inverse polynomial gap between the correctness and security bounds. Unfortunately, we do not know how to prove the equivalence between such a further weaker variant of weak OWSGs and standard OWSGs. To prove this, we will need a quantum analogue of the results of \cite{JC:ImpJaiKab09,TCC:HolSch11}, which are generalization of \cite{TCC:CanHalSte05}.
\end{remark}

First, we define a quantum analogue of weakly verifiable puzzles. 
\begin{definition}[Weakly verifiable quantum puzzles]
A weakly verifiable quantum puzzle is a tuple of algorithms $(\CheckGen,\PuzzleGen,\Ver)$ as follows.
\begin{itemize}
    \item 
    $\CheckGen(1^\secp)\to \che:$
    It is a QPT algorithm that, on input the security parameter $\secp$, outputs a classical string $\che$, which we call a ``check key''. 
    \item
    $\PuzzleGen(\che)\to\puz:$
    It is a QPT algorithm that, on input $\che$, outputs a quantum state $\puz$, which we call a ``puzzle''. 
    \item 
    $\Ver(\ans,\che)\to\top/\bot:$
    It is a QPT algorithm that, on input $k$ and a quantum state $\ans$, which we call an ``answer'', outputs $\top$ or $\bot$.
\end{itemize}
\end{definition}

\begin{remark}\label{rem:divide_gen}
Besides that $\puz$ and $\ans$ are quantum, another difference from the classical version  \cite{TCC:CanHalSte05} is that the generation algorithm is divided into $\CheckGen$ that generates $\che$ and $\PuzzleGen$ that generates $\puz$. We define it in this way because we want to consider poly-copy hardness, i.e., the hardness to find a valid answer even given polynomially many copies of the puzzle. If we use a single generation algorithm $\Gen$ that generates $\che$ and $\puz$ simultaneously like \cite{TCC:CanHalSte05}, it is unclear how to define such a poly-copy hardness. 
\end{remark}

\begin{remark}\label{rem:OWSG_as_puzzle}
Later, we will see that OWSGs can be seen as weakly verifiable quantum puzzles:
$\CheckGen$ corresponds to $\KeyGen$ of OWSGs,
$\PuzzleGen$ corresponds to $\StateGen$ of OWSGs,
and $\Ver$ of weakly verifiable quantum puzzles corresponds to $\Ver$ of OWSGs.
Moreover, $\ans$ corresponds to $k'$ of OWSGs, which is the output of the adversary, i.e., $k'\gets\cA(\phi_k^{\otimes t})$.
Because $k'$ is a classical bit string, it is enough for our purpose to consider only classical $\ans$, but here we assume $\ans$ is quantum
to give a more general result.
$\Ver$ of weakly verifiable quantum puzzles takes $k$, while $\Ver$ of OWSGs takes $\phi_k$.
This discrepancy can be easily solved by considering that $\Ver$ of weakly verifiable quantum puzzles first generates 
$\phi_k$ from $k$ and then runs $\Ver$ of OWSGs.
\end{remark}

\begin{definition}
We say that an adversary $\cA$ $(t,\epsilon)$-solves a weakly verifiable quantum puzzle $(\CheckGen,\PuzzleGen,\Ver)$ if 
\begin{eqnarray*}
\Pr[\top\gets\Ver(\ans,\che):\che\gets\CheckGen(1^\secp),\puz\gets\PuzzleGen(\che),\ans\gets\cA(\puz^{\otimes t})]\ge \epsilon.
\end{eqnarray*}
\end{definition}

\begin{definition}[Parallel repetition]
For a weakly verifiable quantum puzzle $(\CheckGen,\PuzzleGen,\Ver)$ and a positive integer $n$, 
we define its $n$-repetition $(\CheckGen^n,\PuzzleGen^n,\Ver^n)$ as follows.
\begin{itemize}
    \item 
    $\CheckGen^n(1^\secp)\to (\che_1,\ldots,\che_n):$
    Run $\che_i\gets\CheckGen(1^\secp)$ for $i\in[n]$ and output $(\che_1,\ldots,\che_n)$. 
    \item
    $\PuzzleGen^n(\che_1,\ldots,\che_n)\to(\puz_1,\ldots,\puz_n):$
   Run $\puz_i\gets\PuzzleGen(\che_i)$ for $i\in[n]$ and output $(\puz_1,\ldots,\puz_n)$. 
    \item 
    $\Ver^n(
    (\ans_1,\ldots,\ans_n),
    (\che_1\ldots,\che_n)
    )\to\top/\bot:$
   Run $\Ver(\ans_i,\che_i)$ for $i\in[n]$ and output $\top$ if and only if all the execution of $\Ver$ outputs $\top$. 
\end{itemize}
\end{definition}

\begin{theorem}\label{thm:amplify_qpuzzle}
Let $n,q,t$ be polynomials and $\delta\in (0,1)$ be an inverse polynomial in $\secp$. 
If there exists a QPT adversary $\cA$ that $(t,\delta^n)$-solves $(\CheckGen^n,\PuzzleGen^n,\Ver^n)$, then there exists a polynomial $t'$ and a QPT adversary $\cA'$ that $(t',\delta(1-\frac{1}{q}))$-solves $(\CheckGen,\PuzzleGen,\Ver)$.\footnote{$q$ is a parameter chosen freely that affects the time complexity and the success probability of $\cA'$, but here we do not write the time complexity
of $\cA'$ explicitly, because we are not interested in it as long as it is QPT.}
\end{theorem}
\begin{proof}[Proof of \cref{thm:amplify_qpuzzle} (sketch)]
Since the proof is similar to that in the classical case (\cite[Lemma 1]{TCC:CanHalSte05}), we only explain how to modify it. See 
\ifnum\submission=1
the full version
\else
\cref{sec:proof_amplification} 
\fi
for the full proof. 

First, let us recall the proof overview in the classical case where $\puz$ and $\ans$ are classical.  
The construction of $\cA'$ can be divided into the ``preprocessing phase'' and ``online phase''. 
In the preprocessing phase, $\cA'$ finds a ``good'' prefix $\prefix_{v-1}=(\puz_1,\ldots,\puz_{v-1})$ for some $v\in [n]$. (The actual definition of goodness does not matter in this proof sketch.\footnote{For the readers who want to find the correspondence with the proof in \cite{TCC:CanHalSte05}, we say that $\prefix_{v-1}$ is good if it satisfies Equations (1) and (2) of \cite{TCC:CanHalSte05}.})
In the online phase, given a puzzle $\puz$ as an instance, $\cA'$ repeats the following:
\begin{itemize}
    \item 
$\cA'$ runs  $(\ans_1,\ldots,\ans_n)\gets \cA(\puz_1,\ldots,\puz_{v-1},\puz,\puz_{v+1},\ldots,\puz_n)$ where $\prefix_{v-1}=(\puz_1,\ldots,\puz_{v-1})$ is the ``good'' prefix found in the preprocessing phase, $\puz$ is the given instance, and $(\puz_{v+1},\ldots,\puz_{n})$ is generated by $\cA'$ itself along with the corresponding check key  $(\che_{v+1},\ldots,\che_{n})$. 
\item $\cA'$ runs $\Ver(\ans_i,\che_{i})$ for $i\in \{v+1,\ldots,n\}$ and outputs $\ans_{v}$ if  $\Ver(\ans_i,\che_{i})$ outputs $\top$ for all $i\in \{v+1,\ldots,n\}$. (Note that this is possible because $\cA'$ generates $\che_i$ for $i\in \{v+1,\ldots,n\}$ by itself.) Otherwise, it continues running the loop.  
\end{itemize}
If $\cA'$ does not halt after running the loop sufficiently many times, it aborts. 
They show that 
\begin{itemize}
    \item The probability that $\cA'$ fails to find a good prefix is at most $\frac{\delta}{6q}$, and
    \item For any ``good'' prefix, $\cA'$ solves the puzzle with probability at least $\delta(1-\frac{5}{6q})$. 
\end{itemize}
Combining the above, we can conclude that the overall probability that $\cA'$ solves the puzzle is at least $\delta(1-\frac{1}{q})$. 

When we generalize this proof to the quantum setting where $\puz$ and $\ans$ are quantum, there is an issue that $\cA'$ cannot reuse the ``good'' prefix in the online phase since it consists of \emph{quantum} puzzles, which cannot be copied.\footnote{Actually, a similar problem already occurs when searching for a ``good'' prefix in the preprocessing phase.}
To resolve this problem, we observe that we can regard $\che$ as a classical description of $\puz\gets \PuzzleGen(\che)$. 
Then our idea is to define a prefix to be a sequence of check keys $\che$ instead of puzzles $\puz$.\footnote{We remark that this does not work in the original formulation of (classical) weakly verifiable puzzle in \cite{TCC:CanHalSte05} because $\che$ and $\puz$ are generated simultaneously by a single generation algorithm. For this idea to work, it is important that we divide the generation algorithm in to $\CheckGen$ that generates $\che$ and $\PuzzleGen$ that generates $\puz$ from $\che$ as noted in \cref{rem:divide_gen}. \label{footnote:divide_gen}} 
Since $\che$ is classical and in particular can be reused many times, the above issue is resolved. 
Another issue is that the online phase of $\cA'$ has to run $\cA$ many times where it embeds the problem instance $\puz$ into the input of $\cA$, but $\puz$ is quantum and thus cannot be copied. To resolve this issue, we simply provide sufficiently many copies of $\puz$ to $\cA'$.\footnote{If $\cA$ is a non-uniform adversary with quantum advice, $\cA'$ also needs to take sufficiently many copies of the advice of $\cA$ as its own advice.}  Due to this modification,  \cref{thm:amplify_qpuzzle} does not preserve the number of copies, i.e., $t'$ should be much larger than $t$, but we still have $t'=t\cdot \poly(\secp)$ since $\cA'$ runs $\cA$ only polynoimally many times.

With the above differences in mind, it is straightforward to extend the proof of \cite[Lemma 1]{TCC:CanHalSte05} to prove \cref{thm:amplify_qpuzzle}. 
\end{proof}

\begin{remark}
We remark that the above proof does not work if $\che$ is quantum and potentially entangled with $\puz$.  
(In such a setting, we can only consider single-copy hardness.)
A generalization to such a setting would resolve the open problem about the amplification for computational binding of canonical quantum bit commitments asked by Yan~\cite[Section 12, Problem 5]{cryptoeprint:2020/1488}.  
In the above proof, it is crucial that we have a classical description of $\puz$ as $\che$, and we do not know how to extend it to the case where $\che$ is quantum. 
\end{remark}

We are ready to prove \Cref{thm:amplification_OWSG}. 
\begin{proof}[Proof~of~\Cref{thm:amplification_OWSG}]
The ``only if'' direction is trivial since any OWSG is also wOWSG. In the following, we show the ``if'' direction, i.e., OWSGs exist if wOWSGs exist.
Let $(\KeyGen,\StateGen,\Ver)$ be a wOWSG. Then there exists a polynomial $p$ such that for any QPT adversary $\cA'$ and any polynomial $t'$, 
\begin{eqnarray*}
\Pr[\top\gets\Ver(k',\phi_k):k\gets\KeyGen(1^\secp),\phi_k\gets\StateGen(k),k'\gets\cA'(1^\secp,\phi_k^{\otimes t'})]\le 1-\frac{1}{p}.
\end{eqnarray*}
Then we construct a OWSG $(\KeyGen^n,\StateGen^n,\Ver^n)$ as follows, where $n=\poly(\secp)$ is chosen in such a way that $(1-\frac{1}{2p})^n=2^{-\Omega(\secp)}$. (For example, we can set $n=p\secp$.) 

\begin{itemize}
\item
$\KeyGen^n(1^\secp)\to (k_1,\ldots,k_n):$
Run $k_i\gets \KeyGen^n(1^\secp)$ for $i\in [n]$ and output $(k_1,\ldots,k_n)$.
    \item 
    $\StateGen^n((k_1,\ldots,k_n))\to (\phi_{k_1},\ldots,\phi_{k_n}):$ 
    Run $\phi_{k_i}\gets \StateGen(k_i)$ for $i\in [n]$ and output $(\phi_{k_1},\ldots,\phi_{k_n})$. 
    \item
    $\Ver^n((k'_1,\ldots,k'_n),(\phi_{k_1},\ldots,\phi_{k_n}))\to\top/\bot:$ 
    Run $\Ver(k'_i,\phi_{k_i})$ for $i\in [n]$ and output $\top$ if and only if $\Ver(k'_i,\phi_{k_i})$ outputs $\top$ for all $i\in [n]$.
\end{itemize}
It is clear that $(\KeyGen^n,\StateGen^n,\Ver^n)$ satisfies correctness. 
Let us next show the security.
Toward contradiction, suppose that $(\KeyGen^n,\StateGen^n,\Ver^n)$ is not secure. 
Then there exists a QPT adversary $\cA$ and a polynomial $t$ such that 
\begin{eqnarray*}
\Pr\left[\top\gets\Ver^n((k'_1,\ldots,k'_n),(\phi_{k_1},\ldots,\phi_{k_n})):
\begin{array}{l}
(k_1,\ldots,k_n)\gets\KeyGen^n(1^\secp),\\
(\phi_{k_1},\ldots,\phi_{k_n})\gets\StateGen^n(k_1,\ldots,k_n),\\
(k'_1,\ldots,k'_n)\gets\cA(1^\secp,\phi_{k_1}^{\otimes t},\ldots,\phi_{k_n}^{\otimes t})
\end{array}\right]
\end{eqnarray*}
is non-negligible, which, in particular, is larger than  $(1-\frac{1}{2p})^n=2^{-\Omega(\secp)}$ for infinitely many $\secp$. 
Since OWSGs can be seen as a weakly verifiable quantum puzzle as noted in \cref{rem:OWSG_as_puzzle}, 
by setting 
$q\seteq 2p$ and 
$\delta\seteq 1-\frac{1}{2p}$ in
\cref{thm:amplify_qpuzzle}, there is a QPT algorithm $\A'$ and a polynomial $t'$ such that  
\begin{align*}
&\Pr[\top\gets\Ver(k',\phi_k):k\gets\KeyGen(1^\secp),\phi_k\gets\StateGen(k),k'\gets\cA'(1^\secp,\phi_k^{\otimes t'})]\\
&\ge \delta \left(1-\frac{1}{q}\right)
=
\left(1-\frac{1}{2p}\right) \left(1-\frac{1}{2p}\right)
> 1-\frac{1}{p}
\end{align*}
for infinitely many $\secp$. 
This contradicts the assumption. Therefore $(\KeyGen^n,\StateGen^n,\Ver^n)$ is a OWSG. 
\end{proof}

\section{QDSs}
\label{sec:QDS}

In this section,  
we first define QDSs (\cref{sec:QDS_def}), and show that
one-time-secure QDSs can be extended to $q$-time-secure ones (\cref{sec:QDS_q}).
We then show that one-time-secure QDSs are equivalent to OWSGs (\cref{sec:QDS_OWSG}).

\subsection{Definition of QDSs} 
\label{sec:QDS_def}
Quantum digital signatures are defined as follows.
\begin{definition}[Quantum digital signatures (QDSs)~\cite{C:MorYam22}]
\label{definition:signatures}
A quantum digital signature (QDS) scheme is a set of algorithms 
$(\SKGen,\PKGen,\Sign,\Ver)$ such that 
\begin{itemize}
\item
$\SKGen(1^\secp)\to\sk:$ It is a QPT algorithm that, on input the security parameter $\lambda$, outputs a classical secret key $\sk$.
\item
$\PKGen(\sk)\to\pk:$ It is a QPT algorithm that, on input $\sk$, outputs 
a quantum public key $\pk$. 
\item
$\Sign(\sk,m)\to\sigma:$ It is a QPT algorithm that, on input $\sk$ and a message $m$, outputs a classical signature $\sigma$.
\item
$\Ver(\pk,m,\sigma)\to\top/\bot:$ It is a QPT algorithm that, on input $\pk$, $m$, and $\sigma$, outputs $\top/\bot$.
\end{itemize}

We require the correctness and the security as follows.

\paragraph{\bf Correctness:}
For any $m$,
\begin{eqnarray*}
\Pr\left[
\top\leftarrow\Ver(\pk,m,\sigma):
\begin{array}{l}
\sk\leftarrow\SKGen(1^\lambda),\\
\pk\leftarrow \PKGen(\sk),\\
\sigma\leftarrow\Sign(\sk,m)
\end{array}
\right]
\ge1-\negl(\secp).
\end{eqnarray*}

\paragraph{\bf $q$-time security:} 
Let us consider the following security game, $\mathsf{Exp}$, between a challenger $\cC$ and a QPT adversary $\cA$:
\begin{enumerate}
\item
$\cC$ runs $\sk\gets\SKGen(1^\secp)$. 
\item
$\cC$ runs $\pk\gets\PKGen(\sk)$ $t$ times, and
sends $\pk^{\otimes t}$ to $\cA$.
\item For $i=1$ to $q$, do:
\begin{enumerate}
\item
$\cA$ sends a message $m^{(i)}$ to $\cC$.
\item
$\cC$ runs $\sigma^{(i)}\gets\mathsf{Sign}(\sk,m^{(i)})$, and sends $\sigma^{(i)}$ to $\cA$.
\end{enumerate}
\item
$\cA$ sends $\sigma'$ and $m'$ to $\cC$.
\item
$\cC$ runs $\pk\gets\PKGen(\sk)$ and $v\gets\mathsf{Ver}(\pk,m',\sigma')$.
If $m'\notin \{m^{(1)},\ldots,m^{(q)}\}$ and $v=\top$, the output of the game is 1.
Otherwise, the output of the game is 0.
\end{enumerate}
For any QPT adversary $\cA$ and any polynomial $t$,
$
\Pr[\mathsf{Exp}=1]\le\negl(\secp).
$
\end{definition}

\begin{remark}
By using the shadow tomography, we can show that statistically-secure QDSs do not exist.
For a proof, see \cref{app:stat_QDSs}.
\end{remark}

\subsection{Extension to $q$-time Security}
\label{sec:QDS_q}

\begin{theorem}\label{thm:q-time}
If one-time-secure QDSs exist, then $q$-time-secure QDSs exist for any polynomial $q$.
\end{theorem}

\ifnum\submission=1
The idea is similar to the one-time to $q$-time conversion for attribute-based encryption in \cite{PKC:ISVWY17}. 
We first consider a scheme where we generate $q^2$ key pairs of one-time-secure scheme and uniformly chooses one of $q^2$ signing keys to generate a signature whenever we run the signing algorithm. This scheme is not $q$-bounded-secure because the probability that the same signing key is used more than once is non-negligible. However, by a simple combinatorial argument, we can upper bound the probability of such a ``bad'' event by some constant smaller than $1$. Thus, by repeating this construction $\secp$ times, we can amplify the security to get $q$-bounded-secure scheme.     

For a formal proof, see the full version.
\else
\begin{proof}[Proof of \cref{thm:q-time}]
The proof is similar to the one-time to $q$-time conversion for attribute-based encryption in \cite{PKC:ISVWY17}.\footnote{We also note that the proof works also for the classical setting.} 

Let $(\mathsf{1DS}.\SKGen,\mathsf{1DS}.\PKGen,\mathsf{1DS}.\Sign,\mathsf{1DS}.\Ver)$ be a one-time-secure QDS scheme. For a polynomial $q$, we construct a $q$-time-secure QDS scheme $(\mathsf{qDS}.\SKGen,\allowbreak\mathsf{qDS}.\PKGen,\mathsf{qDS}.\Sign,\mathsf{qDS}.\Ver)$ as follows. 
\begin{itemize}
\item
$\mathsf{qDS}.\SKGen(1^\secp)\to\sk:$ 
Run $\sk_{a,b}\gets \mathsf{1DS}.\SKGen(1^\secp)$ for $a\in [\secp]$ and $b\in [q^2]$ and output $\sk\seteq (\sk_{a,b})_{a\in[\secp],b\in[q^2]}$. 
\item
$\mathsf{qDS}.\PKGen(\sk)\to\pk:$ 
Parse $\sk= (\sk_{a,b})_{a\in[\secp],b\in[q^2]}$,  
run $\pk_{a,b}\gets \mathsf{1DS}.\PKGen(\sk_{a,b})$ for $a\in [\secp]$ and $b\in [q^2]$ and output $\pk\seteq \bigotimes_{a\in[\secp],b\in[q^2]}\pk_{a,b}$.  
\item
$\mathsf{qDS}.\Sign(\sk,m)\to\sigma:$ 
Parse $\sk= (\sk_{a,b})_{a\in[\secp],b\in[q^2]}$, choose $b_a\gets[q^2]$ for $a\in[\secp]$, run  $\sigma_a\gets\mathsf{1DS}.\Sign(\sk_{a,b_a},m)$ for $a\in[\secp]$, and output $\sigma=(b_a,\sigma_a)_{a\in[\secp]}$. 
\item
$\mathsf{qDS}.\Ver(\pk,m,\sigma)\to\top/\bot:$ 
Parse 
$\pk= \bigotimes_{a\in[\secp],b\in[q^2]}\pk_{a,b}$ and
$\sigma=(b_a,\sigma_a)_{a\in[\secp]}$, run $\mathsf{1DS}.\Ver(\pk_{a,b_a},m,\sigma_{a})$ for $a\in[\secp]$, and output $\top$ if and only if all execution of $\mathsf{1DS}.\Ver$ output $\top$.  
\end{itemize}
We show that the above scheme satisfies $q$-time security. Let $\cA$ be an adversary against $q$-time security of the above scheme. In the $q$-time security $\mathsf{Exp}$ as defined in  \cref{definition:signatures}, let $\sigma^{(i)}=(b_a^{(i)},\sigma_{a}^{(i)})_{a\in[\secp]}$ be the $i$-th signature generated by $\cC$ and $\sigma'=(b_a',\sigma_{a}')_{a\in[\secp]}$ be $\cA$'s final output. Let $\mathsf{Good}$ be the event that there is $a\in [\secp]$ such that $b_{a}^{(i)}$ is distinct for all $i\in[q]$. Then we can show that $\Pr[\mathsf{Good}]=1-\negl(\secp)$ similarly to \cite[Lemma 1]{PKC:ISVWY17}. Specifically, this is shown as follows. 
For each $a\in[\secp]$, 
let $\mathsf{Bad}_a$  be the event that $b_{a}^{(i)}$ is \emph{not} distinct for $i\in[q]$.
Then by a simple combinatorial argument, we can see that 
\begin{align*}
    \Pr[\mathsf{Bad}_a]&=1-\left(\frac{q^2(q^2-1)\cdots (q^2-q+1)}{(q^2)^q}\right)\\
    &\le 1-\left(1-\frac{q-1}{q^2}\right)^q.
\end{align*}
for each $a\in[\secp]$. 
Thus, we have 
\begin{align*}
    \Pr[\mathsf{Good}]&=1-\prod_{a\in[\secp]}\Pr[\mathsf{Bad}_a]\\
    &\ge 1-\left( 1-\left(1-\frac{q-1}{q^2}\right)^q\right)^\secp\\
    &= 1-\left( 1-\left(1-\frac{1}{q}+\frac{1}{q^2}\right)^q\right)^\secp\\
    &\ge 1-(1-e^{-1})^\secp\\
    &=1-\negl(\secp).
\end{align*}

Conditioned on $\mathsf{Good}$ occurs and $\cA$ wins (i.e., $\mathsf{Exp}$ returns $1$), if we uniformly choose $a^*\gets[\secp]$ and $b^*\gets[q^2]$, then the probability that $b_{a^*}^{(i)}$ is distinct for all $i\in[q]$ and $b^*=b'_{a^*}$ is at least $\frac{1}{q^2\secp}$. Thus, by randomly guessing $a^*\gets[\secp]$ and $b^*\gets[q^2]$ and embedding a problem instance of $\mathsf{1QDS}$ into the corresponding position, we can reduce the $q$-time security of $\mathsf{qQDS}$ to the one-time security of $\mathsf{1QDS}$. That is, we can construct an adversary $\cB$ against  the one-time security of $\mathsf{1QDS}$ as follows:
\begin{description}
\item[$\cB(\pk^{\otimes t})$:]
Choose $a^*\gets[\secp]$ and $b^*\gets[q^2]$, generate $\sk_{a,b}\gets \mathsf{1DS}.\SKGen(1^\secp)$
and run $\pk_{a,b}\gets \mathsf{1DS}.\PKGen(\sk_{a,b})$ times to get $\pk_{a,b}^{\otimes t}$ for 
$(a,b)\in ([\secp]\times [q^2])\setminus \{(a^*,b^*)\}$, 
set $\pk_{a^*,b^*}^{\otimes t}\seteq\pk^{\otimes t}$, and sends $\pk^{\otimes t}\seteq \bigotimes_{a\in[\secp],b\in[q^2]}\pk_{a,b}^{\otimes t}$ to $\cA$. When $\cA$ sends $m^{(i)}$ for $i\in[q]$, choose $b_a^{(i)}$ for $a\in[\secp]$, generates $\sigma_a\gets\mathsf{1DS}.\Sign(\sk_{a,b_a^{(i)}},m)$ for $a\in [q]\setminus \{a^*\}$ and $\sigma_{a^*}$ in one of the following way:
\begin{itemize}
    \item If $b_{a^*}^{(i)}=b^*$ and $\cB$ has never sent $m$ to the external challenger, send $m$ to the external challenger and let $\sigma_{a^*}^{(i)}$ be the returned signature;  
    \item If $b_{a^*}^{(i)}=b^*$ and $\cB$ has sent $m$ to the external challenger before, immediately abort;
    \item If $b_{a^*}^{(i)}\ne b^*$, generate $\sigma_{a^*}\gets\mathsf{1DS}.\Sign(\sk_{a^*,b_{a^*}^{(i)}},m)$. Note that this is possible without the help of the external challenger since $\cB$ generated $\sk_{a,b}^{(i)}$ by itself for $(a,b)\ne (a^*,b^*)$.
\end{itemize}
Return $\sigma^{(i)}\seteq (b_a^{(i)},\sigma_a^{(i)})_{a\in[\secp]}$ to $\cA$. 
When $\cA$ sends $\sigma'=(b'_a,\sigma'_a)_{a\in[\secp]}$, 
send $\sigma'_{a^*}$ to the external challenger
if $b'_{a^*}=b^*$ and otherwise abort. 
\end{description}
As already discussed, $\mathsf{Good}$ occurs with probability $1-\negl(\secp)$ and conditioned on that $\mathsf{Good}$ occurs,
the probability that $\cB$ wins is at least $\frac{1}{q^2\secp}$ times the probability that $\cA$ wins. Since $q$ is polynomial, the reduction loss is polynomial and thus $\mathsf{qDS}$ is $q$-time secure if $\mathsf{qDS}$ is one-time secure.
\end{proof}
\fi

\subsection{Equivalence of OWSGs and QDSs}
\label{sec:QDS_OWSG}

\begin{theorem}
\label{thm:OWSG_QDS}
OWSGs exist if and only if one-time-secure QDSs exist. 
\end{theorem}

\begin{remark}
By using the equivalence between OWSGs and wOWSGs (\cref{thm:amplification_OWSG}),
the result that one-time-secure QDSs imply OWSGs can be improved to a stronger result (with a similar proof) that
one-time-secure QDSs with weak security imply OWSGs.
Here, the weak security of QDSs means that
there exists a polynomial $p$ such that for any QPT adversary $\cA$ and any polynomial $t$,
$\Pr[\mathsf{Exp}=1]\le1-\frac{1}{p}$.
\end{remark}

\ifnum\submission=1
It is proven in \cite{cryptoeprint:2021/1691} that OWSGs implies one-time-secure QDSs. However, since we generalize the definition of OWSGs, we need to reprove it. Fortunately, almost the same construction as that in \cite{cryptoeprint:2021/1691} works with the generalized definition of OWSGs. Roughly, the construction is as follows when the message space is one-bit: a secret key is $\sk=(k_0,k_1)$, a public key is $\pk=(\phi_{k_0},\phi_{k_1})$, and a signature for a bit $b\in \bit$ is $k_b$. The verification algorithm of QDSs simply runs that of the OWSG. 

For the other direction, we construct OWSGs from QDSs by regarding $\sk$ and $\pk$ of QDSs as $k$ and $\phi_k$ of OWSGs.

For a formal proof, see the full version.
\else
\begin{proof}[Proof of \cref{thm:OWSG_QDS}]
Let us first show that one-time-secure QDSs imply OWSGs.
Let $(\mathsf{DS}.\SKGen,\mathsf{DS}.\PKGen,\allowbreak\mathsf{DS}.\Sign,\mathsf{DS}.\Ver)$ be a one-time-secure QDS scheme for
a single-bit message.
From it, we construct a OWSG as follows.
\begin{itemize}
    \item 
    $\KeyGen(1^\secp)\to k:$ 
    Run $\sk\gets\mathsf{DS}.\SKGen(1^\secp)$.
    Output $k\coloneqq \sk$.
    \item
    $\StateGen(k)\to \phi_k:$
    Parse $k=\sk$.
    Run $\pk\gets\mathsf{DS}.\PKGen(\sk)$.
    Output $\phi_k\coloneqq\pk$.
    \item
    $\Ver(k',\phi_k)\to\top/\bot:$
    Parse $\phi_k=\pk$.
    Run $\sigma\gets\mathsf{DS}.\Sign(k',1)$.
    Run $v\gets\mathsf{DS}.\Ver(\pk,1,\sigma)$.
    Output $v$.
\end{itemize}
Correctness is clear. Let us show the security. Assume that it is not secure.
It means that there exists a QPT adversary $\cA$ and a polynomial $t$ such that
\begin{eqnarray*}
\Pr\left[
\top\gets\mathsf{DS}.\Ver(\pk,1,\sigma):
\begin{array}{cc}
\sk\gets\mathsf{DS}.\SKGen(1^\secp)\\
\pk\gets\mathsf{DS}.\PKGen(\sk)\\
\sk'\gets\cA(\pk^{\otimes t})\\
\sigma\gets\mathsf{DS}.\Sign(\sk',1)
\end{array}
\right]\ge\frac{1}{\poly(\secp)}.
\end{eqnarray*}
From such $\cA$, let us construct a QPT adversary $\cB$ that breaks the security of the QDS as follows:
On input $\pk^{\otimes t}$, it runs $\sk'\gets\cA(\pk^{\otimes t})$. 
It then runs $\sigma\gets\mathsf{DS}.\Sign(\sk',1)$, and outputs $(1,\sigma)$.
It is clear that the probability that $\sigma$ is accepted as a valid signature of the message 1 is equal to the left-hand-side of the above
inequality. Hence the security of the QDS is broken by $\cB$.

Let us next give a proof that OWSGs imply QDSs. The construction is a ``quantum version" of the Lamport signatures~\cite{Lamport_signature},
and is similar to that given in \cite{C:MorYam22}.
Let $(\mathsf{OWSG}.\KeyGen,\allowbreak\mathsf{OWSG}.\StateGen,\mathsf{OWSG}.\Ver)$ be a OWSG.
From it, we construct a one-time-secure QDS for a single-bit message as follows.
\begin{itemize}
\item
$\SKGen(1^\secp)\to\sk$: 
Run $k_0\gets\mathsf{OWSG}.\KeyGen(1^\secp)$.
Run $k_1\gets\mathsf{OWSG}.\KeyGen(1^\secp)$.
Output $\sk\coloneqq(\sk_0,\sk_1)$, where $\sk_b\coloneqq k_b$ for $b\in\bit$.
\item
$\PKGen(\sk)\to\pk$: 
Parse $\sk=(\sk_0,\sk_1)$, where $\sk_b= k_b$ for $b\in\bit$.
Run 
$\phi_{k_b}\gets \mathsf{OWSG}.\StateGen(k_b)$ for each $b\in\bit$.
Output $\pk\coloneqq(\pk_0,\pk_1)$, where $\pk_b\coloneqq\phi_{k_b}$ for $b\in\bit$.
\item
$\Sign(\sk,m)\to\sigma$: 
Parse $\sk=(\sk_0,\sk_1)$, where $\sk_b= k_b$ for $b\in\bit$.
Output $\sigma\coloneqq \sk_m$.
\item
$\Ver(\pk,m,\sigma)\to\top/\bot$:
Parse $\pk=(\pk_0,\pk_1)$, where $\pk_b=\phi_{k_b}$ for $b\in\bit$.
Run $v\gets\mathsf{OWSG}.\Ver(\sigma,\phi_{k_m})$.
Output $v$.
\end{itemize}

It is clear that this construction satisfies correctness.
Let us next show the security.
Let us consider the following security game between the challenger $\cC$ and a QPT adversary $\cA$:
\begin{enumerate}
\item
$\cC$ runs $k_0\gets\mathsf{OWSG}.\KeyGen(1^\secp)$.
$\cC$ runs $k_1\gets\mathsf{OWSG}.\KeyGen(1^\secp)$.
\item
$\cC$ runs $\phi_{k_b}\gets\mathsf{OWSG}.\StateGen(k_b)$ $t+1$ times for $b\in\bit$.
$\cC$ sends $(\phi_{k_0}^{\otimes t},\phi_{k_1}^{\otimes t})$ to $\cA$.
\item
$\cA$ sends $m\in\bit$ to $\cC$.
\item
$\cC$ sends $k_m$ to $\cA$.
\item
$\cA$ sends $\sigma$ to $\cC$.
\item
$\cC$ runs $v\gets\mathsf{OWSG}.\Ver(\sigma,\phi_{k_{m\oplus 1}})$.
If the result is $\top$, $\cC$ outputs 1.
Otherwise, $\cC$ outputs 0. 
\end{enumerate}
Assume that our construction is not one-time secure, which means that $\Pr[\cC\to1]$ is non-negligible for
a QPT $\cA$ and a polynomial $t$.
Then, we can construct a QPT adversary $\cB$ that breaks the security of the OWSG as follows.
(Here, $\cC'$ is the challenger of the security game of the OWSG.)
\begin{enumerate}
\item
$\cC'$ runs $k\gets \mathsf{OWSG}.\KeyGen(1^\secp)$.
$\cC'$ runs $\phi_k\gets\mathsf{OWSG}.\StateGen(k)$ $t+1$ times.
$\cC'$ sends $\phi_k^{\otimes t}$ to $\cB$.
\item
$\cB$ chooses $r\leftarrow\bit$.
$\cB$ runs $k'\gets\mathsf{OWSG}.\KeyGen(1^\secp)$.
$\cB$ runs $\phi_{k'}\gets\mathsf{OWSG}.\StateGen(k')$ $t$ times.
If $r=0$, $\cB$ sends $(\phi_k^{\otimes t},\phi_{k'}^{\otimes t})$ to $\cA$.
If $r=1$, $\cB$ sends $(\phi_{k'}^{\otimes t},\phi_k^{\otimes t})$ to $\cA$.
\item
$\cA$ sends $m\in\bit$ to $\cB$.
\item
If $r=m$, $\cB$ aborts.
If $r\neq m$, $\cB$ sends $k'$ to $\cA$.
\item
$\cA$ sends $\sigma$ to $\cB$.
\item
$\cB$ sends $\sigma$ to $\cC'$.
\item
$\cC'$ runs $v\gets\mathsf{OWSG}.\Ver(\sigma,\phi_k)$.
If $v=\top$, $\cC'$ outputs 1.
Otherwise, $\cC'$ outputs 0. 
\end{enumerate}
By a straightforward calculation, which is given below, 
\begin{eqnarray}
\Pr[\cC'\to1]=\frac{1}{2}\Pr[\cC\to1].
\label{exp1}
\end{eqnarray}
Therefore, 
if $\Pr[\cC\to1]$ is non-negligible, 
$\Pr[\cC'\to1]$ is also non-negligible,
which means that $\cB$ breaks
the security of the OWSG.

Let us show Eq.~(\ref{exp1}).
For simplicity, we write $\Pr[k]$ to mean $\Pr[k\gets\mathsf{OWSG}.\KeyGen(1^\secp)]$.
Then,
\begin{eqnarray*}
&&\Pr[\cC'\to1]\\
&=&\sum_{k,k',\sigma}
\Pr[k]\Pr[k']
\frac{1}{2}\Pr[1\leftarrow \cA(\phi_k^{\otimes t},\phi_{k'}^{\otimes t})]
\Pr[\sigma\leftarrow \cA(k')]
\Pr[\top\gets\mathsf{OWSG}.\Ver(\sigma,\phi_k)]\\
&&+\sum_{k,k',\sigma}\Pr[k]\Pr[k']\frac{1}{2}\Pr[0\leftarrow \cA(\phi_{k'}^{\otimes t},\phi_k^{\otimes t})]
\Pr[\sigma\leftarrow \cA(k')]
\Pr[\top\gets\mathsf{OWSG}.\Ver(\sigma,\phi_k)]\\
&=&\sum_{k,k',\sigma}\Pr[k]\Pr[k']\frac{1}{2}\Pr[1\leftarrow \cA(\phi_k^{\otimes t},\phi_{k'}^{\otimes t})]
\Pr[\sigma\leftarrow \cA(k')]
\Pr[\top\gets\mathsf{OWSG}.\Ver(\sigma,\phi_k)]\\
&&+\sum_{k,k',\sigma}\Pr[k]\Pr[k']\frac{1}{2}\Pr[0\leftarrow \cA(\phi_{k}^{\otimes t},\phi_{k'}^{\otimes t})]
\Pr[\sigma\leftarrow \cA(k)]
\Pr[\top\gets\mathsf{OWSG}.\Ver(\sigma,\phi_{k'})]\\
&=&\frac{1}{2}\Pr[\cC\to1].
\end{eqnarray*}

\end{proof}
\fi

\section{Quantum Money}
\label{sec:money}

In this section, we first define private-key quantum money schemes (\cref{sec:money_def}).
We then construct OWSGs from quantum money schemes with pure money states (\cref{sec:OWSGfromQmoney_pure}).
We also show that OWSGs can be constructed from quantum money schemes where the verification algorithms
satisfy a certain symmetric property (\cref{sec:OWSGfromQmoney_symmetric}).

\subsection{Definition of Private-key Quantum Money}
\label{sec:money_def}
Private-key quantum money schemes are defined as follows. 
\begin{definition}[Private-key quantum money \cite{C:JiLiuSon18,STOC:AarChr12}]
A private-key quantum money scheme is a set of algorithms 
$(\KeyGen,\Mint,\Ver)$ such that 
\begin{itemize}
\item
$\KeyGen(1^\secp)\to k:$
It is a QPT algorithm that, on input the security parameter $\secp$,
outputs a classical secret key $k$.
\item
$\Mint(k)\to \$_k:$
It is a QPT algorithm that, on input
$k$, outputs an $m$-qubit quantum state $\$_k$.
\item
$\Ver(k,\rho)\to\top/\bot:$
It is a QPT algorithm that, on input $k$ and a quantum
state $\rho$, outputs $\top/\bot$.
\end{itemize}
We require the following correctness and security.

\paragraph{\bf Correctness:}
\begin{eqnarray*}
\Pr[\top\gets\Ver(k,\$_k):k\gets\KeyGen(1^\secp),\$_k\gets\Mint(k)]
\ge1-\negl(\secp).
\end{eqnarray*}

\paragraph{\bf Security:}
For any QPT adversary $\cA$ and any polynomial $t$,
\begin{eqnarray*}
\Pr[\mathsf{Count}(k,\xi)\ge t+1
:k\gets\KeyGen(1^\secp),\$_k\gets\Mint(k),\xi\gets\cA(1^\secp,\$_k^{\otimes t})]
\le\negl(\secp),
\end{eqnarray*}
where $\xi$ is a quantum state on $\ell$ registers, $R_1,...,R_\ell$, each of which is of $m$ qubits,
and $\mathsf{Count}$ is the following QPT algorithm:
on input $\xi$, it runs $\top/\bot\gets\Ver(k,\xi_j)$ for each $j\in[1,2,...,\ell]$, where
$\xi_j\coloneqq\mbox{Tr}_{R_1,...,R_{j-1},R_{j+1},...,R_\ell}(\xi)$,
and outputs the total number of $\top$.
\end{definition}

\begin{remark}
Private-key quantum money schemes are
constructed 
from PRSGs~\cite{C:JiLiuSon18}.
\end{remark}

\begin{remark}
As is shown in \cite{Shadow2}, private-key quantum money schemes    
are broken by an unbounded adversary,
and therefore statistically-secure private-key quantum money schemes do not exist.
(The idea is as follows: the unbounded adversary first finds all $\{k_i\}_i$ such that
$\Ver(k_i,\$_k)$ is large with the shadow tomography, and then searches a state $\rho$ by the brute-force
such that $\Ver(k_i,\rho)$ is close to $\Ver(k_i,\$_k)$ FOR ALL $i$. Finally, the adversary outputs
many copies of $\rho$.)
\end{remark}

\subsection{OWSGs from Quantum Money with Pure Money States}
\label{sec:OWSGfromQmoney_pure}

\begin{theorem}\label{thm:money_pure}
If private-key quantum money schemes with pure quantum money states exist, then OWSGs exist.
\end{theorem}
\begin{remark}
For example, the private-key quantum money scheme of \cite{C:JiLiuSon18} has pure quantum money states.
\end{remark}
\begin{remark}
By using the equivalence between OWSGs and wOWSGs (\cref{thm:amplification_OWSG}),
this result can be improved to a stronger result (with a similar proof) that
private-key quantum money schemes with pure quantum money states and with weak security imply OWSGs.
Here, the weak security means that
there exists a polynomial $p$ such that for any QPT adversary $\cA$ and any polynomial $t$,
\begin{eqnarray*}
\Pr[\mathsf{Count}(k,\xi)\ge t+1
:k\gets\KeyGen(1^\secp),\$_k\gets\Mint(k),\xi\gets\cA(1^\secp,\$_k^{\otimes t})]
\le1-\frac{1}{p}.
\end{eqnarray*}
\end{remark}


\begin{proof}[Proof of \cref{thm:money_pure}]
Let $(\mathsf{QM}.\KeyGen,\mathsf{QM}.\Mint,\mathsf{QM}.\Ver)$ be a private-key quantum money scheme with pure money states.
From it, we construct a OWSG as follows.
\begin{itemize}
    \item 
    $\KeyGen(1^\secp)\to k:$
    Run $k\gets\mathsf{QM}.\KeyGen(1^\secp)$.
    Output $k$.
    \item
    $\StateGen(k)\to\phi_k:$
    Run $|\$_k\rangle\gets\mathsf{QM}.\Mint(k)$.
    Output $\phi_k\coloneqq |\$_k\rangle\langle\$_k|$.
    \item
    $\Ver(k',\phi_k)\to\top/\bot:$
    Parse $\phi_k=|\$_k\rangle\langle\$_k|$.
    Measure $|\$_k\rangle$ with the basis
    $\{|\$_{k'}\rangle\langle\$_{k'}|,I-|\$_{k'}\rangle\langle\$_{k'}|\}$,
    and output $\top$ if the first result is obtained.
    Output $\bot$ if the second result is obtained.
    (This measurement is done in the following way: generate $U(|k'\rangle|0...0\rangle)=|\$_{k'}\rangle|\eta_{k'}\rangle$,
    and discard the first register. Then apply $U^\dagger$ on $|\$_k\rangle|\eta_{k'}\rangle$, and measure all qubits
    in the computationl basis. If the result is $k'0...0$, accept. Otherwise, reject.)
\end{itemize}
The correctness is clear. Let us show the security. Assume that it is not secure.
Then, there exists a QPT adversary $\cA$, a polynomial $t$, and a polynomial $p$ such that
\begin{eqnarray*}
\sum_{k,k'}
\Pr[k\gets\mathsf{QM}.\KeyGen(1^\secp)]
    \Pr[k'\gets\cA(|\$_k\rangle^{\otimes t})]
    |\langle\$_k|\$_{k'}\rangle|^2
    \ge\frac{1}{p}.
    \end{eqnarray*}
    
Define the set
    \begin{eqnarray*}
    S\coloneqq\Big\{(k,k')~\Big|~|\langle\$_k|\$_{k'}\rangle|^2\ge
    \frac{1}{2p}\Big\}.
    \end{eqnarray*}
    Then, we have
    \begin{eqnarray*}
    \sum_{(k,k')\in S}\Pr[k\gets\mathsf{QM}.\KeyGen(1^\secp)]\Pr[k'\gets\cA(|\$_k\rangle^{\otimes t})]> \frac{1}{2p}.
    \end{eqnarray*}
    This is shown as follows.
   \begin{eqnarray*} 
   \frac{1}{p}&\le&
    \sum_{k,k'}\Pr[k\gets\mathsf{QM}.\KeyGen(1^\secp)]\Pr[k'\gets\cA(|\$_k\rangle^{\otimes t})]
    |\langle\$_k|\$_{k'}\rangle|^2\\
    &=&\sum_{(k,k')\in S}\Pr[k\gets\mathsf{QM}.\KeyGen(1^\secp)]\Pr[k'\gets\cA(|\$_k\rangle^{\otimes t})]
    |\langle\$_k|\$_{k'}\rangle|^2\\
    &&+\sum_{(k,k')\notin S}\Pr[k\gets\mathsf{QM}.\KeyGen(1^\secp)]\Pr[k'\gets\cA(|\$_k\rangle^{\otimes t})]
    |\langle\$_k|\$_{k'}\rangle|^2\\
    &<&
    \sum_{(k,k')\in S}\Pr[k\gets\mathsf{QM}.\KeyGen(1^\secp)]\Pr[k'\gets\cA(|\$_k\rangle^{\otimes t})]
    +\frac{1}{2p}.
   \end{eqnarray*} 
  Let us also define  
  \begin{eqnarray*}
  T\coloneqq\Big\{k~\Big|~\Pr[\top\gets\mathsf{QM}.\Ver(k,|\$_k\rangle)]
  \ge1-\frac{1}{8p}\Big\}.
  \end{eqnarray*}
  Then, 
  \begin{eqnarray*}
  \sum_{k\in T}\Pr[k\gets\mathsf{QM}.\KeyGen(1^\secp)]> 1-\negl(\secp). 
  \end{eqnarray*}
  This is shown as follows.
  \begin{eqnarray*}
  1-\negl(\secp)&\le&\sum_k\Pr[k\gets\mathsf{QM}.\KeyGen(1^\secp)]\Pr[\top\gets\mathsf{QM}.\Ver(k,|\$_k\rangle)]\\
  &=&\sum_{k\in T}\Pr[k\gets\mathsf{QM}.\KeyGen(1^\secp)]\Pr[\top\gets\mathsf{QM}.\Ver(k,|\$_k\rangle)]\\
  &&+\sum_{k\notin T}\Pr[k\gets\mathsf{QM}.\KeyGen(1^\secp)]\Pr[\top\gets\mathsf{QM}.\Ver(k,|\$_k\rangle)]\\
  &<& \sum_{k\in T}\Pr[k\gets\mathsf{QM}.\KeyGen(1^\secp)]\\
  &&+\Big(1-\frac{1}{8p}\Big)\Big(1-\sum_{k\in T}\Pr[k\gets\mathsf{QM}.\KeyGen(1^\secp)]\Big).
  \end{eqnarray*}
  Here, the first inequality is from the correctness of the quantum money scheme.
   
  Let us fix $(k,k')$ such that $(k,k')\in S$ and $k\in T$. 
  The probability of having such $(k,k')$ is, from the union bound, 
  \begin{eqnarray*}
  \sum_{(k,k')\in S \wedge k\in T}\Pr[k\gets\mathsf{QM}.\KeyGen(1^\secp)]\Pr[k'\gets\cA(|\$_k\rangle^{\otimes t})]
  &>& \frac{1}{2p}+1-\negl(\secp)-1\nonumber\\
  &=& \frac{1}{2p}-\negl(\secp).
  \end{eqnarray*}
  
     From the $\cA$, we construct a QPT adversary $\cB$ that breaks the security of the private-key quantum money 
    scheme as follows:
    On input $|\$_k\rangle^{\otimes t}$, it runs $k'\gets\cA(|\$_k\rangle^{\otimes t})$.
    It then runs $|\$_{k'}\rangle\gets\mathsf{QM}.\Mint(k')$ $\ell$ times, where 
  $\ell$ is a polynomial specified later,
    and outputs $\xi\coloneqq |\$_{k'}\rangle^{\otimes \ell}$.
    Let us show that thus defined $\cB$ breaks the security of the private-key quantum money scheme.
  Let 
   $v_j$ be the bit that is 1 if the output of $\mathsf{QM}.\Ver(k,\xi_j)$ is $\top$, and
   is 0 otherwise.
   Then, for any $(k,k')$ such that $(k,k')\in S$ and $k\in T$,
   \begin{eqnarray*}
   \Pr[v_j=1]
   &=&\Pr[\top\gets\mathsf{QM}.\Ver(k,\xi_j)]\\
   &=&\Pr[\top\gets\mathsf{QM}.\Ver(k,|\$_{k'}\rangle)]\\
   &\ge&\Pr[\top\gets\mathsf{QM}.\Ver(k,|\$_k\rangle)]-\sqrt{1-\frac{1}{2p}}\\
   &\ge&1-\frac{1}{8p}-\sqrt{1-\frac{1}{2p}}\\
   &\ge&\frac{1}{8p}
   \end{eqnarray*}
   for each $j\in[1,2,...,\ell]$.
   Here, in the first inequality, we have used the fact that
   $\Pr[1\gets\cD(|\$_k\rangle)]-\Pr[1\gets\cD(|\$_{k'}\rangle)]\le \sqrt{1-\frac{1}{2p}}$ for any algorithm $\cD$.
   This is because $|\langle\$_k|\$_{k'}\rangle|^2\ge\frac{1}{2p}$ for any $(k,k')\in S$.\footnote{Due to the relation between the fidelity and
   the trace distance, we have $\frac{1}{2}\||\$_k\rangle\langle\$_k|-|\$_{k'}\rangle\langle\$_{k'}|\|_1\le \sqrt{1-|\langle\$_k|\$_{k'}\rangle|^2}$,
   which means that $\langle\$_k|\Pi|\$_k\rangle-\langle\$_{k'}|\Pi|\$_{k'}\rangle\le\sqrt{1-|\langle\$_k|\$_{k'}\rangle|^2}$
   for any POVM element $\Pi$.}
   Moreover, in the second inequality, we have used the fact that 
   $\Pr[\top\gets\mathsf{QM}.\Ver(k,|\$_k\rangle)]\ge1-\frac{1}{8p}$ for any $k\in T$.
   Finally, in the last inequality, we have used the Bernoulli's inequality.\footnote{$(1+x)^r\le 1+rx$ for any real $r$ and $x$ such that $0\le r\le1$ and $x\ge-1$.}

  Let us take $\ell\ge \max(16p(t+1),16^2p^3)$.
  Then, for any $(k,k')$ such that $(k,k')\in S$ and $k\in T$,
    \begin{eqnarray*}
    \Pr[\mathsf{Count}(k,|\$_{k'}\rangle^{\otimes \ell})\ge t+1]
    &=&
    \Pr\Big[\sum_{j=1}^\ell v_j\ge t+1\Big]\\
    &\ge&
    \Pr\Big[\sum_{j=1}^\ell v_j\ge \frac{\ell}{16p}\Big]\\
    &=&
    \Pr\Big[\sum_{j=1}^\ell v_j\ge \frac{\ell}{8p} -\frac{\ell}{16p}\Big]\\
    &\ge&
   \Pr\Big[\sum_{j=1}^\ell v_j\ge \mathbb{E}(\sum_{j=1}^\ell v_j)-\frac{\ell}{16p}\Big]\\
   &\ge&1-2\exp\Big[-\frac{2\ell}{16^2p^2}\Big]\\
   &\ge&1-2e^{-2p}.
   \end{eqnarray*}
   Here, in the third inequality, we have used Hoeffding's inequality.
The probability that $\cB$ breaks the security of the quantum money scheme is therefore
    \begin{eqnarray*} 
    &&\sum_{k,k'}\Pr[k\gets\mathsf{QM}.\KeyGen(1^\secp)]\Pr[k'\gets\cA(|\$_k\rangle^{\otimes t})]
    \Pr[\mathsf{Count}(k,|\$_{k'}\rangle^{\otimes \ell})\ge t+1]\\
    &\ge&\sum_{(k,k')\in S\wedge k\in T}\Pr[k\gets\mathsf{QM}.\KeyGen(1^\secp)]\Pr[k'\gets\cA(|\$_k\rangle^{\otimes t})]
    \Pr[\mathsf{Count}(k,|\$_{k'}\rangle^{\otimes \ell})\ge t+1]\\
    &\ge&
   (1-2e^{-2p})
    \sum_{(k,k')\in S\wedge k\in T}\Pr[k\gets\mathsf{QM}.\KeyGen(1^\secp)]\Pr[k'\gets\cA(|\$_k\rangle^{\otimes t})]\\
   &\ge&(1-2e^{-2p})
  \Big(\frac{1}{2p}-\negl(\secp)\Big),
   \end{eqnarray*} 
   which is non-negligible. The $\cB$ therefore breaks the security of the private-key quantum money scheme.
\end{proof}

\subsection{OWSGs from Quantum Money with Symmetric Verifiability}
\label{sec:OWSGfromQmoney_symmetric}

We consider the following restriction for quantum money. 

\begin{definition}[Symmetric-verifiability]
We say that a private-key quantum money scheme satisfies the symmetric-verifiability
if $\Pr[\top\gets\Ver(k,\$_{k'})]=\Pr[\top\gets\Ver(k',\$_{k})]$
for all $k\neq k'$.
\end{definition}

\begin{remark}
For example, if all money states are pure, and $\Ver(\alpha,\rho)$ is the following algorithm,
the symmetric-verifiability is satisfied:
Measure $\rho$ with the basis $\{|\$_{\alpha}\rangle\langle\$_\alpha|,I-|\$_\alpha\rangle\langle\$_\alpha|\}$.
If the first result is obtained, output $\top$. Otherwise, output $\bot$.
\end{remark}

\begin{theorem}\label{thm:money_symmetric}
If private-key quantum money schemes with symmetric-verifiability
exist, then OWSGs exist.
\end{theorem}

\begin{remark}
By using the equivalence between OWSGs and wOWSGs (\cref{thm:amplification_OWSG}),
this result can be improved to a stronger result (with a similar proof) that
private-key quantum money schemes with symmetric-verifiability and with weak security imply OWSGs.
Here, the weak security means that
there exists a polynomial $p$ such that for any QPT adversary $\cA$ and any polynomial $t$,
\begin{eqnarray*}
\Pr[\mathsf{Count}(k,\xi)\ge t+1
:k\gets\KeyGen(1^\secp),\$_k\gets\Mint(k),\xi\gets\cA(1^\secp,\$_k^{\otimes t})]
\le1-\frac{1}{p}.
\end{eqnarray*}
\end{remark}

\ifnum\submission=1
The proof of \cref{thm:money_symmetric} is similar to that of \cref{thm:money_pure}.
For a proof, see the full version.
\else
\begin{proof}[Proof of \cref{thm:money_symmetric}]
Let $(\mathsf{QM}.\KeyGen,\mathsf{QM}.\Mint,\mathsf{QM}.\Ver)$ be a private-key quantum money scheme with the symmetric verifiability.
From it, we construct a OWSG as follows.
\begin{itemize}
    \item 
    $\KeyGen(1^\secp)\to k:$
    Run $k\gets\mathsf{QM}.\KeyGen(1^\secp)$.
    Output $k$.
    \item
    $\StateGen(k)\to\phi_k:$
    Run $\$_k\gets\mathsf{QM}.\Mint(k)$.
    Output $\phi_k\coloneqq \$_k$.
    \item
    $\Ver(k',\phi_k)\to\top/\bot:$
    Parse $\phi_k=\$_k$.
    Run $v\gets\mathsf{QM}.\Ver(k',\$_k)$.
    Output $v$. 
\end{itemize}
The correctness is clear. Let us show the security. Assume that it is not secure.
Then, there exists a QPT adversary $\cA$, a polynomial $t$, and a polynomial $p$ such that
\begin{eqnarray*}
\sum_{k,k'}
\Pr[k\gets\mathsf{QM}.\KeyGen(1^\secp)]
    \Pr[k'\gets\cA(\$_k^{\otimes t})]
    \Pr[\top\gets\mathsf{QM}.\Ver(k',\$_k)]
    \ge\frac{1}{p}.
    \end{eqnarray*}
    
Define the set
    \begin{eqnarray*}
    S\coloneqq\Big\{(k,k')~\Big|~\Pr[\top\gets\mathsf{QM}.\Ver(k',\$_k)]\ge
    \frac{1}{2p}\Big\}.
    \end{eqnarray*}
    Then, we have
    \begin{eqnarray*}
    \sum_{(k,k')\in S}\Pr[k\gets\mathsf{QM}.\KeyGen(1^\secp)]\Pr[k'\gets\cA(\$_k^{\otimes t})]> \frac{1}{2p}.
    \end{eqnarray*}
    This is shown as follows.
   \begin{eqnarray*} 
   \frac{1}{p}&\le&
    \sum_{k,k'}\Pr[k\gets\mathsf{QM}.\KeyGen(1^\secp)]\Pr[k'\gets\cA(\$_k^{\otimes t})]
    \Pr[\top\gets\mathsf{QM}.\Ver(k',\$_k)]\\
    &=&\sum_{(k,k')\in S}\Pr[k\gets\mathsf{QM}.\KeyGen(1^\secp)]\Pr[k'\gets\cA(\$_k^{\otimes t})]
    \Pr[\top\gets\mathsf{QM}.\Ver(k',\$_k)]\\
    &&+\sum_{(k,k')\notin S}\Pr[k\gets\mathsf{QM}.\KeyGen(1^\secp)]\Pr[k'\gets\cA(\$_k^{\otimes t})]
    \Pr[\top\gets\mathsf{QM}.\Ver(k',\$_k)]\\
    &<&
    \sum_{(k,k')\in S}\Pr[k\gets\mathsf{QM}.\KeyGen(1^\secp)]\Pr[k'\gets\cA(\$_k^{\otimes t})]
    +\frac{1}{2p}.
   \end{eqnarray*} 
   
     From the $\cA$, we construct a QPT adversary $\cB$ that breaks the security of the private-key quantum money 
    scheme as follows:
    On input $\$_k^{\otimes t}$, it runs $k'\gets\cA(\$_k^{\otimes t})$.
    It then runs $\$_{k'}\gets\mathsf{QM}.\Mint(k')$ $\ell$ times, where 
  $\ell$ is a polynomial specified later,
    and outputs $\xi\coloneqq \$_{k'}^{\otimes \ell}$.
    Let us show that thus defined $\cB$ breaks the security of the private-key quantum money scheme.
  Let 
   $v_j$ be the bit that is 1 if the output of $\mathsf{QM}.\Ver(k,\xi_j)$ is $\top$, and
   is 0 otherwise.
   Then, for any $(k,k')\in S$,
   \begin{eqnarray*}
   \Pr[v_j=1]
   &=&\Pr[\top\gets\mathsf{QM}.\Ver(k,\xi_j)]\\
   &=&\Pr[\top\gets\mathsf{QM}.\Ver(k,\$_{k'})]\\
   &=&\Pr[\top\gets\mathsf{QM}.\Ver(k',\$_k)]\\
   &\ge&\frac{1}{2p}
   \end{eqnarray*}
   for each $j\in[1,2,...,\ell]$.
  Let us take $\ell\ge \max(4p(t+1),16p^3)$.
  Then, for any $(k,k')\in S$,
    \begin{eqnarray*}
    \Pr[\mathsf{Count}(k,\$_{k'}^{\otimes \ell})\ge t+1]
    &=&
    \Pr\Big[\sum_{j=1}^\ell v_j\ge t+1\Big]\\
    &\ge&
    \Pr\Big[\sum_{j=1}^\ell v_j\ge \frac{\ell}{4p}\Big]\\
    &=&
    \Pr\Big[\sum_{j=1}^\ell v_j\ge \frac{\ell}{2p} -\frac{\ell}{4p}\Big]\\
    &\ge&
   \Pr\Big[\sum_{j=1}^\ell v_j\ge \mathbb{E}(\sum_{j=1}^\ell v_j)-\frac{\ell}{4p}\Big]\\
   &\ge&1-2\exp\Big[-\frac{2\ell}{16p^2}\Big]\\
   &\ge&1-2e^{-2p}.
   \end{eqnarray*}
   
The probability that the $\cB$ breaks the security of the quantum money scheme is therefore
    \begin{eqnarray*} 
    &&\sum_{k,k'}\Pr[k\gets\mathsf{QM}.\KeyGen(1^\secp)]\Pr[k'\gets\cA(\$_k^{\otimes t})]
    \Pr[\mathsf{Count}(k,\$_{k'}^{\otimes \ell})\ge t+1]\\
    &\ge&\sum_{(k,k')\in S}\Pr[k\gets\mathsf{QM}.\KeyGen(1^\secp)]\Pr[k'\gets\cA(\$_k^{\otimes t})]
    \Pr[\mathsf{Count}(k,\$_{k'}^{\otimes \ell})\ge t+1]\\
    &\ge&
   (1-2e^{-2p})
    \sum_{(k,k')\in S}\Pr[k\gets\mathsf{QM}.\KeyGen(1^\secp)]\Pr[k'\gets\cA(\$_k^{\otimes t})]\\
   &\ge&(1-2e^{-2p})
  \frac{1}{2p},
   \end{eqnarray*} 
   which is non-negligible. The $\cB$ therefore breaks the security of the private-key quantum money scheme.
   \end{proof}
\fi

\section{QPOTP}
\label{sec:QPOTP}
In this section, we first define (IND-based) QPOTP schemes (\cref{sec:QPOTP_def}).
We then show that QPOTP schemes imply OWSGs (\cref{sec:OWSGfromQPOTP}),
and that single-copy-secure QPOTP schemes imply EFI pairs (\cref{sec:EFIfromQPOTP}).

\subsection{Definition of QPOTP}
\label{sec:QPOTP_def}

Quantum pseudo one-time pad schemes are defined as follows.
\begin{definition}[(IND-based) quantum pseudo one-time pad (QPOTP)]
\label{def:OTP}
An (IND-based) quantum pseudo one-time pad (QPOTP) scheme with the key length $\kappa$ and
the plaintext length $\ell$ ($\ell>\kappa$) is a set of algorithms $(\KeyGen,\Enc,\Dec)$ such that 
\begin{itemize}
\item 
$\KeyGen(1^\secp)\to \sk:$
It is a QPT algorithm that, on input the security parameter $\secp$, outputs a classical secret key $\sk\in\bit^\kappa$.
    \item 
    $\Enc(\sk,x)\to\ct:$ It is a QPT algorithm that, on input $\sk$ and 
    a classical plaintext message $x\in\bit^\ell$,
    outputs an $\ell n$-qubit quantum ciphertext $\ct$.
    \item
    $\Dec(\sk,\ct)\to x':$ It is a QPT algorithm that, on input $\sk$ and $\ct$, outputs $x'\in\bit^\ell$.
\end{itemize}
We require the following correctness and security.

\paragraph{\bf Correctness:}
For any $x\in\bit^\ell$,
\begin{eqnarray*}
\Pr[x\gets\Dec(\sk,\ct):\sk\gets\KeyGen(1^\secp),\ct\gets\Enc(\sk,x)]\ge1-\negl(\secp).
\end{eqnarray*}
\paragraph{\bf Security:}
For any $x_0,x_1\in\bit^\ell$, any QPT adversary $\cA$, and any polynomial $t$,
\begin{eqnarray*}
&&|\Pr[1\gets\cA(\ct_0^{\otimes t}):\sk\gets\KeyGen(1^\secp),\ct_0\gets\Enc(\sk,x_0)]\\
&&-\Pr[1\gets\cA(\ct_1^{\otimes t}):\sk\gets\KeyGen(1^\secp),\ct_1\gets\Enc(\sk,x_1)]|
\le\negl(\secp).
\end{eqnarray*}
\end{definition}

\begin{definition}
We say that a QPOTP scheme is single-copy-secure if the security holds only for $t=1$.
\end{definition}

\begin{remark}
Note that the above definition of QPOTP is different from that of \cite{C:AnaQiaYue22} in the following two points.
First, we consider a general secret key generation QPT algorithm, while they consider uniform sampling of the secret key.
Second, we consider the IND-based version of the security, while
the security definition of \cite{C:AnaQiaYue22} is as follows:
For any $x\in\bit^\ell$, any QPT adversary $\cA$, and any polynomial $t$,
\begin{eqnarray*}
&&|\Pr[1\gets\cA(\ct^{\otimes t}):\sk\gets\bit^\kappa,\ct\gets\Enc(\sk,x)]\\
&&-\Pr[1\gets\cA((|\psi_1\rangle\otimes...\otimes|\psi_\ell\rangle)^{\otimes t})
:|\psi_1\rangle,...,|\psi_\ell\rangle\gets\mu_n]|
\le\negl(\secp),
\end{eqnarray*}
where $|\psi\rangle\gets\mu_n$ means the Haar random sampling of $n$-qubit states.
It is clear that the security definition of \cite{C:AnaQiaYue22} implies
our IND-based security, and therefore if QPOTP schemes of \cite{C:AnaQiaYue22} exist,
those of \cref{def:OTP} exist. Since our results are constructions of OWSGs and EFI pairs from QPOTP, the above modification only makes our results stronger. 
\end{remark}

\begin{remark}
QPOTP is constructed from PRSGs~\cite{C:AnaQiaYue22}.
\end{remark}

\subsection{OWSGs from QPOTP}
\label{sec:OWSGfromQPOTP}

\begin{theorem}\label{thm:QPOTP_OWSG}
If QPOTP schemes with $\kappa<\ell$ exist, then OWSGs exist.
\end{theorem}

\begin{proof}[Proof of \cref{thm:QPOTP_OWSG}]
Let $(\mathsf{OTP}.\KeyGen,\mathsf{OTP}.\Enc,\mathsf{OTP}.\Dec)$ be a QPOTP scheme
with $\kappa<\ell$. From it, we construct a wOWSG as follows.\footnote{
A similar proof idea was given in Lemma 4.6 of \cite{SIAM:GGKT05}. However, the direct application of the proof will not work, 
because ciphertexts (and therefore output states of OWSGs) are quantum and the verification of "preimages" is done by the additional verification algorithm.}
(From \cref{thm:amplification_OWSG}, it is enough for the existence of OWSGs.)
\begin{itemize}
    \item 
    $\KeyGen(1^\secp)\to k:$
    Run $\sk\gets\mathsf{OTP}.\KeyGen(1^\secp)$.
    Choose $x\gets\bit^\ell$.
    Output $k\coloneqq(\sk,x)$.
    \item
    $\StateGen(k)\to\phi_k:$
    Parse $k=(\sk,x)$.
    Run $\ct_{\sk,x}\gets\mathsf{OTP}.\Enc(\sk,x)$.
    Output $\phi_k\coloneqq \ct_{\sk,x}\otimes|x\rangle\langle x|$.
    \item
    $\Ver(k',\phi_k)\to\top/\bot:$
    Parse $k'=(\sk',x')$.
    Parse $\phi_k=\ct_{\sk,x}\otimes|x\rangle\langle x|$.
    Run $x''\gets\mathsf{OTP}.\Dec(\sk',\ct_{\sk,x})$.
    If $x''=x'=x$, output $\top$. Otherwise, output $\bot$.
\end{itemize}

The correctness is clear. Let us show the security. Assume that it is not secure.
It means that for any polynomial $p$ there exist a QPT adversary $\cA$ and a polynomial $t$ such that
\begin{eqnarray}
\Pr
\left[
    x'= x''=x:
    \begin{array}{ll}
    \sk\gets\mathsf{OTP}.\KeyGen(1^\secp),\\
    x\gets\bit^\ell,\\
    \ct_{\sk,x}\gets\mathsf{OTP}.\Enc(\sk,x),\\
    (\sk',x')\gets\cA(\ct_{\sk,x}^{\otimes t}\otimes |x\rangle\langle x|^{\otimes t})\\
    x''\gets\mathsf{OTP}.\Dec(\sk',\ct_{\sk,x})
    \end{array}
    \right]\ge1-\frac{1}{p}.
    \label{eq:OWSGfromQPOTP1}
    \end{eqnarray}
    From this $\cA$, we construct a QPT adversary $\cB$ that breaks the security of the QPOTP scheme as follows.
    Let $b\in\bit$ be the parameter of the following security game.
   \begin{enumerate}
       \item 
       $\cB$ chooses $x_0,x_1\gets\bit^\ell$, and sends them to the challenger $\cC$.
       \item
       $\cC$ runs $\sk\gets\mathsf{OTP}.\KeyGen(1^\secp)$.
       \item
       $\cC$ runs $\ct_{\sk,x_b}\gets\mathsf{OTP}.\Enc(\sk,x_b)$ $t+1$ times.
       \item
       $\cC$ sends $\ct_{\sk,x_b}^{\otimes t+1}$ to $\cB$.
       \item
      $\cB$ runs $(\sk',x')\gets\cA(\ct_{\sk,x_b}^{\otimes t}\otimes|x_0\rangle\langle x_0|^{\otimes t})$. 
      \item
     $\cB$ runs $x''\gets\mathsf{OTP}.\Dec(\sk',\ct_{k,x_b})$. 
     If $x'=x''=x_0$, $\cB$ outputs $b'=0$.
     Otherwise, it outputs $b'=1$.
   \end{enumerate} 
   It is clear that $\Pr[b'=0|b=0]$ is equivalent to the left-hand-side of Eq.~(\ref{eq:OWSGfromQPOTP1}).
   On the other hand,
   \begin{eqnarray*}
&&   \Pr[b'=0|b=1]\\
&=&
\frac{1}{2^{2\ell}}\sum_{x_0,x_1,\sk,\sk'}
   \Pr[\sk\gets\mathsf{OTP}.\KeyGen(1^\secp)]
   \Pr[\sk'\gets\cA(\ct_{\sk,x_1}^{\otimes t}\otimes|x_0\rangle\langle x_0|^{\otimes t})]\\
&&\times    \Pr[x_0 \gets\mathsf{OTP}.\Dec(\sk',\ct_{\sk,x_1})]\\
&\le&\frac{1}{2^{2\ell}}\sum_{x_0,x_1,\sk,\sk'}
   \Pr[\sk\gets\mathsf{OTP}.\KeyGen(1^\secp)]
   \Pr[x_0 \gets\mathsf{OTP}.\Dec(\sk',\ct_{\sk,x_1})]\\
 &=&\frac{1}{2^{2\ell}}\sum_{x_1,\sk,\sk'}
   \Pr[\sk\gets\mathsf{OTP}.\KeyGen(1^\secp)]
   \sum_{x_0}\Pr[x_0 \gets\mathsf{OTP}.\Dec(\sk',\ct_{\sk,x_1})]\\
  &=&\frac{1}{2^{2\ell}}\sum_{x_1,\sk,\sk'}
   \Pr[\sk\gets\mathsf{OTP}.\KeyGen(1^\secp)]\\
&=&\frac{2^\kappa}{2^\ell}
\le\frac{1}{2}.
   \end{eqnarray*}
   Therefore 
$|\Pr[b'=0|b=0]-\Pr[b'=0|b=1]|$ is non-negligible, which means that the $\cB$ breaks the security of
the QPOTP.
\end{proof}

\subsection{EFI Pairs from Single-Copy-Secure QPOTP}
\label{sec:EFIfromQPOTP}

\ifnum\submission=1
\else
We will use the following result, which is (implicitly) shown in \cite{LC19}.\footnote{See Theorem 4 of \cite{LC19}.
$\rho_0$ and $\rho_1$ correspond to $\rho_{MC}$ and $\sigma_{MC}$, respectively. Moreover,
Take $\epsilon=\gamma=\negl(\secp)$.}
\begin{theorem}[\cite{LC19}]
\label{theorem:LC19}
Let $\KeyGen(1^\secp)\to k$
be an algorithm that, on input the security parameter $\secp$,
outputs a classical secret key $k\in\bit^\kappa$.
Let $\{\Enc^k,\Dec^k\}_k$ be quantum operations. 
Assume that the following is satisfied:
For any $\rho_{\regA,\regB}$,
\begin{eqnarray}
\frac{1}{2}\Big\|
\sum_k\Pr[k]
(\Dec^k_{\regA}\otimes I_{\regB})
(\Enc^k_{\regA}\otimes I_{\regB})
(\rho_{\regA,\regB})-\rho_{\regA,\regB}\Big\|_1\le \negl(\secp),
\label{SVOWSG_ITS_correctness}
\end{eqnarray}
where $\Pr[k]\coloneqq \Pr[k\gets\KeyGen(1^\secp)]$.
Let us define 
\begin{eqnarray*}
\rho_0&\coloneqq&\sum_{k}\Pr[k]
(\Enc^k_{\regA}\otimes I_{\regB})|\Psi\rangle\langle\Psi|_{\regA,\regB},\\
\rho_1&\coloneqq&\sum_{k}\Pr[k]
(\Enc^k_{\regA}\otimes I_{\regB})\Big(|0^n\rangle\langle0^n|_\regA\otimes\frac{I^{\otimes n}}{2^n}_{\regB}\Big),
\end{eqnarray*}
where $|\Psi\rangle_{\regA,\regB}$ is the maximally entangled state over $n$-qubit registers $\regA$ and $\regB$.
If $\frac{1}{2}\|\rho_0-\rho_1\|_1\le\negl(\secp)$,
then $\kappa\ge 2n+\log(1-\negl(\secp))$. 
\end{theorem}
\fi

\begin{theorem}\label{thm:QPOTP_EFI}
If single-copy-secure QPOTP schemes with $\kappa< \ell$ exist then EFI pairs exist.  
\end{theorem}

\ifnum\submission=1
We prove this theorem based on a result shown by Lai and Chung~\cite{LC19}, which gives a quantum analogue of Shannon's impossibility. Roughly speaking, they show that if a SKE scheme for $n$-qubit messages and $\kappa$-bit secret keys is information theoretically one-time-secure, then we must have $\kappa \geq 2n$. By a reduction to their result via a hybrid encryption of QPOTP and quantum one-time pads, we can show that any QPOTP scheme with $\kappa< \ell$ is \emph{not} one-time-secure against unbounded-time adversaries. On the other hand, we assume that it is one-time-secure against QPT adversaries. This computationally-secure and information-theoretically-insecure encryption scheme can be directly used to construct EFI pairs. 
For a formal proof, see the full version.
\else
\begin{proof}[Proof of \cref{thm:QPOTP_EFI}]
Let $(\mathsf{OTP}.\KeyGen,\mathsf{OTP}.\Enc,\mathsf{OTP}.\Dec)$ be a single-copy-secure 
QPOTP scheme with $\kappa<\ell$.
From it, we construct an EFI pair $\StateGen(1^\secp,b)\to\rho_b$ as a QPT algorithm that outputs
\begin{eqnarray*}
\rho_0&\coloneqq&\frac{1}{2^{2n}}\sum_{\sk,x,z}\Pr[\sk]\ct_{\sk,(x,z)}\otimes
\Big[(X^xZ^z\otimes I)|\Psi\rangle\langle\Psi|(X^xZ^z\otimes I)\Big],\\
\rho_1&\coloneqq&\frac{1}{2^{2n}}\sum_{\sk,x,z}\Pr[\sk]\ct_{\sk,(x,z)}\otimes
\Big[(X^xZ^z\otimes I)\Big(|0^n\rangle\langle0^n|\otimes\frac{I^{\otimes n}}{2^n}\Big)|(X^xZ^z\otimes I)\Big],
\end{eqnarray*}
where 
$\Pr[\sk]\coloneqq \Pr[\sk\gets\mathsf{OTP}.\KeyGen(1^\secp)]$,
$\ct_{\sk,(x,z)}\gets \mathsf{OTP}.\Enc(\sk,(x,z))$, 
$x,z\in\bit^n$,
$X^x\coloneqq \bigotimes_{i=1}^n X_i^{x_i}$,
$Z^z\coloneqq \bigotimes_{i=1}^n Z_i^{z_i}$,
$2n=\ell$ and $|\Psi\rangle$ is the maximally-entangled state on two $n$-qubit registers.
It is clear that $\rho_0$ and $\rho_1$ can be generated in QPT in a natural way.

First let us show the computational indistinguishability of $\rho_0$ and $\rho_1$.
Let us define
\begin{eqnarray*}
\rho_0'&\coloneqq&\frac{1}{2^{2n}}\sum_{\sk,x,z}\Pr[\sk]\ct_{\sk,(0^n,0^n)}\otimes
\Big[(X^xZ^z\otimes I)|\Psi\rangle\langle\Psi|(X^xZ^z\otimes I)\Big],\\
\rho_1'&\coloneqq&\frac{1}{2^{2n}}\sum_{\sk,x,z}\Pr[\sk]\ct_{\sk,(0^n,0^n)}\otimes
\Big[(X^xZ^z\otimes I)\Big(|0^n\rangle\langle0^n|\otimes\frac{I^{\otimes n}}{2^n}\Big)(X^xZ^z\otimes I)\Big].
\end{eqnarray*}
Then, for any QPT adversary $\cA$,
from the security of the QPOTP scheme,
\begin{eqnarray*}
&&|\Pr[1\gets\cA(\rho_0)]
-\Pr[1\gets\cA(\rho_0')]|\\
&\le&
\frac{1}{2^{2n}}\sum_{x,z}\Big|
\Pr\Big[1\gets\cA\Big(\sum_\sk\Pr[\sk]\ct_{\sk,(x,z)}\otimes
\Big[(X^xZ^z\otimes I)|\Psi\rangle\langle\Psi|(X^xZ^z\otimes I)\Big]\Big)\Big]\\
&&-\Pr\Big[1\gets\cA\Big(\sum_\sk\Pr[\sk]\ct_{\sk,(0^n,0^n)}\otimes
\Big[(X^xZ^z\otimes I)|\Psi\rangle\langle\Psi|(X^xZ^z\otimes I)\Big]\Big)\Big]
\Big|\\
&\le&\negl(\secp),
\end{eqnarray*}
and
\begin{eqnarray*}
|\Pr[1\gets\cA(\rho_1)]
-\Pr[1\gets\cA(\rho_1')]|
\le\negl(\secp)
\end{eqnarray*}
in a similar way.
Moreover, $\rho_0'=\rho_1'$, because
\begin{eqnarray*}
\rho_0'&=&\sum_{\sk}\Pr[\sk]\ct_{\sk,(0^n,0^n)}\otimes
\frac{1}{2^{2n}}\sum_{x,z}
\Big[(X^xZ^z\otimes I)|\Psi\rangle\langle\Psi|(X^xZ^z\otimes I)\Big],\\
&=&\sum_{\sk}\Pr[\sk]\ct_{\sk,(0^n,0^n)}\otimes
\frac{I^{\otimes n}}{2^{n}}
\otimes
\frac{I^{\otimes n}}{2^{n}},\\
\rho_1'&=&\sum_\sk\Pr[\sk]\ct_{\sk,(0^n,0^n)}\otimes
\frac{1}{2^{2n}}\sum_{x,z}
\Big[(X^xZ^z\otimes I)\Big(|0^n\rangle\langle0^n|\otimes\frac{I^{\otimes n}}{2^n}\Big)(X^xZ^z\otimes I)\Big]\\
&=&\sum_\sk\Pr[\sk]\ct_{\sk,(0^n,0^n)}\otimes
\frac{I^{\otimes n}}{2^{n}}
\otimes
\frac{I^{\otimes n}}{2^{n}}.
\end{eqnarray*}
Hence we have
\begin{eqnarray*}
|\Pr[1\gets\cA(\rho_0)]-\Pr[1\gets\cA(\rho_1)]|\le\negl(\secp)
\end{eqnarray*}
for any QPT adversary $\cA$.

Next let us show the statistical distinguishability of $\rho_0$ and $\rho_1$.
From the QPOTP, construct $\KeyGen$ and $\Enc^k,\Dec^k$ of \cref{theorem:LC19}
as follows.
\begin{itemize}
    \item 
    $\KeyGen(1^\secp)\to k:$
    Run $\sk\gets\mathsf{OTP}.\KeyGen(1^\secp)$.
    Output $k\coloneqq \sk$.
    \item
    $\Enc^k(\rho_{\regA,\regB})=\rho'_{\regA',\regB}:$
    Parse $k=\sk$.
    Choose $x,z\gets\bit^n$.
    Run $\ct_{\sk,(x,z)}\gets\mathsf{OTP}.\Enc(\sk,(x,z))$.
    Output 
   \begin{eqnarray*} 
    \rho_{\regA',\regB}'\coloneqq \frac{1}{2^{2n}}\sum_{x,z}
    \ct_{\sk,(x,z)}\otimes\Big([(X^xZ^z)_\regA\otimes I_\regB]
    \rho_{\regA,\regB}[(X^xZ^z)_\regA\otimes I_\regB]\Big).
   \end{eqnarray*} 
    \item
   $\Dec^k(\rho_{\regA',\regB}')=\rho_{\regA,\regB}:$ 
   Parse $k=\sk$.
   Parse
   \begin{eqnarray*}
   \rho_{\regA',\regB}'= \frac{1}{2^{2n}}\sum_{x,z}\ct_{\sk,(x,z)}\otimes\Big([(X^xZ^z)_\regA\otimes I_\regB]
   \rho_{\regA,\regB}[(X^xZ^z)_\regA\otimes I_\regB]\Big).
   \end{eqnarray*}
   Run $(x,z)\gets\mathsf{OTP}.\Dec(\sk,\ct_{\sk,(x,z)})$.
   Obtain $\rho_{\regA,\regB}$ by correcting the Pauli one-time pad.
   Output $\rho_{\regA,\regB}$.
\end{itemize}
It satisfies Eq.~(\ref{SVOWSG_ITS_correctness}).
Assume that 
$\frac{1}{2}\|\rho_0-\rho_1\|_1\le\negl(\secp)$.
Then, from \cref{theorem:LC19},
$\kappa\ge 2n+\log(1-\negl(\secp))$, which means
$\kappa\ge \ell+\log(1-\negl(\secp))$, 
but it contradicts the assumption of $\kappa<\ell$.
Therefore, the statistical distinguishability of $\rho_0$ and $\rho_1$ is shown.
\end{proof}
\fi

\section{SV-SI-OWSGs}
\label{sec:SVOWSG}

In this section, we define SV-SI-OWSGs (\cref{sec:SVOWSG_def}), and show that SV-SI-OWSGs are equivalent to EFI pairs (\cref{sec:SVOWSG_EFI}).
In \cref{sec:SVOWSG_def}, before defining SV-SI-OWSGs, we first define SV-OWSGs for a didactic purpose.
We will point out that SV-OWSGs seem to need a more constraint so that they become equivalent to EFI.
We then define SV-SI-OWSGs.

\subsection{Definition of SV-SI-OWSGs}
\label{sec:SVOWSG_def}

We first define secretly-verifiable OWSGs (SV-OWSGs) as follows.
\begin{definition}[Secretly-verifiable OWSGs (SV-OWSGs)]\label{def:SV-OWSG}
A secretly-verifiable OWSG (SV-OWSG) is a set of algorithms $(\KeyGen,\StateGen,\Ver)$ as follows.
\begin{itemize}
    \item 
    $\KeyGen(1^\secp)\to k:$
    It is a QPT algorithm that, on input the security parameter $\secp$, outputs a key $k\in\bit^\kappa$.
    \item
    $\StateGen(k)\to\phi_k:$
    It is a QPT algorithm that, on input $k$, outputs an $m$-qubit state $\phi_k$.
    \item
    $\Ver(k',k)\to\top/\bot:$
    It is a QPT algorithm that, on input $k$ and $k'$, outputs $\top/\bot$.
\end{itemize}
We require the following two properties.

\paragraph{\bf Correctness:}
\begin{eqnarray*}
\Pr[\top\gets\Ver(k,k):k\gets\KeyGen(1^\secp)]\ge 1-\negl(\secp).
\end{eqnarray*}

\paragraph{\bf Security:}
For any QPT adversary $\cA$ and any polynomial $t$, 
\begin{eqnarray*}
\Pr[\top\gets\Ver(k',k):k\gets\KeyGen(1^\secp),\phi_k\gets\StateGen(k),k'\gets\cA(\phi_k^{\otimes t})]\le \negl(\secp).
\end{eqnarray*}
\end{definition}

The following lemma shows that, without loss of generality, $\Ver$ can be replaced with the algorithm
of just checking whether $k=k'$ or not.
\begin{lemma}\label{lemma:SVOWSG_amp}
Let $(\KeyGen,\StateGen,\Ver)$ be a SV-OWSG.
Then, the following SV-OWSG $(\KeyGen',\StateGen',\Ver')$ exists.
\begin{itemize}
    \item 
    $\KeyGen'$ and $\StateGen'$ are the same as $\KeyGen$ and $\StateGen$, respectively.
    \item
    $\Ver'(k',k)\to\top/\bot:$
    On input $k$ and $k'$, output $\top$ if $k=k'$.
    Otherwise, output $\bot$.
\end{itemize}
\end{lemma}

\ifnum\submission=1
For a proof, see the full version.
\else
\begin{proof}
The correctness of $(\KeyGen',\StateGen',\Ver')$ is clear.
Let us show the security. Assume that it is not secure.
Then there exist a QPT adversary $\cA$, a polynomial $t$, and a polynomial $p$ such that
\begin{eqnarray*}
\sum_k\Pr[k\gets\KeyGen(1^\secp)]\Pr[k\gets\cA(\phi_k^{\otimes t})]\ge \frac{1}{p}.
\end{eqnarray*}
Define the set
\begin{eqnarray*}
S\coloneqq\Big\{k~\Big|~\Pr[k\gets\cA(\phi_k^{\otimes t})]\ge \frac{1}{2p}\Big\}.
\end{eqnarray*}
Then, we have
\begin{eqnarray*}
\sum_{k\in S}\Pr[k\gets\KeyGen(1^\secp)]> \frac{1}{2p}.
\end{eqnarray*}
This is because
\begin{eqnarray*}
\frac{1}{p}&\le&\sum_k\Pr[k\gets\KeyGen(1^\secp)]\Pr[k\gets\cA(\phi_k^{\otimes t})]\\
&=&\sum_{k\in S}\Pr[k\gets\KeyGen(1^\secp)]\Pr[k\gets\cA(\phi_k^{\otimes t})]\\
&&+\sum_{k\notin S}\Pr[k\gets\KeyGen(1^\secp)]\Pr[k\gets\cA(\phi_k^{\otimes t})]\\
&<&\sum_{k\in S}\Pr[k\gets\KeyGen(1^\secp)]+\frac{1}{2p}.
\end{eqnarray*}

Also define the set
\begin{eqnarray*}
T\coloneqq\Big\{k~\Big|~\Pr[\top\gets\Ver(k,k)]\ge 1-\frac{1}{p}\Big\}.
\end{eqnarray*}
Then, we have
\begin{eqnarray*}
\sum_{k\in T}\Pr[k\gets\KeyGen(1^\secp)]> 1-\negl(\secp).
\end{eqnarray*}
This is because
\begin{eqnarray*}
1-\negl(\secp)&\le&\sum_{k}\Pr[k\gets\KeyGen(1^\secp)]\Pr[\top\gets\Ver(k,k)]\\
&=&\sum_{k\in T}\Pr[k\gets\KeyGen(1^\secp)]\Pr[\top\gets\Ver(k,k)]\\
&&+\sum_{k\notin T}\Pr[k\gets\KeyGen(1^\secp)]\Pr[\top\gets\Ver(k,k)]\\
&<&\sum_{k\in T}\Pr[k\gets\KeyGen(1^\secp)]\\
&&+\Big(1-\frac{1}{p}\Big)\Big(1-\sum_{k\in T}\Pr[k\gets\KeyGen(1^\secp)]\Big).
\end{eqnarray*}
Here, the first inequality is from the correctness of $(\KeyGen,\StateGen,\Ver)$.
From the union bound, we have
\begin{eqnarray*}
\sum_{k\in S\cap T}\Pr[k\gets\KeyGen(1^\secp)]> \frac{1}{2p}-\negl(\secp).
\end{eqnarray*}

From the $\cA$, we construct a QPT adversary that breaks the security of $(\KeyGen,\StateGen,\Ver)$ as follows:
On input $\phi_k^{\otimes t}$, run $k'\gets\cA(\phi_k^{\otimes t})$.
Output $k'$.
Then, the probability that $\cB$ breaks the security is
\begin{eqnarray*}
&&\sum_{k,k'}\Pr[k\gets\KeyGen(1^\secp)]\Pr[k'\gets\cA(\phi_k^{\otimes t})]\Pr[\top\gets\Ver(k',k)]\\
&\ge&\sum_k\Pr[k\gets\KeyGen(1^\secp)]\Pr[k\gets\cA(\phi_k^{\otimes t})]\Pr[\top\gets\Ver(k,k)]\\
&\ge&\sum_{k\in S\cap T}\Pr[k\gets\KeyGen(1^\secp)]\Pr[k\gets\cA(\phi_k^{\otimes t})]\Pr[\top\gets\Ver(k,k)]\\
&\ge&\frac{1}{2p}\Big(1-\frac{1}{p}\Big)\Big(\frac{1}{2p}-\negl(\secp)\Big),
\end{eqnarray*}
which is non-negligible.
Therefore $\cB$ breaks the security of
$(\KeyGen,\StateGen,\Ver)$.
\end{proof}
\fi

Note that statistically-secure SV-OWSGs are easy to realize.
For example, consider the following construction:
\begin{itemize}
\item
$\KeyGen(1^\secp):$ Sample $k\gets\bit^\secp$.
\item
$\StateGen(k):$ Output $\frac{I^{\otimes m}}{2^m}$.
\item
$\Ver(k',k):$ Output $\top$ if $k'=k$. Otherwise, output $\bot$.
\end{itemize}
We therefore need a constraint to have a meaningful primitive.
We define secretly-verifiable and statistically-invertible OWSGs (SV-SI-OWSGs) as follows.
Introducing the statistical invertibility allows us to avoid trivial constructions with
the statistical security.

\begin{definition}[Secretly-verifiable and statistically-invertible OWSGs (SV-SI-OWSGs)]
A secretly-verifiable and statistically-invertible OWSG (SV-SI-OWSG) is a set of algorithms $(\KeyGen,\StateGen)$ as follows.
\begin{itemize}
    \item 
    $\KeyGen(1^\secp)\to k:$
    It is a QPT algorithm that, on input the security parameter $\secp$, outputs a key $k\in\bit^\kappa$.
    \item
    $\StateGen(k)\to\phi_k:$
    It is a QPT algorithm that, on input $k$, outputs an $m$-qubit state $\phi_k$.
\end{itemize}
We require the following two properties.

\paragraph{\bf Statistical invertibility:}
There exists a polynomial $p$ such that, for any $k$ and $k'$ ($k\neq k'$),
$
\frac{1}{2}\|\phi_k-\phi_{k'}\|_1\ge\frac{1}{p}.
$

\paragraph{\bf Computational non-invertibility:}
For any QPT adversary $\cA$ and any polynomial $t$, 
\begin{eqnarray*}
\Pr[k\gets\cA(\phi_k^{\otimes t}):
k\gets\KeyGen(1^\secp),\phi_k\gets\StateGen(k)]\le \negl(\secp).
\end{eqnarray*}
\end{definition}

The following lemma shows that the statistical invertibility with advantage $\frac{1}{\poly(\secp)}$
can be amplified to $1-2^{-q}$ for any polynomial $q$.
\begin{lemma}
\label{lemma:SVOWSG_amp}
If a SV-SI-OWSG exists
then
a SV-SI-OWSG with statistical invertibility larger than $1-2^{-q}$ with any polynomial $q$ exists.
\end{lemma}

\begin{proof}
Let $(\KeyGen,\StateGen)$ be a SV-SI-OWSG with statistical invertibility larger than $\frac{1}{p}$,
where $p$ is a polynomial.
From it, we construct a new SV-SI-OWSG ($\KeyGen',\StateGen'$) as follows:
\begin{itemize}
    \item 
    $\KeyGen'(1^\secp)\to k$:
    Run $k\gets\KeyGen(1^\secp)$, and output $k$.
    \item
    $\StateGen'(k)\to\phi_k':$
    Run $\phi_k\gets \StateGen(k)$ $2pq$ times, and output
    $\phi_k'\coloneqq \phi_k^{\otimes 2pq}$.
\end{itemize}
First, for any $k$ and $k'$ ($k\neq k'$),
\begin{eqnarray*}
\frac{1}{2}\|\phi_k'-\phi_{k'}'\|_1
&=&\frac{1}{2}\|\phi_k^{\otimes 2pq}-\phi_{k'}^{\otimes 2pq}\|_1\\
&\ge& 1-\exp(-2qp\|\phi_k-\phi_{k'}\|_1/4)\\
&\ge& 1-\exp(-q)\\
&\ge& 1-2^{-q},
\end{eqnarray*}
which shows the statistical invertibility of $(\KeyGen',\StateGen')$ with the advantage
larger than $1-2^{-q}$.
Second, from the computational non-invertibility of $(\KeyGen,\StateGen)$, 
\begin{eqnarray*}
&&\Pr[k\gets\cA(\phi_k'^{\otimes t}):k\gets\KeyGen'(1^\secp),\phi_k'\gets\StateGen'(k)]\\
&=&\Pr[k\gets\cA(\phi_k^{\otimes 2pqt}):k\gets\KeyGen(1^\secp),\phi_k\gets\StateGen(k)]\\
&\le& \negl(\secp)
\end{eqnarray*}
for any QPT adversary $\cA$ and any polynomial $t$, 
which shows the computational non-invertibility of $(\KeyGen',\StateGen')$.
\end{proof}

The following lemma shows that the statistical invertibility 
is equivalent to the existence of a (unbounded) adversary that can find
the correct $k$ given many copies of $\phi_k$ except for a negligible error.
\begin{lemma}
\label{lemma:SVOWSG_equivalence}
The statistical invertibility is satisfied if and only if the 
following is satisfied:
There exists a 
(not necessarily QPT) POVM measurement $\{\Pi_k\}_{k\in\bit^\kappa}$ 
and a polynomial $t$ such that
$\mbox{Tr}(\Pi_k\phi_k^{\otimes t})\ge1-\negl(\secp)$
and $\mbox{Tr}(\Pi_{k'}\phi_k^{\otimes t})\le\negl(\secp)$ for all $k$ and $k'$ ($k\neq k'$).

\end{lemma}

\begin{proof}
First, we show the if part.
Assume that
there exists a 
POVM measurement $\{\Pi_k\}_{k\in\bit^\kappa}$ and a polynomial $t$ such that
$\mbox{Tr}(\Pi_k\phi_k^{\otimes t})\ge1-\negl(\secp)$
and $\mbox{Tr}(\Pi_{k'}\phi_k^{\otimes t})\le\negl(\secp)$ for all $k$ and $k'$ ($k\neq k'$).
Then,
\begin{eqnarray*}
\frac{t}{2}\|\phi_k-\phi_{k'}\|_1
&\ge&
\frac{1}{2}\|\phi_k^{\otimes t}-\phi_{k'}^{\otimes t}\|_1\\
&\ge&
\Tr(\Pi_k \phi_k^{\otimes t})
-\Tr(\Pi_k \phi_{k'}^{\otimes t})\\
&\ge&
1-\negl(\secp)
-\negl(\secp)\\
&=&1-\negl(\secp),
\end{eqnarray*}
which means
\begin{eqnarray*}
\frac{1}{2}\|\phi_k-\phi_{k'}\|_1
\ge
\frac{1}{t}-\negl(\secp)
\ge
\frac{1}{2t}.
\end{eqnarray*}

Next, we show the only if part.
Assume that the statistical invertibility is satisfied.
Then, there exists a polynomial $p$ such that
$\frac{1}{2}\|\phi_k-\phi_{k'}\|_1\ge\frac{1}{p}$
for all $k$ and $k'$ ($k\neq k'$).
Let $t\coloneqq 12p\kappa$. Then,
\begin{eqnarray*}
\frac{1}{2}\|\phi_k^{\otimes t}-\phi_{k'}^{\otimes t}\|_1
\ge
1-e^{-t\frac{\|\phi_k-\phi_{k'}\|_1}{4}}
\ge
1-e^{-6\kappa}
\ge
1-2^{-6\kappa},
\end{eqnarray*}
which 
means
$F(\phi_k^{\otimes t},\phi_{k'}^{\otimes t})\le 2^{-6\kappa+1}$.
From \cref{theorem:Montanaro} below,
\begin{eqnarray*}
\max_k(1-\mbox{Tr}(\mu_k\phi_k^{\otimes t}))
\le \sum_{k\neq k'}\sqrt{F(\phi_k^{\otimes t},\phi_{k'}^{\otimes t})}
\le 2^{-3\kappa+1}(2^{2\kappa}-2^\kappa)
\le 2^{-\kappa+1},
\end{eqnarray*}
which means $\mbox{Tr}(\mu_k \phi_k^{\otimes t})\ge1-2^{-\kappa+1}$ and
$\mbox{Tr}(\mu_{k'} \phi_k^{\otimes t})\le 2^{-\kappa+1}$ for any $k$ and $k'$ ($k'\neq k$).
\end{proof}

\begin{theorem}[\cite{Montanaro19}]
\label{theorem:Montanaro}
Let $\{\rho_i\}_i$ be a set of states.
Define the POVM measurement $\{\mu_i\}_i$ with
$\mu_i\coloneqq\Sigma^{-1/2}\rho_i\Sigma^{-1/2}$,
where $\Sigma\coloneqq \sum_i\rho_i$,
and the inverse is taken on the support of $\Sigma$.
Then,
$\max_i(1-\mbox{Tr}(\mu_i\rho_i))\le\sum_{i\neq j}\sqrt{F(\rho_i,\rho_j)}$.
\end{theorem}

\subsection{Equivalence of SV-SI-OWSGs and EFI Pairs}
\label{sec:SVOWSG_EFI}

\begin{theorem}\label{thm:SVSIOWSG_EFI}
SV-SI-OWSGs exist if and only if EFI pairs exist.
\end{theorem}
This Theorem is shown by combining the following two theorems.

\begin{theorem}\label{thm:SVSIOWSGfromEFI}
If EFI pairs exist then SV-SI-OWSGs exist.
\end{theorem}

\begin{theorem}\label{thm:EFIfromSVSIOWSG}
If 
SV-SI-OWSGs exist
then EFI pairs exist. 
\end{theorem}

\begin{proof}[Proof of \cref{thm:SVSIOWSGfromEFI}]
We show that if
EFI pairs exist then SV-SI-OWSGs exist.
Let $\mathsf{EFI}.\StateGen(1^\secp,b)\to\rho_b$ be an EFI pair.
As is explained in \cref{remark:EFI}, we can assume without loss of generality that
$\frac{1}{2}\|\rho_0-\rho_1\|_1\ge1-\negl(\secp)$,
which means $F(\rho_0,\rho_1)\le\negl(\secp)$.
From the EFI pair, we construct a SV-SI-OWSG as follows.
\begin{itemize}
    \item 
    $\KeyGen(1^\secp)\to k:$
    Choose $k\gets\bit^\kappa$, and output $k$.
    \item
    $\StateGen(k)\to\phi_k:$
    Run $\mathsf{EFI}.\StateGen(1^\secp,k_i)\to\rho_{k_i}$ for each $i\in[\kappa]$.
    Output $\phi_k\coloneqq \bigotimes_{i=1}^\kappa \rho_{k_i}$.
\end{itemize}

The statistical invertibility is easily shown as follows. 
If $k\neq k'$, there exists a $j\in[\kappa]$ such that $k_j\neq k_j'$.
Then,
\begin{eqnarray*}
F(\phi_k,\phi_{k'})&=&\prod_{i=1}^\kappa F(\rho_{k_i},\rho_{k_i'})
\le F(\rho_{k_j},\rho_{k_j'})
\le\negl(\secp),
\end{eqnarray*}
which means $\frac{1}{2}\|\phi_k-\phi_{k'}\|_1\ge1-\negl(\secp)$.
This shows the statistical invertibility.

Let us next show the computational non-invertibility.
From the standard hybrid argument, and the computational indistinguishability of $\rho_0$ and $\rho_1$, we have
\begin{eqnarray}
\Big|\frac{1}{2^\kappa}\sum_{k\in\bit^\kappa}
\Pr[k\gets\cA(\phi_k^{\otimes t})]-
\frac{1}{2^\kappa}\sum_{k\in\bit^\kappa}
\Pr[k\gets\cA(\phi_{0^\kappa}^{\otimes t})]\Big|\le\negl(\secp)
\label{SVSIOWSGsfromEFI_hybrid}
\end{eqnarray}
for any QPT adversary $\cA$ and any polynomial $t$.
(It will be shown later.)
Hence
\begin{eqnarray*}
&&\Pr[k\gets\cA(\phi_k^{\otimes t}):k\gets\KeyGen(1^\secp),\phi_k\gets\StateGen(k)]\\
&=&\frac{1}{2^\kappa}\sum_{k\in\bit^\kappa}\Pr[k\gets\cA(\phi_k^{\otimes t})]\\
&\le&\frac{1}{2^\kappa}\sum_{k\in\bit^\kappa}\Pr[k\gets\cA(\phi_{0^\kappa}^{\otimes t})]
+\negl(\secp)\\
&=&\frac{1}{2^\kappa}+\negl(\secp),
\end{eqnarray*}
which shows the computational non-invertibility.

Let us show Eq.~(\ref{SVSIOWSGsfromEFI_hybrid}).
For each $z\in\bit^{\kappa t}$,
define $\Phi_z\coloneqq\bigotimes_{i=1}^{\kappa t}\rho_{z_i}$.
Let $z,z'\in\bit^{\kappa t}$ be two bit strings 
such that, for a single $j\in[\kappa t]$, $z_j=0$, $z'_j=1$, and $z_i=z'_i$ for all $i\neq j$.
(In other words, $z$ and $z'$ are the same except for the $j$th bit.)
Then, we can show that
\begin{eqnarray}
\Big|\frac{1}{2^\kappa}\sum_{k\in\bit^\kappa}\Pr[k\gets\cA(\Phi_z)]
-\frac{1}{2^\kappa}\sum_{k\in\bit^\kappa}\Pr[k\gets\cA(\Phi_{z'})]\Big|\le\negl(\secp)
\label{SVSIOWSGsfromEFI_hybrid2}
\end{eqnarray}
for any QPT adversary $\cA$. 
In fact, assume that
\begin{eqnarray*}
\Big|\frac{1}{2^\kappa}\sum_k\Pr[k\gets\cA(\Phi_z)]
-\frac{1}{2^\kappa}\sum_k\Pr[k\gets\cA(\Phi_{z'})]\Big|\ge \frac{1}{\poly(\secp)}
\end{eqnarray*}
for a QPT adversary $\cA$. 
Then, from this $\cA$, we can construct a QPT adversary $\cB$ that breaks the security of
the EFI pair as follows:
On input $\rho_b$, choose $k\gets\bit^\kappa$,
and run $k'\gets \cA((\bigotimes_{i=1}^{j-1}\rho_{z_i})\otimes \rho_b \otimes(\bigotimes_{i=j+1}^{\kappa t}\rho_{z_i}))$.
If $k'=k$, output $b'=1$.
If $k'\neq k$, output $b'=0$.
Because
\begin{eqnarray*}
\Pr[b'=1|b=0]&=&\frac{1}{2^\kappa}\sum_k\Pr[k\gets\cA(\Phi_z)],\\
\Pr[b'=1|b=1]&=&\frac{1}{2^\kappa}\sum_k\Pr[k\gets\cA(\Phi_{z'})],
\end{eqnarray*}
we have
\begin{eqnarray*}
|\Pr[b'=1|b=0]-
\Pr[b'=1|b=1]|\ge\frac{1}{\poly(\secp)},
\end{eqnarray*}
which means that
the $\cB$ breaks the security of the EFI pair.
From the standard hybrid argument and
Eq.~(\ref{SVSIOWSGsfromEFI_hybrid2}),
we have Eq.~(\ref{SVSIOWSGsfromEFI_hybrid}).
\end{proof}

\begin{proof}[Proof of \cref{thm:EFIfromSVSIOWSG}]
We show that
if SV-SI-OWSGs exist then EFI pairs exist.
Let $(\mathsf{OWSG}.\KeyGen,\mathsf{OWSG}.\StateGen)$ be a SV-SI-OWSG.
Without loss of generality, we can assume that
$\mathsf{OWSG}.\KeyGen$ is the following algorithm:
first apply a QPT unitary $U$ on $|0...0\rangle$ to generate 
\begin{eqnarray*}
U|0...0\rangle=\sum_{k}\sqrt{\Pr[k\gets\mathsf{OWSG}.\KeyGen(1^\secp)]}|k\rangle|\mu_k\rangle, 
\end{eqnarray*}
and trace out the second register,
where $\{|\mu_k\rangle\}_k$ are some normalized states.
Moreover, without loss of generality, we can also assume that $\mathsf{OWSG}.\StateGen$ is the following algorithm:
first apply a QPT unitary $V_k$ that depends on $k$ on $|0...0\rangle$ to generate 
$V_k|0...0\rangle=|\psi_k\rangle_{\regA,\regB}$, and trace out the register $\regA$.

From the SV-SI-OWSG, we want to construct an EFI pair.
For that goal, we construct a statistically-hiding and computationally-binding canonical quantum bit commitment scheme from
SV-SI-OWSG. Due to \cref{thm:convertingflavors} (the equivalence between different flavors of commitments), we then have
a statistically-binding and computationally-hiding canonical quantum bit commitment scheme, which is equivalent to an EFI pair.
From the SV-SI-OWSG, we construct
a statistically-hiding and computationally-binding canonical quantum bit commitment scheme $\{Q_0,Q_1\}$ as follows.
\begin{eqnarray*}
Q_0|0\rangle_{\regC,\regR}&\coloneqq&\sum_{k}\sqrt{\Pr[k]} 
(|k\rangle|\mu_k\rangle)_{\regC_1}
|\psi_k\rangle_{\regC_2,\regR_2}^{\otimes t}|0\rangle_{\regR_3},\\
Q_1|0\rangle_{\regC,\regR}&\coloneqq& \sum_{k}\sqrt{\Pr[k]}
(|k\rangle|\mu_k\rangle)_{\regC_1}
|\psi_k\rangle_{\regC_2,\regR_2}^{\otimes t}|k\rangle_{\regR_3},
\end{eqnarray*}
where $\Pr[k]\coloneqq\Pr[k\gets\mathsf{OWSG}.\KeyGen(1^\secp)]$, $\regC_2$ is the combination of all ``$\regA$ registers'' of $|\psi_k\rangle$,
$\regR_2$ is the combination of all ``$\regB$ registers'' of $|\psi_k\rangle$,
$\regC\coloneqq(\regC_1,\regC_2)$ and $\regR\coloneqq(\regR_2,\regR_3)$.
Moreover, $t$ is a polynomial specified later.
It is clear that such $\{Q_0,Q_1\}$ is implemented in QPT in a natural way.

Let us first show the computational binding of $\{Q_0,Q_1\}$.
Assume that it is not computationally binding. Then, there exists a QPT unitary $U$, an ancilla state $|\tau\rangle$,
and a polynomial $p$ such that
\begin{eqnarray*}
\|
(\langle 0|Q_1^\dagger)_{\regC,\regR} U_{\regR,\regZ}(Q_0|0\rangle_{\regC,\regR}\otimes|\tau\rangle_\regZ)
\|
\ge\frac{1}{p}.
\end{eqnarray*}
Then,
\begin{eqnarray}
\frac{1}{p^2}
&\le&
\|
(\langle 0|Q_1^\dagger)_{\regC,\regR} U_{\regR,\regZ}(Q_0|0\rangle_{\regC,\regR}\otimes|\tau\rangle_\regZ)
\|^2\nonumber\\
&=&
\Big\|
\Big(\sum_{k'}\sqrt{\Pr[k']}\langle k',\mu_{k'}|_{\regC_1}\langle\psi_{k'}|^{\otimes t}_{\regC_2,\regR_2}\langle k'|_{\regR_3}\Big)\nonumber\\
&&\times\Big(\sum_{k}\sqrt{\Pr[k]}|k,\mu_k\rangle_{\regC_1}U_{\regR,\regZ}|\psi_k\rangle^{\otimes t}_{\regC_2,\regR_2}|0\rangle_{\regR_3}|\tau\rangle_\regZ\Big)
\Big\|^2\nonumber\\
&=&
\Big\|
\sum_{k}\Pr[k]
\langle\psi_{k}|^{\otimes t}_{\regC_2,\regR_2}\langle k|_{\regR_3}
U_{\regR,\regZ}|\psi_k\rangle^{\otimes t}_{\regC_2,\regR_2}|0\rangle_{\regR_3}|\tau\rangle_\regZ
\Big\|^2\nonumber\\
&\le&
\Big(
\sum_{k}\Pr[k]
\Big\|\langle\psi_{k}|^{\otimes t}_{\regC_2,\regR_2}\langle k|_{\regR_3}
U_{\regR,\regZ}|\psi_k\rangle^{\otimes t}_{\regC_2,\regR_2}|0\rangle_{\regR_3}|\tau\rangle_\regZ
\Big\|\Big)^2\nonumber\\
&\le&
\sum_{k}\Pr[k]
\Big\|
\langle\psi_{k}|^{\otimes t}_{\regC_2,\regR_2}\langle k|_{\regR_3}
U_{\regR,\regZ}|\psi_k\rangle^{\otimes t}_{\regC_2,\regR_2}|0\rangle_{\regR_3}|\tau\rangle_\regZ
\Big\|^2\nonumber\\
&\le&
\sum_{k}\Pr[k]
\Big\|
\langle k|_{\regR_3}
U_{\regR,\regZ}|\psi_k\rangle^{\otimes t}_{\regC_2,\regR_2}|0\rangle_{\regR_3}|\tau\rangle_\regZ
\Big\|^2.
\label{POWSG_B}
\end{eqnarray}
In the third inequality, we have used Jensen's inequality.\footnote{For a real convex function $f$, $f(\sum_i p_i x_i)\le\sum_i p_i f(x_i)$.}
From this $U$, we construct a QPT adversary $\cB$ that breaks the computational non-invertibility of the SV-SI-OWSG as follows:
On input the $\regR_2$ register of $|\psi_k\rangle^{\otimes t}_{\regC_2,\regR_2}$, apply $U_{\regR,\regZ}$ on
$|\psi_k\rangle_{\regC_2,\regR_2}^{\otimes t}|0\rangle_{\regR_3}|\tau\rangle_\regZ$, and measure
the $\regR_3$ register in the computational basis. Output the result.
Then, the probability that $\cB$
correctly outputs $k$ is equal to Eq.~(\ref{POWSG_B}).
Therefore, $\cB$ breaks the computational non-invertibility of the SV-SI-OWSG.

Let us next show the statistical hiding of $\{Q_0,Q_1\}$.
In the following, we construct a (not necessarily QPT) unitary $W_{\regR,\regZ}$ such that
\begin{eqnarray}
\|W_{\regR,\regZ}Q_0|0\rangle_{\regC,\regR}|0\rangle_\regZ
-Q_1|0\rangle_{\regC,\regR}|0\rangle_\regZ \|_1
&\le&\negl(\secp).
\label{SVOWSG_V}
\end{eqnarray}
Then, we have
\begin{eqnarray*}
&&\|\mbox{Tr}_{\regR}(Q_0|0\rangle_{\regC,\regR})
-\mbox{Tr}_{\regR}(Q_1|0\rangle_{\regC,\regR})\|_1\\
&=&
\|\mbox{Tr}_{\regR,\regZ}(Q_0|0\rangle_{\regC,\regR}|0\rangle_\regZ)
-\mbox{Tr}_{\regR,\regZ}(Q_1|0\rangle_{\regC,\regR}|0\rangle_\regZ)\|_1\\
&=&
\|\mbox{Tr}_{\regR,\regZ}(W_{\regR,\regZ}Q_0|0\rangle_{\regC,\regR}|0\rangle_\regZ)
-\mbox{Tr}_{\regR,\regZ}(Q_1|0\rangle_{\regC,\regR}|0\rangle_\regZ)\|_1\\
&\le&
\|W_{\regR,\regZ}Q_0|0\rangle_{\regC,\regR}|0\rangle_\regZ
-Q_1|0\rangle_{\regC,\regR}|0\rangle_\regZ\|_1\\
&\le&
\negl(\secp),
\end{eqnarray*}
which shows the statistical hiding of $\{Q_0,Q_1\}$.

Now we explain how to construct $W_{\regR,\regZ}$.
From \cref{lemma:SVOWSG_equivalence}, there exists a (not necessarily QPT) POVM measurement 
$\{\Pi_k\}_k$ and a polynomial $t$ such that
$\mbox{Tr}(\Pi_k \phi_k^{\otimes t})\ge1-\negl(\secp)$ and
$\mbox{Tr}(\Pi_{k'} \phi_k^{\otimes t})\le\negl(\secp)$ for all $k$ and $k'$ ($k\neq k'$).
Let $U_{\regR_2,\regZ}$ be a unitary operator 
that implements the POVM measurement $\{\Pi_k\}_k$ in the following way
\begin{eqnarray*}
U_{\regR_2,\regZ}|\psi_k\rangle^{\otimes t}_{\regC_2,\regR_2}|0...0\rangle_\regZ
&=&\sqrt{1-\epsilon_k}|k\rangle|junk_k\rangle+\sum_{k':k'\neq k} \sqrt{\epsilon_{k'}}|k'\rangle|junk_{k'}\rangle,
\end{eqnarray*}
where $\regZ$ is the ancilla register, $\{\epsilon_i\}_i$ are real numbers such that $1-\epsilon_k\ge1-\negl(\secp)$
and $\epsilon_{k'}\le\negl(\secp)$ for all $k'\neq k$,
and $\{|junk_i\rangle\}_i$ are ``junk'' states that are normalized.
Measuring the first register of the state realizes the POVM.
Let $V_{\regR,\regZ}$ be the following unitary:\footnote{For simplicity, we define $V_{\regR,\regZ}$ by explaining how it acts on
$|\psi_k\rangle^{\otimes t}_{\regC_2,\regR_2}|0...0\rangle_\regZ|0\rangle_{\regR_3}$, but it is clear from the explanation how $V_{\regR,\regZ}$ is defined.}
\begin{enumerate}
    \item 
    Apply $U_{\regR_2,\regZ}$ on $|\psi_k\rangle^{\otimes t}_{\regC_2,\regR_2}|0...0\rangle_\regZ|0\rangle_{\regR_3}$:
    \begin{eqnarray*}
U_{\regR_2,\regZ}|\psi_k\rangle^{\otimes t}_{\regC_2,\regR_2}|0...0\rangle_\regZ|0\rangle_{\regR_3}
&=&\Big[\sqrt{1-\epsilon_k}|k\rangle|junk_k\rangle\\
&&+\sum_{k':k'\neq k} \sqrt{\epsilon_{k'}}|k'\rangle|junk_{k'}\rangle\Big]|0\rangle_{\regR_3}.
    \end{eqnarray*}
    \item
    Copy the content of the first register to the register $\regR_3$:
     \begin{eqnarray*}
\sqrt{1-\epsilon_k}|k\rangle|junk_k\rangle|k\rangle_{\regR_3}
+\sum_{k':k'\neq k} \sqrt{\epsilon_{k'}}|k'\rangle|junk_{k'}\rangle|k'\rangle_{\regR_3}.
    \end{eqnarray*}
\end{enumerate}
Define $W_{\regR,\regZ}\coloneqq U_{\regR_2,\regZ}^\dagger V_{\regR,\regZ}$.

Let us show that thus constructed $W_{\regR,\regZ}$ satisfies Eq.~(\ref{SVOWSG_V}).
\begin{eqnarray*}
&&\Big((\langle 0|Q_1^\dagger)_{\regC,\regR}\langle 0|_{\regZ}\Big)
\Big(W_{\regR,\regZ}Q_0|0\rangle_{\regC,\regR}|0\rangle_{\regZ}\Big)\\
&=&\Big((\langle 0|Q_1^\dagger)_{\regC,\regR}\langle 0|_{\regZ}\Big)
\Big(U_{\regR_2,\regZ}^\dagger V_{\regR,\regZ}Q_0|0\rangle_{\regC,\regR}|0\rangle_{\regZ}\Big)\\
&=&\Big((\langle 0|Q_1^\dagger)_{\regC,\regR}\langle 0|_{\regZ}
U_{\regR_2,\regZ}^\dagger 
\Big)
\Big(
V_{\regR,\regZ}Q_0|0\rangle_{\regC,\regR}|0\rangle_{\regZ}\Big)\\
&=&
\Big(\sum_k\sqrt{\Pr[k]}(\langle k|\langle \mu_k|)_{\regC_1}
\Big[\sqrt{1-\epsilon_k}\langle k|\langle junk_k|\langle k|_{\regR_3}
+\sum_{k'\neq k}\sqrt{\epsilon_{k'}}\langle k'|\langle junk_{k'}|\langle k|_{\regR_3}\Big]\Big)\\
&&\times\Big(\sum_k\sqrt{\Pr[k]}(| k\rangle|\mu_k\rangle)_{\regC_1}
\Big[\sqrt{1-\epsilon_k}| k\rangle | junk_k\rangle| k\rangle_{\regR_3}
+\sum_{k'\neq k}\sqrt{\epsilon_{k'}}| k'\rangle| junk_{k'}\rangle| k'\rangle_{\regR_3}\Big]\Big)\\
&=&
\sum_k\Pr[k]
(1-\epsilon_k)\\
&\ge&1-\negl(\secp).
\end{eqnarray*}

\end{proof}

\ifnum\anonymous=1
\else
{\bf Acknowledgements.}
TM is supported by
JST CREST JPMJCR23I3,
JST Moonshot R\verb|&|D JPMJMS2061-5-1-1, 
JST FOREST, 
MEXT QLEAP, 
the Grant-in-Aid for Scientific Research (B) No.JP19H04066, 
the Grant-in Aid for Transformative Research Areas (A) 21H05183,
and 
the Grant-in-Aid for Scientific Research (A) No.22H00522.
\fi

\ifnum\submission=0
\bibliographystyle{alpha} 
\else
\bibliographystyle{splncs04}
\fi
\bibliography{abbrev3,crypto,reference}

\ifnum\submission=1
\else
\appendix
\section{PRSGs}
\label{sec:PRSs}
PRSGs are defined as follows.

\begin{definition}[Pseudorandom quantum states generators (PRSGs)~\cite{C:JiLiuSon18}]
\label{definition:PRSGs}
A pseudorandom quantum states generator (PRSG) is a set of algorithms $(\KeyGen,\StateGen)$ as follows.
\begin{itemize}
    \item $\KeyGen(1^\secp)\to k:$
    It is a QPT algorithm that, on input the security parameter $\secp$, outputs a classical key $k$.
    \item
    $\StateGen(k)\to|\phi_k\rangle:$
It is a QPT algorithm that, on input $k$,
outputs an $m$-qubit quantum state $|\phi_k\rangle$.
\end{itemize}
We require the following security:
For any polynomial $t$ and any QPT adversary $\cA$,  
\begin{eqnarray*}
|\Pr_{k\gets \KeyGen(1^\secp)}[1\gets\cA(|\phi_k\rangle^{\otimes t})]
-\Pr_{|\psi\rangle\leftarrow \mu_m}[1\gets\cA(|\psi\rangle^{\otimes t})]|
\le \negl(\secp),
\end{eqnarray*}
where $\mu_m$ is the Haar measure on $m$-qubit states.
\end{definition}

\begin{remark}
In the original definition of PRSGs~\cite{C:JiLiuSon18}, the classical key $k$ is uniformly sampled at random.
We use a more general definition where $k$ is sampled by a QPT algorithm, which was introduced in \cite{TCC:BraShm19}.
\end{remark}

\begin{remark}
Note that in the above definition, the outputs of $\StateGen$ are assumed to be pure. 
It could be possible to consider mixed states, but anyway the states have to be negligibly close to pure states
(except for a negligible fraction of $k$).\footnote{
Consider the case $t=2$, and consider a QPT adversary $\cA$ that, given $\phi_k^{\otimes 2}$ or $|\psi\rangle^{\otimes 2}$, runs
the SWAP test on them. Then,
\begin{eqnarray*}
\Big|\frac{1}{2^\secp}\sum_k\Pr[1\gets\cA(\phi_k^{\otimes 2})]-\mathbb{E}_{|\psi\rangle\gets\mu_m}\Pr[1\gets\cA(|\psi\rangle^{\otimes 2})]\Big|
&=&\Big|\frac{1}{2^\secp}\sum_k\frac{1+\mbox{Tr}(\phi_k^2)}{2}-\frac{1+1}{2}\Big|\\
&=&\frac{1}{2}\Big|\frac{1}{2^\secp}\sum_k\mbox{Tr}(\phi_k^2)-1\Big|,
\end{eqnarray*}
which has to be negligible.
Therefore, 
$\mbox{Tr}(\phi_k^2)$ is negligibly close to 1 (except for a negligible fraction of $k$).
}
We therefore, for simplicity, assume that the outputs of $\StateGen$ are pure.
In that case, $\StateGen$ runs as follows: apply a QPT unitary $U$ on $|k\rangle|0...0\rangle$ to generate
$U(|k\rangle|0...0\rangle)=|\phi_k\rangle\otimes|\eta_k\rangle$, and output $|\phi_k\rangle$. Note that the existance of
the ``junk'' register $|\eta_k\rangle$ is essential, because otherwise it is not secure.\footnote{Consider a QPT adversary who applies $U^\dagger$ on each copy of the received states, and measures them in the computational
basis. If the same $k$ is obtained many times, the adversary concludes that
the states are copies of pseudorandom states.}
(The simplest example of the junk register $|\eta_k\rangle$ would be $|\eta_k\rangle=|k\rangle$,
i.e., $U(|k\rangle|0...0\rangle)=|k\rangle\otimes|\phi_k\rangle$.
In that case, it is often convenient to consider that
$\StateGen$ applies a QPT unitary $U_k$ that depends on $k$ on the state $|0...0\rangle$,
and outputs $|\phi_k\rangle\coloneqq U_k|0...0\rangle$.)
\end{remark}

\begin{remark}
PRSGs can be constructed from (quantum-secure) one-way functions~\cite{C:JiLiuSon18,TCC:BraShm19,C:BraShm20}.
\end{remark}

\begin{remark}
\cite{C:BraShm20} showed that PRSGs with $m=c\log\secp$ for some $0<c<1$ can be constructed
with the statistical security.
On the other hand, \cite{TCC:Luowen}, showed that PRSGs with $m\ge\log\secp$ require computational assumptions.
\cite{TCC:Luowen} also pointed out that the result of \cite{Kre21} can be refined to show that
the existence of PRSGs with $m=(1+\epsilon)\log\secp$ for all $\epsilon>0$ implies $\BQP\neq \bf{PP}$. 
\end{remark}

\section{Verification Algorithm for Special Case}
\label{sec:Verpurephi}

\begin{lemma}
Let $(\KeyGen,\StateGen,\Ver)$ be a OWSG that satisfies the following.
\begin{itemize}
\item 
All $\phi_k$ are pure.
\item
$\Pr[\top\gets\Ver(k,\phi_k)]\ge1-\negl(\secp)$
for all $k$.
\end{itemize}
Then, the following OWSG $(\KeyGen',\StateGen',\Ver')$ exists.
\begin{itemize}
    \item 
    $\KeyGen',\StateGen'$ are the same as $\KeyGen$ and $\StateGen$, respectively.
    \item
    $\Ver'(k',\phi_k)$ is the following algorithm.
    Measure $\phi_k$ with the basis $\{|\phi_{k'}\rangle\langle\phi_{k'}|,I-|\phi_{k'}\rangle\langle\phi_{k'}|\}$,
    and output $\top$ if the first result is obtained.
    Otherwise, output $\bot$.
\end{itemize}
\end{lemma}

\begin{proof}
The correctness of $(\KeyGen',\StateGen',\Ver')$ is clear.
Let us show the security.
Assume that it is not secure.
Then, there exists a QPT adversary $\cA$, a polynomial $t$, and a polynomial $p$, such that
\begin{eqnarray*}
\sum_{k,k'}\Pr[k\gets\KeyGen(1^\secp)]
\Pr[k'\gets\cA(|\phi_k\rangle^{\otimes t})]|\langle\phi_k|\phi_{k'}\rangle|^2\ge\frac{1}{p}.
\end{eqnarray*}
If we define
\begin{eqnarray*}
S\coloneqq\Big\{(k,k')~\Big|~|\langle\phi_k|\phi_{k'}\rangle|^2\ge\frac{1}{2p}\Big\},
\end{eqnarray*}
we have
\begin{eqnarray*}
\sum_{(k,k')\in S}\Pr[k\gets\KeyGen(1^\secp)]
\Pr[k'\gets\cA(|\phi_k\rangle^{\otimes t})]\ge\frac{1}{2p}.
\end{eqnarray*}
From the $\cA$, we construct a QPT adversary $\cB$ that breaks the security of the original OWSG as follows:
On input $|\phi_k\rangle^{\otimes t}$, run $k'\gets\cA(|\phi_k\rangle^{\otimes t})$, and output $k'$.
Then the probability that $\cB$ breaks the original OWSG is
\begin{eqnarray*}
&&\sum_{k,k'}\Pr[k\gets\KeyGen(1^\secp)]
\Pr[k'\gets\cA(|\phi_k\rangle^{\otimes t})]
\Pr[\top\gets\Ver(k',|\phi_k\rangle)]\\
&\ge&
\sum_{(k,k')\in S}\Pr[k\gets\KeyGen(1^\secp)]
\Pr[k'\gets\cA(|\phi_k\rangle^{\otimes t})]
\Pr[\top\gets\Ver(k',|\phi_k\rangle)]\\
&\ge&
\sum_{(k,k')\in S}\Pr[k\gets\KeyGen(1^\secp)]
\Pr[k'\gets\cA(|\phi_k\rangle^{\otimes t})]
\Big(\Pr[\top\gets\Ver(k',|\phi_{k'}\rangle)]-\sqrt{1-\frac{1}{2p}}\Big)\\
&\ge&
\frac{1}{2p}
\Big(1-\negl(\secp)-\Big(1-\frac{1}{4p}\Big)\Big),
\end{eqnarray*}
which is non-negligible. Therefore the $\cB$
breaks the security of the original OWSG.
Here, in the second inequality, we have used the fact that for any $(k,k')\in S$, $|\langle\phi_k|\phi_{k'}\rangle|^2\ge\frac{1}{2p}$,
which means that $\mbox{Tr}(\Pi |\phi_k\rangle\langle\phi_k|)-\mbox{Tr}(\Pi |\phi_{k'}\rangle\langle\phi_{k'}|)\le\sqrt{1-\frac{1}{2p}}$
for any POVM element $\Pi$.
In the last inequality, we have used Bernoulli's inequality, $(1+x)^r\le 1+rx$ for any $0\le r\le 1$ and $x\ge-1$.
\end{proof}
\section{OWSGs from PRSGs with Improved Parameters}
\label{sec:OWSGfromPRSG_improved}

The following result was pointed out by Luowen Qian.
Also, a concurrent paper~\cite{CGGHLP23} shows similar results.
\begin{theorem}\label{thm:PRSG_OWSG}
If PRSGs with $m\ge\log\kappa$ exist, then OWSGs exist.
\end{theorem}

\begin{proof}[Proof of \cref{thm:PRSG_OWSG}]
Let $\mathsf{PRSG}.\KeyGen(1^\secp)$ and $\mathsf{PRSG}.\StateGen(k)\to |\xi_k\rangle$ be a PRSG, where $|\xi_k\rangle$ is an $m$-qubit state.
From it, we construct a OWSG as follows.
\begin{itemize}
    \item 
    $\KeyGen(1^\secp)\to k:$ Run $k\gets\mathsf{PRSG}.\KeyGen(1^\secp)$.
    \item
    $\StateGen(k)\to|\phi_k\rangle:$ Run $|\xi_k\rangle\gets\mathsf{PRSG}.\StateGen(k)$ $r$ times,
    where $r$ is a polynomial specified later.
    Output $|\phi_k\rangle\coloneqq|\xi_k\rangle^{\otimes r}$.
    \item
    $\Ver(k',|\phi_k\rangle)\to\top/\bot:$
    Parse $|\phi_k\rangle=|\xi_k\rangle^{\otimes r}$.
    Measure $|\phi_k\rangle$ with the basis $\{|\phi_{k'}\rangle\langle\phi_{k'}|,I-|\phi_{k'}\rangle\langle\phi_{k'}|\}$.
    If the first result is obtained, output $\top$. Otherwise, output $\bot$.
\end{itemize}

Its correctness is clear. Let us show the security.
Assume that it is not secure. Then, there exists a QPT adversary $\cA$, a polynomial $t$, and a polynomial $p$, such that
\begin{eqnarray}
\sum_{k,k'}\Pr[k\gets\mathsf{PRSG}.\KeyGen(1^\secp)]\Pr[k'\gets\cA(|\phi_k\rangle^{\otimes t})]|\langle\phi_k|\phi_{k'}\rangle|^2\ge\frac{1}{p}.
\label{OWSGfromPRSG_A}
\end{eqnarray}
From this $\cA$, we construct a QPT adversary $\cB$ that breaks the security of the PRSG as follows.
Let $b\in\bit$ be the parameter of the following security game.
\begin{enumerate}
    \item 
    The challenger $\cC$ of the security game of the PRSG sends $\rho^{\otimes rt+r}$ to $\cB$.
    Here, $\rho$ is chosen as follows.
    If $b=0$,
    $\cC$ runs $k\gets\mathsf{PRSG}.\KeyGen(1^\secp)$,
    runs $|\xi_k\rangle\gets\mathsf{PRSG}.\StateGen(k)$, and set $\rho=|\xi_k\rangle$.
    If $b=1$,
    $\rho$ is an $m$-qubit Haar random state $|\psi\rangle$.
    \item
    $\cB$ runs $k'\gets \cA(\rho^{\otimes rt})$. 
    $\cB$ measures $\rho^{\otimes r}$ with the basis $\{|\xi_{k'}\rangle\langle\xi_{k'}|^{\otimes r},I-|\xi_{k'}\rangle\langle\xi_{k'}|^{\otimes r}\}$.
    If the first result is obtained, $\cB$ outputs $b'=0$.
    Otherwise, it outputs $b'=1$.
\end{enumerate}
Then, $\Pr[b'=0~|~b=0]$ is equivalent to the left-hand-side of Eq.~(\ref{OWSGfromPRSG_A}), and therefore non-negligible.
On the other hand,
\begin{eqnarray*}
\Pr[b'=0~|~b=1]
&=&
\mathbb{E}_{|\psi\rangle\gets\mu_m}
\sum_{k'}
\Pr[k'\gets\cA(|\psi\rangle^{\otimes rt})]
|\langle\xi_{k'}|\psi\rangle|^{2r}\\
&\le&
\mathbb{E}_{|\psi\rangle\gets\mu_m}
\sum_{k'}
|\langle\xi_{k'}|\psi\rangle|^{2r}\\
&=&
\sum_{k'}
\langle\xi_{k'}|^{\otimes r}
\Big(\mathbb{E}_{|\psi\rangle\gets\mu_m}
|\psi\rangle\langle\psi|^{\otimes r}
\Big)
|\xi_{k'}\rangle^{\otimes r}\\
&=&
\sum_{k'}
\langle\xi_{k'}|^{\otimes r}
\frac{\Pi_{sym}^{2^m,r}}
{\mbox{Tr}(\Pi_{sym}^{2^m,r})}
|\xi_{k'}\rangle^{\otimes r}\\
&\le&
\sum_{k'}
\frac{1}
{\mbox{Tr}(\Pi_{sym}^{2^m,r})}\\
&\le&
2^\kappa 2^{-rm}r!\\
&\le&
\negl(\secp),
\end{eqnarray*}
where $\mathbb{E}_{|\psi\rangle\gets\mu_m}$ is the expectation value over
$m$-qubit Haar random,
$\Pi_{sym}^{2^m,r}$ is the projector onto the symmetric subspace of $(\mathbb{C}^{2^m})^{\otimes r}$.
In the third inequality, we have used the fact that
$\mbox{Tr}(\Pi_{sym}^{N,t})\ge\frac{N^t}{t!}$.
In the last inequality, we have taken $m\ge\log\kappa$ and $r=\kappa$,
and used
Stirling's approximation.
The $\cB$ therefore breaks the security of the PRSG.
\end{proof} 
\section{Impossibility of Statistically-Secure QDSs}
\label{app:stat_QDSs}
In this appendix, we show the following:
\begin{theorem}
Statistically-secure QDSs do not exist.    
\end{theorem}

\begin{proof}
Let us consider the following unbounded adversary $\cA$:    
\begin{enumerate}
    \item 
    Given $\pk^{\otimes t}$ with a certain polynomial $t$, run the shadow tomography algorithm to find $\sigma$
    such that $\Pr[\top\gets\Ver(\pk,m,\sigma)]\ge 1-\frac{1}{p(\secp)}$,
    where $m$ is any fixed bit string, and $p$ is any fixed polynomial.
    If there is no such $\sigma$, choose a bit string $\sigma$ uniformly at random.
    \item 
    Output $m$ and $\sigma$.
\end{enumerate}
We show that the probability that such $\cA$ win the security game is non-negligible.
First, define the set 
\begin{align}
G\coloneqq\left\{\sk:\sum_\sigma\Pr[\sigma\gets\Sign(\sk,m)]\Pr[\top\gets\Ver(\pk,m,\sigma)]\ge1-\frac{1}{p}\right\}.    
\label{setG}
\end{align}
Then, from the correctness,
\begin{align}
1-\negl(\secp)&\le \sum_\sk \Pr[\sk\gets\SKGen(1^\secp)]\sum_\sigma\Pr[\sigma\gets\Sign(\sk,m)]\Pr[\top\gets\Ver(\pk,m,\sigma)]\\    
&=
\sum_{\sk\in G} \Pr[\sk\gets\SKGen(1^\secp)]\sum_\sigma\Pr[\sigma\gets\Sign(\sk,m)]\Pr[\top\gets\Ver(\pk,m,\sigma)]\\    
&+\sum_{\sk\notin G} \Pr[\sk\gets\SKGen(1^\secp)]\sum_\sigma\Pr[\sigma\gets\Sign(\sk,m)]\Pr[\top\gets\Ver(\pk,m,\sigma)]\\    
&\le
\sum_{\sk\in G} \Pr[\sk\gets\SKGen(1^\secp)]
+1-\frac{1}{p},
\end{align}
which means
\begin{align}
\sum_{\sk\in G} \Pr[\sk\gets\SKGen(1^\secp)]
\ge\frac{1}{p}-\negl(\secp).
\label{fraction}
\end{align}
For each $\sk\in G$,
there exists $\sigma_\sk$ such that
$\Pr[\top\gets\Ver(\pk,m,\sigma_\sk)]\ge1-\frac{1}{p}$, because otherwise it contradicts
\cref{setG}.
Therefore, for each $\sk\in G$, $\cA$ can find $\sigma$ such that
$\Pr[\top\gets\Ver(\pk,m,\sigma)]\ge1-\frac{1}{p}$, 
and because of \cref{fraction}, the probability that $\cA$ wins is non-negligible.

\end{proof}
\section{Quantum SKE and Quantum PKE}\label{sec:QSKE}
In this section, we define quantum SKE (with classical keys and quantum ciphertexts) and its IND-CPA security. 
Then we show that IND-CPA secure quantum SKE implies QPOTP. 
By \cref{thm:QPOTP_OWSG,thm:QPOTP_EFI}, this means that it implies both OWSGs and EFI pairs. 
We also observe that IND-CPA secure quantum SKE as defined here is equivalent to quantum PKE defined in \cite{JC:KKNY12}. 
\subsection{Quantum SKE}
We define quantum SKE as SKE with classical keys and quantum ciphertexts. 
\begin{definition}[Quantum SKE]
A quantum SKE scheme is a tuple of algorithms $(\KeyGen,\Enc,\Dec)$ such that
\begin{itemize}
    \item 
    $\KeyGen(1^\secp)\to \sk:$
    It is a QPT algorithm that, on input the security parameter $\secp$, outputs
    a classical secret key $\sk$.
    \item
    $\Enc(\sk,x)\to\ct:$
    It is a QPT algorithm that, on input a key $\sk$ and a message $x\in\bit^\ell$, outputs
    a quantum ciphertext $\ct$.
    \item
    $\Dec(\sk,\ct)\to x':$
    It is a QPT algorithm that, on input $\sk$ and $\ct$, outputs $x'\in\bit^\ell$.
\end{itemize}
\paragraph{\bf Correctness:}
For any $x\in\bit^\ell$,
\begin{eqnarray*}
\Pr[x\gets\Dec(\sk,\ct):\sk\gets\KeyGen(1^\secp),\ct\gets\Enc(\sk,x)]
\ge1-\negl(\secp).
\end{eqnarray*}
\paragraph{\bf IND-CPA Security:}
For any QPT adversary $\cA=(\cA_1,\cA_2)$, 
\begin{eqnarray*}
&&|\Pr[1\gets \cA_2^{\Enc(\sk,\cdot)}(\mathsf{st},\ct^*):\sk\gets\KeyGen(1^\secp),(x_0,x_1,\mathsf{st})\gets \cA_1^{\Enc(\sk,\cdot)}(1^\secp),\ct^*\gets\Enc(\sk,x_0)]\\
&&-\Pr[1\gets \cA_2^{\Enc(\sk,\cdot)}(\mathsf{st},\ct^*):\sk\gets\KeyGen(1^\secp),(x_0,x_1,\mathsf{st})\gets \cA_1^{\Enc(\sk,\cdot)}(1^\secp),\ct^*\gets\Enc(\sk,x_1)]|\\
&&\le\negl(\secp)
\end{eqnarray*}
where $\Enc(\sk,\cdot)$ is a classically-accessible oracle that takes $x$ as input and returns  $\ct\gets \Enc(\sk,x)$. 
\end{definition}
\begin{remark}
In the above definition, we only give a single-copy of $\ct^*$ to $\cA_2$. However, by a standard hybrid argument, we can show that the above security implies the security against $\cA_2$ that obtains $t$ copies of $\ct^*$ for any polynomial $t$.
\end{remark}
\begin{theorem}
If there exists IND-CPA secure QSKE, there exists QPOTP.
\end{theorem}
\begin{proof}[Proof (sketch)]
By a standard hybrid argument, IND-CPA security implies the multi-instance security, where the adversary is given unbounded-polynomially many ciphertexts. Thus, by taking the message length to be much larger than the secret key length and encrypting the massage by using the IND-CPA secure QSKE in a bit-by-bit manner, we can obtain QPOTP.   
\end{proof}

Combined with  \cref{thm:QPOTP_OWSG,thm:QPOTP_EFI} we obtain the following theorem. 
\begin{theorem}
If there exists IND-CPA secure QSKE, there exist OWSGs and EFI pairs.
\end{theorem}

\subsection{Quantum PKE}
We define quantum PKE (with quantum public keys and quantum ciphertexts) following \cite{JC:KKNY12}. 
\begin{definition}[Quantum PKE~\cite{JC:KKNY12}]
A quantum PKE scheme is a tuple of algorithms $(\SKGen,\PKGen,\Enc,\Dec)$ such that
\begin{itemize}
    \item 
    $\SKGen(1^\secp)\to \sk:$
    It is a QPT algorithm that, on input the security parameter $\secp$, outputs
    a classical secret key $\sk$.
    \item 
    $\PKGen(\sk)\to \pk:$
    It is a QPT algorithm that, on input a secret key $\sk$, outputs
    a quantum public key $\pk$.
    \item
    $\Enc(\pk,x)\to\ct:$
    It is a QPT algorithm that, on input a public key $\pk$ and a message $x\in\bit^\ell$, outputs
    a quantum ciphertext $\ct$.
    \item
    $\Dec(\sk,\ct)\to x':$
    It is a QPT algorithm that, on input $\sk$ and $\ct$, outputs $x'\in\bit^\ell$.
\end{itemize}
\paragraph{\bf Correctness:}
For any $x\in\bit^\ell$,
\begin{eqnarray*}
\Pr[x\gets\Dec(\sk,\ct):\sk\gets\SKGen(1^\secp),\pk\gets\PKGen(\sk),\ct\gets\Enc(\pk,x)]
\ge1-\negl(\secp).
\end{eqnarray*}
\paragraph{\bf IND-CPA Security:}
For any QPT adversary $\cA=(\cA_1,\cA_2)$, and any polynomial $t$,
\begin{eqnarray*}
&&\left|\Pr\left[1\gets \cA_2(\mathsf{st},\ct^*):
\begin{array}{l}
\sk\gets\SKGen(1^\secp),\\
\pk\gets \PKGen(\sk),\\ 
(x_0,x_1,\mathsf{st})\gets \cA_1(\pk^{\otimes t}),\\
\ct^*\gets\Enc(\pk,x_0)
\end{array}\right]\right.\\
&-&
\left.
\Pr\left[1\gets \cA_2(\mathsf{st},\ct^*):
\begin{array}{l}
\sk\gets\SKGen(1^\secp),\\
\pk\gets \PKGen(\sk),\\ 
(x_0,x_1,\mathsf{st})\gets \cA_1(\pk^{\otimes t}),\\
\ct^*\gets\Enc(\pk,x_1)
\end{array}\right]
\right|
\le\negl(\secp).
\end{eqnarray*}
\end{definition}

IND-CPA QPKE implies IND-CPA QSKE in a trivial manner similarly to the classical setting. 
On the other hand, we prove that the other direction also holds, which means that the existence of IND-CPA QSKE and QPKE are equivalent.
\begin{theorem}
There exists IND-CPA secure QPKE if and only if there exists IND-CPA secure QSKE.
\end{theorem}
\begin{proof}[Proof (sketch.)]
IND-CPA QPKE trivially implies IND-CPA QSKE by considering the combination of $\SKGen$ and $\PKGen$ of QPKE as $\KeyGen$ of QSKE. We prove the other direction. 
Let $(\QSKE.\KeyGen,\QSKE.\Enc,\QSKE.\Dec)$ be an IND-CPA secure QSKE with one-bit messages. We construct a QPKE scheme $(\QPKE.\PKGen,\QPKE.\SKGen,\QPKE.\Enc,\QPKE.\Dec)$ with one-bit messages as follows. (Note that the message space of IND-CPA secure QPKE schemes can be extended to arbitrarily many bits by bit-wise encryption.)
\begin{itemize}
    \item 
    $\QPKE.\SKGen(1^\secp):$
    Run $\sk\gets \QSKE.\KeyGen(1^\secp)$ and outputs $\sk$.  
    \item 
    $\QPKE.\PKGen(\sk):$
    Run $\ct_0\gets \QSKE.\Enc(\sk,0)$ and $\ct_1\gets \QSKE.\Enc(\sk,1)$ and outputs $\pk\seteq (\ct_0,\ct_1)$. 
    \item
    $\QPKE.\Enc(\pk,x):$
    On input $\pk=(\ct_0,\ct_1)$ and $x\in \bit$, output $\ct=\ct_x$.  
    \item
    $\QPKE.\Dec(\sk,\ct):$
    Output $\QSKE.\Dec(\sk,\ct)$. 
\end{itemize}
The correctness of $\QPKE$ immediately follows from that of $\QSKE$. 
Noting that we can simulate arbitrarily many copies of $\pk= (\ct_0,\ct_1)$ by querying $0$ and $1$ to the encryption oracle $\SKE.\Enc(\sk,\cdot)$ many times, there is a straightforward reduction from  the IND-CPA security of  $\QPKE$ to that of $\QSKE$. 
\end{proof}
\begin{remark}
One may think that this theorem is surprisingly strong because it is unlikely that SKE implies PKE in the classical setting~\cite{C:ImpRud88}. However, we would like to point out that QPKE as defined in \cite{JC:KKNY12} is not as useful as classical PKE because public keys are quantum. Their security model assumes that public keys are delivered to senders without being forged. This itself is the same as classical PKE, but the crucial difference is that forgery of public keys can be prevented by additionally using digital signatures in the classical setting, but that is not possible for QPKE because we cannot sign on quantum messages~\cite{Alagic_2021}. In fact, the assumption that the adversary cannot forge public keys mean that it  also cannot eavesdrop it because eavesdropping may change the state of the public key. Thus, their security model essentially assumes a secure channel to send public keys to the sender. If there is such a secure channel, we could simply send a key for SKE through that channel. Thus, this theorem should not be understood as a new useful feasibility result on PKE. 
We remark that the construction does not work if we consider QPKE with classical public keys~\cite{EC:HhaMorYam23}. 
We conjecture that it is impossible to construct QPKE with classical public keys from QSKE under a certain class of black-box constructions. 
\end{remark}
\newcommand{\rsp}{\mathsf{rsp}}
\newcommand{\Ext}{\mathsf{Ext}}
\newcommand{\Wrong}{\mathsf{Wrong}}

\section{Proof of \cref{thm:amplify_qpuzzle}}\label{sec:proof_amplification}
In this section, we prove \cref{thm:amplify_qpuzzle}. We remark that the following proof is almost identical to that of \cite[Lemma~1]{TCC:CanHalSte05} except for the modification explained in the proof sketch in \cref{sec:amplification}.  

\protocol
{The Adversary $\cA'(\puz^{\otimes t'})$}
{The Adversary $\cA'$ for $(\CheckGen,\PuzzleGen,\Ver)$}
{fig:Aprime}
{~\\
\textbf{Preprocessing Phase}
\begin{enumerate}
    \item Initialize $\prefix$ to be an empty vector.
    \item For $i=1$ to $n-1$, do the following:
    \begin{enumerate}
        \item Run $\che^*\gets \extend(\prefix,i)$.
        \item If $\che^*=\bot$, set $v\gets i$ and go to \textbf{Online Phase}.
        \item Update $\prefix\gets \prefix \circ \che^*$.
    \end{enumerate}
    \item Set $v\gets n$ and go to \textbf{Online Phase}.
\end{enumerate}
\textbf{Online Phase}
\begin{enumerate}
    \item If $v=n$, do the following:
    \begin{enumerate}
        \item Parse $\prefix=(\che_1,\ldots,\che_{n-1})$.
        \item Run $\puz_i^{\otimes t}\gets \PuzzleGen^{\otimes t}(\che_i)$ for $i\in [n-1]$. 
        \item Run $(\ans_1,\ldots ,\ans_n)\gets \cA(\puz_1^{\otimes t},\ldots, \puz_{n-1}^{\otimes t},\puz^{\otimes t})$.
        \item Output $\ans_n$.
    \end{enumerate}
    \item Otherwise, repeat the following $L=\lceil \frac{6q\ln(6q)}{\delta^{n-v+1}}\rceil$ times.
    \begin{enumerate}
    \item Parse $\prefix=(\che_1,\ldots,\che_{v-1})$.
     \item Run $\puz_i^{\otimes t}\gets \PuzzleGen^{\otimes t}(\che_i)$ for $i\in [v-1]$.
        \item For $i=v+1$ to $n$, do the following:
        \begin{enumerate}
        \item Run $\che_i \gets \CheckGen(1^\secp)$ 
        \item Run $\puz_i^{\otimes t}\gets \PuzzleGen^{\otimes t}(\che_i)$. 
    \end{enumerate}
    \item Run $(\ans_1,\ldots ,\ans_n)\gets \cA(\puz_1^{\otimes t},\ldots,\puz_{v-1}^{\otimes t}, \puz^{\otimes t},\puz_{v+1}^{\otimes t},\ldots, \puz_n^{\otimes t})$.
    \item Run $d_i \gets \Ver(\ans_i,\che_i)$ for $i\in \{v+1.\ldots,n\}$.
    \item Output $\ans_v$ if $d_i=\top$ for all $i\in \{v+1.\ldots,n\}$.
    \end{enumerate}
    If none of the above repetitions outputs $\ans_v$, then abort.
\end{enumerate}
\textbf{The Subroutine} $\extend(\prefix,i)$
\begin{enumerate}
    \item Repeat the following $N_i=\lceil \frac{6q}{\delta^{n-i+1}}\ln(\frac{18qn}{\delta})\rceil$ times:
    \begin{enumerate}
        \item Run $\che^*\gets \CheckGen(1^\secp)$.
        \item Run $\bar{\mu}_{\che^*}\gets \estimate(\prefix\circ \che^*, i)$
        \item If  $\bar{\mu}_{\che^*}\geq \delta^{n-i}$, output $\che^*$. 
    \end{enumerate}
    \item Output $\bot$.
\end{enumerate}
\textbf{The Subroutine} $\estimate(\prefix,i)$
\begin{enumerate}
    \item Parse $\prefix=(\che_1,\ldots,\che_{i})$.
    \item Initialize $\mathsf{count}\gets 0$.
    \item Repeat the following $M_i=\lceil \frac{84q^2}{\delta^{n-i}}\ln(\frac{18qnN_i}{\delta})\rceil$ times:
    \begin{enumerate}
        \item Run $\puz_j^{\otimes t}\gets \PuzzleGen^{\otimes t}(\che_j)$ for $j\in [i]$.
        \item For $j=i+1$ to $n$, do the following:
        \begin{enumerate}
        \item Run $\che_{j}\gets \CheckGen(1^\secp)$.
        \item Run $\puz_j^{\otimes t}\gets \PuzzleGen^{\otimes t}(\che_j)$.
        \end{enumerate}
        \item Run $(\ans_1,\ldots ,\ans_n)\gets \cA(\puz_1^{\otimes t},\ldots, \puz_n^{\otimes t})$.
        \item If $\Ver(\ans_j,\che_j)=\top$ for all $j\in \{i+1,\ldots,n\}$, increment $\mathsf{count}\gets \mathsf{count} +1$.
    \end{enumerate}
    \item Output $\frac{\mathsf{count}}{M_i}$.
\end{enumerate}
}

\begin{proof}[Proof of \cref{thm:amplify_qpuzzle}]
We prepare several definitions. We write $\puz^{\otimes t} \gets \PuzzleGen^{\otimes t}(\che)$ to mean that we run $\puz \gets \PuzzleGen(\che)$ $t$ times to generate $\puz^{\otimes t}$. 
For $\vec{\che}=(\che_1,\ldots, \che_i)$ and $\che$, we define $\vec{\che}\circ \che\seteq (\che_1,\ldots, \che_i,\che)$.
We set $t'\seteq \lceil \frac{6q\ln(6q)}{\delta^{n}}\rceil\cdot t$.  
The construction of $\cA'$ is given in \cref{fig:Aprime}. 
Note that $t'$ copies of $\puz$ are indeed sufficient for $\cA'$ since it feeds $t$ copies of $\puz$ to $\cA$ at most $L=\lceil \frac{6q\ln(6q)}{\delta^{n-v+1}}\rceil\leq \lceil \frac{6q\ln(6q)}{\delta^{n}}\rceil$ times.  
We analyze $\cA'$ below. 

For a sequence $\vec{\che}=(\che_1,\ldots,\che_i)$ of $i\le n-1$ check keys, we define 
\begin{align*}
    \rsp(\vec{\che})\seteq \Pr\left[
    \begin{array}{l}
    \forall j\in\{i+1,\ldots,n\}\\
    \Ver(\ans_j,\che_j)=\top
    \end{array}
    :
    \begin{array}{l}
    (\che_{i+1},\ldots,\che_{n}) \gets\CheckGen^{n-i}(1^\secp) \\
       \puz_j^{\otimes t} \gets\PuzzleGen^{\otimes t}(\che_j) \text{~for~}j\in[n]\\
      (\ans_1,\ldots,\ans_n) \gets \cA(\puz_1^{\otimes t},\ldots,\puz_n^{\otimes t})
    \end{array}
    \right].
\end{align*}
For a sequence $\vec{\che}=(\che_1,\ldots,\che_{i-1})$ of $i-1\le n-2$ check keys, we define a subset $\Ext_i(\vec{\che})$ of check keys as follows:
\begin{align*}
    \Ext_i(\vec{\che})\seteq \left\{\che: \rsp(\vec{\che}\circ \che)\ge \delta^{n-i}\left(1+\frac{1}{6q}\right) \right\}.
\end{align*}
Let $\prefix=(\che_1,\ldots,\che_{v-1})$ be the prefix at the end of preprocessing phase.  
Let $\prefix_i$ be the random variable defined to be $(\che_1,\ldots,\che_i)$ if $i\le v-1$ and otherwise $\bot$.  
For convenience, we assign to $\prefix_0$ a special symbol $\Lambda$ that is different from $\bot$.  
For $i\in [n-1]$, let $\Wrong_i$ be the event that $\prefix_{i-1}\neq \bot$ and one of the following holds:
\begin{enumerate}
    \item $\prefix_i\neq \bot$ and $\rsp(\prefix_i)<\delta^{n-i}\left(1-\frac{1}{6q}\right)$.
    \item $\prefix_i=\bot$ and $\Pr[\che \in \Ext_i(\prefix_{i-1}):\che\gets \CheckGen(1^\secp)]\ge \frac{\delta^{n-i+1}}{6q}$. 
\end{enumerate}
By the Chernoff bound, we can show that 
\begin{align}\label{eq:Wrong_prob}
    \Pr[\Wrong_i]\le \frac{\delta}{6qn}
\end{align}
for all $i\in [n-1]$. We omit the proof since it is exactly the same as the proof of \cite[Claim 2]{TCC:CanHalSte05}. 

We move onto the analysis of the online phase. We say that $\cA'$ succeeds if its answer passes the verification.\footnote{Note that the success or failure of $\cA'$ is not determined by the execution of $\cA'$ itself. It is determined after running the verification algorithm to verify the answer output by $\cA'$. This is because we consider non-deterministic verification algorithm unlike~\cite{TCC:CanHalSte05}. Thus, whenever we refer to the success of $\cA$',   we implicitly run the verification algorithm on its output.} 
We prove that $\cA'$ succeeds with probability at least $\delta\left(1-\frac{5}{6q}\right)$ unless $\Wrong_i$ occurs for some $i\in[n-1]$. If this is proven, it finishes the proof of \cref{thm:amplify_qpuzzle} 
since the probability that $\Wrong_i$ occurs for some $i\in[n-1]$ is at most $\frac{\delta}{6q}$ by the union bound and \cref{eq:Wrong_prob}.    

Suppose that the preprocessing phase generates prefix $\prefix_{v-1}=(\che_1,\ldots,\che_{v-1})$ of length $v-1$. 

When $v=n$, we have $\prefix_{i}\ne \bot$ for all $i\in[n-1]$. Thus, if we assume that $\Wrong_{n-1}$ does not occur, we have 
\begin{align*}
    \rsp(\prefix_{n-1})\ge \delta\left(1-\frac{1}{6q}\right).
\end{align*}
This directly means that the probability that $\cA$' succeeds is at least  $\delta\left(1-\frac{1}{6q}\right)\ge \delta\left(1-\frac{5}{6q}\right)$. 

In the following, we consider the case where $v\le n-1$.  
In this case, we have $\prefix_{i}\ne \bot$ for all $i\in[v-1]$ and $\prefix_{v}= \bot$. Then if we assume that neither of $\Wrong_{v-1}$ or $\Wrong_{v}$ occurs, we have 
\begin{align}\label{eq:rsp_lb}
    \rsp(\prefix_{v-1})\ge \delta^{n-v+1}\left(1-\frac{1}{6q}\right)
\end{align}
and
\begin{align}\label{eq:prob_E_ub}
    \Pr[\che\in E:\che\gets \CheckGen(1^\secp)] < \frac{\delta^{n-v+1}}{6q} 
\end{align}
where we define
    $E\seteq \Ext_v(\prefix_{v-1})$ for convenience.
In the following, we fix $\prefix_{v-1}=(\che_1,\ldots,\che_{v-1})$ that satisfies \cref{eq:prob_E_ub,eq:rsp_lb} and show that $\cA$' succeeds with probability at least  $\delta\left(1-\frac{5}{6q}\right)$ for any such fixed $\prefix$. As explained above, this suffices for completing the proof of \cref{thm:amplify_qpuzzle}. 

For any check key $\che$, we define 
\begin{align*}
    &w(\che)\seteq \Pr[\CheckGen(1^\secp)=\che],\\
    &s(\che)\seteq \Pr\left[
    \begin{array}{l}
    \forall j\in\{v,\ldots,n\}\\
    \Ver(\ans_j,\che_j)=\top
    \end{array}
    :
    \begin{array}{l}
    \che_v\seteq \che\\
    (\che_{v+1},\ldots,\che_{n}) \gets\CheckGen^{n-v}(1^\secp) \\
       \puz_j^{\otimes t} \gets\PuzzleGen^{\otimes t}(\che_j) \text{~for~}j\in[n]\\
      (\ans_1,\ldots,\ans_n) \gets \cA(\puz_1^{\otimes t},\ldots,\puz_n^{\otimes t})
    \end{array}
    \right]\\
    &a(\che)\seteq \Pr\left[
    \begin{array}{l}
    \forall j\in\{v+1,\ldots,n\}\\
    \Ver(\ans_j,\che_j)=\top
    \end{array}
    :
    \begin{array}{l}
    \che_v\seteq \che\\
    (\che_{v+1},\ldots,\che_{n}) \gets\CheckGen^{n-v}(1^\secp) \\
       \puz_j^{\otimes t} \gets\PuzzleGen^{\otimes t}(\che_j) \text{~for~}j\in[n]\\
      (\ans_1,\ldots,\ans_n) \gets \cA(\puz_1^{\otimes t},\ldots,\puz_n^{\otimes t})
    \end{array}
    \right].
\end{align*}
Note that  $\prefix_{v-1}=(\che_1,\ldots,\che_{v-1})$ is fixed and hardwired in the above definitions. 
By the definitions of $w(\che)$, $s(\che)$, and $\rsp(\prefix_{v-1})$ and \cref{eq:rsp_lb}, we have
\begin{align}\label{eq:lb_sum_ws}
    \sum_{\che}w(\che)s(\che)=\rsp(\prefix_{v-1})\ge \delta^{n-v+1}\left(1-\frac{1}{6q}\right).
\end{align}
By the definition of $E$ and $a(\che)$, for any $\che\notin E$, we have
\begin{align}\label{eq:lb_a}
    a(\che)\le \delta^{n-v}\left(1+\frac{1}{6q}\right).
\end{align}
By the definitions of $s(\che)$ and $a(\che)$, we have
\begin{align}\label{eq:prob_conditional}
\Pr[\Ver(\cA'(\puz^{\otimes t'}),\che)=\top|\cA'(\puz^{\otimes t'})\text{~does~not~abort}]=\frac{s(\che)}{a(\che)}.
\end{align}
We let 
\begin{align}\label{eq:def_B}
    B\seteq \left\{\che:s(\che)<\frac{\delta^{n-v+1}}{6q}\right\}. 
\end{align}
The for any $\che\notin B$, we have
\begin{align}\label{eq:ub_abort}
\Pr[\cA'(\puz^{\otimes t'})\text{~aborts}:\puz^{\otimes t'}\gets \PuzzleGen^{\otimes t'}(\che)]<\left(1-\frac{\delta^{n-v+1}}{6q}\right)^{\lceil\frac{6q\ln(6q)}{\delta^{n-v+1}}\rceil}<\frac{1}{6q}.
\end{align}
Next, we have 
\begin{align*}
    \sum_{\che\in B\cup E}w(\che)s(\che)&\le \sum_{\che\in B}w(\che)s(\che)+ \sum_{\che\in E}w(\che)s(\che)\\
    &\le \sum_{\che\in B}w(\che)\frac{\delta^{n-v+1}}{6q}+\sum_{\che\in E}w(\che)\\
    &\le \frac{\delta^{n-v+1}}{6q}+\frac{\delta^{n-v+1}}{6q}=\frac{\delta^{n-v+1}}{3q}
\end{align*}
where the second inequality follows from \cref{eq:def_B} and the third inequality follows from \cref{eq:prob_E_ub}. Combined with \cref{eq:lb_sum_ws},  we have
\begin{align}\label{eq:lb_sum_ws_two}
   \sum_{\che\notin B\cup E}w(\che)s(\che) \ge \delta^{n-v+1}\left(1-\frac{1}{2q}\right).
\end{align}
Then we have the following where $\puz\gets \PuzzleGen(\che)$ whenever $\puz$ appears: 
\begin{align*}
    &\Pr_{\che}[\A'(\puz^{\otimes t'})\text{~succeeds}]\\
    &=\sum_{\che}w(\che)\Pr[\Ver(\cA'(\puz^{\otimes t'}),\che)=\top]\\
    &\ge\sum_{\che \notin B\cup E}w(\che)\Pr[\cA'(\puz^{\otimes t'})\text{~does~not~abort}]\Pr[\Ver(\cA'(\puz^{\otimes t'}),\che)=\top|\cA'(\puz^{\otimes t'})\text{~does~not~abort}]\\
    &\ge \sum_{\che \notin B\cup E}w(\che)\left(1-\frac{1}{6q}\right)\frac{s(\che)}{a(\che)}\\
    &\ge \sum_{\che \notin B\cup E}w(\che)\left(1-\frac{1}{6q}\right)\frac{s(\che)}{\delta^{n-v}\left(1+\frac{1}{6q}\right)}\\
    &=\frac{1-\frac{1}{6q}}{\delta^{n-v}\left(1+\frac{1}{6q}\right)}\sum_{\che \notin B\cup E}w(\che)s(\che)\\
    &\ge \frac{1-\frac{1}{6q}}{\delta^{n-v}\left(1+\frac{1}{6q}\right)}\delta^{n-v+1}\left(1-\frac{1}{2q}\right)\\
    &>\delta\left(1-\frac{5}{6q}\right)
\end{align*}
where the second inequality follows from \cref{eq:prob_conditional,eq:ub_abort}, 
the third inequality follows from \cref{eq:lb_a}, 
and the fourth inequality follows from \cref{eq:lb_sum_ws_two}. As already explained, the above lower bound on the success probability suffices for completing the proof of \cref{thm:amplify_qpuzzle}.
\end{proof}

\fi

\end{document}